\def\notes{0}

\def\artstyle{1}

\ifnum\artstyle=1
    \documentclass{article}
\else
    \documentclass[11pt]{amsart}
    \usepackage[foot]{amsaddr}
\fi

\newcommand{\art}[1]{\ifnum\artstyle=1 #1 \fi}
\newcommand{\nart}[1]{\ifnum\artstyle=0 #1 \fi}

\usepackage{bbm}
\usepackage{caption}
\usepackage[utf8]{inputenc} 
\usepackage[T1]{fontenc}    
\usepackage{hyperref}       
\usepackage{url}            
\usepackage{booktabs}       
\usepackage{amsfonts}       
\usepackage{nicefrac}       
\usepackage{microtype}      
\usepackage{csquotes}
\usepackage{fullpage}
\usepackage{amsmath}
\usepackage{amsfonts}
\usepackage{graphicx} 
\usepackage[mincitenames=20, maxnames=20]{biblatex}
\usepackage{comment}
\usepackage[table,dvipsnames]{xcolor}
\usepackage{colortbl} 
\usepackage{multirow}
\usepackage{cancel}
\usepackage[normalem]{ulem}
\usepackage[inline,shortlabels]{enumitem}
\usepackage{caption}
\usepackage{subcaption}
\art{
\usepackage{amssymb}
\usepackage{amsthm}
}

\usepackage{algorithm}
\usepackage[noend]{algpseudocode}
\usepackage[capitalize,noabbrev]{cleveref}

\bibliography{main} 

\definecolor{mygreen}{rgb}{0.0, 0.49, 0.28}
\definecolor{mypurple}{rgb}{0.58,0.23,0.94}
\definecolor{mybrown}{rgb}{0.35,0.15,0.0}

\ifnum\notes=1
\newcommand{\pjnote}[1]{{\color{mypurple}\footnote{{\color{mypurple} {\bf PJ:} #1}}}}
\newcommand{\pj}[1]{{\color{mypurple} #1}}
\newcommand{\ssnote}[1]{{\color{red}\footnote{{\color{red} {\bf SS:} #1}}}}

\newcommand{\janote}[1]{{\color{blue}\footnote{{\color{blue} {\bf JA:} #1}}}}

\else
\newcommand{\pjnote}[1]{}
\newcommand{\pj}[1]{{#1}}
\newcommand{\ssnote}[1]{}

\newcommand{\janote}[1]{}

\renewcommand{\sout}[1]{}
\fi

\usepackage[colorinlistoftodos,prependcaption,textsize=tiny,disable]{todonotes}
\newcommand{\palak}[1]{\todo[linecolor=purple,backgroundcolor=purple!25,bordercolor=purple]{{\bf Palak:} #1}}
\newcommand{\satchit}[1]{\todo[linecolor=green,backgroundcolor=green!25,bordercolor=green]{{\bf Satchit:} #1}}

\newtheorem{theorem}{Theorem}[section]

\newtheorem{corollary}[theorem]{Corollary}
\newtheorem{claim}[theorem]{Claim}
\newtheorem{lemma}[theorem]{Lemma}

\newtheorem{definition}[theorem]{Definition}

\newcommand{\Z}{\mathbb{Z}}

\newcommand{\propcount}{\mathsf{PropCount}}
\newcommand{\countdistinct}{\mathsf{CountDistinct}}
\newcommand{\trianglecount}{\mathsf{TriangleCount}}

\newcommand{\degreecount}{\mathsf{DegreeCount}}
\newcommand{\diff}[1]{\mathit{d}_{\,#1}}
\newcommand{\countdiff}{\diff{\countdistinct}}
\newcommand{\trianglediff}{\diff{\trianglecount}}

\newcommand{\degreediff}{\diff{\degreecount}}

\newcommand{\univ}{\mathcal{U}}

\newcommand{\Flip}{\mathsf{Flip}}

\newcommand{\wvec}{\vec{\mathbf{w}}}

\newcommand{\zvec}{\vec{\mathbf{z}}}
\newcommand{\yvec}{\vec{\mathbf{y}}}
\newcommand{\xvec}{\vec{\mathbf{x}}}

\newcommand{\dvec}{\vec{\Delta}}
\newcommand{\bR}{\mathbb{R}}

\newcommand{\cA}{\mathcal{A}}
\newcommand{\cD}{\mathcal{D}}

\newcommand{\cS}{\mathcal{S}}
\newcommand{\cT}{\mathcal{T}}
\newcommand{\cB}{\mathcal{B}}
\newcommand{\cN}{\mathcal{N}}
\newcommand{\cM}{\mathcal{M}}
\newcommand{\cQ}{\mathcal{Q}}
\newcommand{\eps}{\varepsilon}
\newcommand{\sens}{\mathsf{sens}}
\newcommand{\maxse}{\mathsf{MaxSE}}
\newcommand{\meanse}{\mathsf{MeanSE}}

\renewcommand{\epsilon}{\varepsilon}
\newcommand{\boundedvecs}[1]{\cS_{#1}}

\makeatletter
\def\Cline#1#2{\@Cline#1#2\@nil}
\def\@Cline#1-#2#3\@nil{%
  \omit
  \@multicnt#1%
  \advance\@multispan\m@ne
  \ifnum\@multicnt=\@ne\@firstofone{&\omit}\fi
  \@multicnt#2%
  \advance\@multicnt-#1%
  \advance\@multispan\@ne
  \leaders\hrule\@height#3\hfill
  \cr}
\makeatother

\nart{
\title{Fine-Grained Error Bounds for Private Continual Cardinality Estimation in Fully Dynamic Streams}
\author{Joel Daniel Andersson, Palak Jain, Satchit Sivakumar}  
}

\art{
\title{Improved Accuracy for Private Continual Cardinality Estimation in Fully Dynamic Streams via Matrix Factorization}

\author{
    Joel Daniel Andersson\thanks{Institute of Science and Technology Austria.
    \texttt{joel.andersson@ist.ac.at}.
    }
    \and Palak Jain\thanks{Department of Computer Science, Boston University. 
    \texttt{\{palakj,satchit\}@bu.edu}.
    }
    \and Satchit Sivakumar\footnotemark[2]
}
\date{\today}
}

\begin{document}
\art{\maketitle}
\begin{abstract}
    We study differentially-private statistics in the \emph{fully dynamic continual observation} model, where many updates can arrive at each time step and updates to a stream can involve both insertions and deletions of an item. 
    Earlier work (e.g., Jain et al., NeurIPS 2023 for counting distinct elements; Raskhodnikova \& Steiner, PODS 2025 for triangle counting with edge updates) reduced the respective cardinality estimation problem to continual counting on the \emph{difference stream} associated with the true function values on the input stream. In such reductions, a change in the original stream can cause many changes in the difference stream, \pj{this poses} a challenge \pj{for}
    applying private continual counting algorithms to obtain optimal error bounds. We improve the accuracy of several such reductions by studying the associated $\ell_p$-sensitivity vectors of the resulting difference streams and isolating their properties.
    
    We demonstrate that our framework gives improved bounds for counting distinct elements, \pj{estimating degree histograms, and estimating triangle counts}
   (under a slightly relaxed privacy model), thus offering a general approach to \pj{private} continual cardinality estimation in streaming settings. Our improved accuracy stems from tight analysis of known factorization mechanisms for the counting matrix in this setting; 
    the key technical challenge is arguing that one can use state-of-the-art factorizations for sensitivity vector sets with the properties we isolate. In particular, we show that for approximate DP, the celebrated square-root factorization of the counting matrix (Henzinger et al., SODA 2023; Fichtenberger et al., ICML 2023) can be employed to give concrete improvements in accuracy over past work based on the binary tree mechanism.
    We also give a tight analysis of $b$-ary tree mechanisms with subtraction under these sensitivity patterns, including a polytime routine for computing the $\ell_p$ sensitivity, yielding time and space efficient algorithms for both pure and approximate DP.
    Empirically and analytically, we demonstrate that our improved error bounds offer a substantial improvement
    \pj{in accuracy} for cardinality estimation problems over a large range of parameters.
\end{abstract}


\maketitle
\newpage
\setcounter{tocdepth}{2}
\tableofcontents
\newpage



\section{Introduction}
Privacy is a central challenge for systems that process sensitive data, especially in settings where data evolves continuously and must be monitored or used over time. 
The continual observation model of differential privacy \cite{chan_private_2011, dwork_differential_2010} provides rigorous guarantees for individual privacy in exactly such settings. In this work, we revisit the problem of differentially private cardinality estimation under continual observation in the challenging \textit{fully dynamic} setting where updates can involve insertions or deletions. 

Our main contribution is to show that new analyses of factorization mechanisms 
can yield concrete accuracy improvements 
for several problems in the continual observation model of differential privacy. 
To enable this, we (1) define and analyze a new structured class of sensitivity vectors that arise in reductions from cardinality estimation problems like counting distinct elements ($\countdistinct$), releasing degree histograms ($\degreecount$), and releasing triangle counts ($\trianglecount$) to continual counting; 
(2) prove error bounds for factorization mechanisms based on Toeplitz factorizations for continual counting with respect to this new class of sensitivity vectors; and 
(3) show that even tree-aggregation-based factorization mechanisms admit improved error guarantees when adapted to sensitivity vectors with the structure we identify.

\paragraph{\bf Improved Mechanisms for Continual Counting}
The summation of bits, also called `continual counting', is a fundamental building block in various data analysis tasks. The binary tree mechanism \cite{dwork_differential_2010, chan_private_2011} for this problem has asymptotic error that is within logarithmic factors of optimal. However, while the binary tree mechanism is frequently used, the constants involved in its error guarantee are suboptimal \cite{fichtenberger_constant_2023, HenzingerUU23}, which can have a big impact in practical applications \cite{kairouz_practical_2021, denissov_improved_2022}. 


As such, a long line of work has focused on optimizing the concrete accuracy of differentially private algorithms for continual counting. All of these approaches, whether presented explicitly in this way or not, can be understood as factorization mechanisms~\cite{li_matrix_2015}.\footnote{Sometimes also referred to as matrix mechanisms.} In this framework, the target workload for the problem of continual counting (prefix sum queries over a time horizon $T$) is represented by the lower triangular all-ones matrix $A \in \mathbb{R}^{T \times T}$. Then for a chosen factorization $LR = A$, the mechanism estimates the intermediate queries in $R$ via the Gaussian mechanism, and then post-processes them using the linear transformation $L$ to get answers to the original queries.

This line of work has two flavors: The first considers hierarchical methods that correspond to sparse factorizations of $A$. That is, approaches similar to the binary tree mechanism, typically involving a combination of more clever aggregation functions and modifications to the tree structures~\cite{dp_histograms_2010,qardaji_2013,honaker2015,AnderssonP23,AnderssonPST25}.
The second set of works considers dense factorizations $A = LR$~\cite{denissov_improved_2022,fichtenberger_constant_2023,HenzingerUU23,DvijothamMPST24,AnderssonP25,HenzingerU25,henzinger2025normalizedsquarerootsharper}, where $R$ lacks a simple combinatorial structure, but instead is structured in some other way (for example, $R$ could have Toeplitz structure). These methods can achieve near-optimal concrete accuracy for continual counting. However, we cannot directly apply these improvements in the context of cardinality estimation, as discussed below.

\paragraph{\bf Inability to Directly Use Improved Mechanisms} For some applications (such as in differentially private machine learning), the continual counting algorithms are run on data streams where for a neighboring input stream, the input differs by a bounded amount at exactly one (or a small constant number of) time \pj{steps.}
In this case, the improved mechanisms for continual counting can be leveraged directly to obtain better utility.

In cardinality estimation problems such as counting distinct elements or releasing the degree list of a graph, the reduction to continual counting is more complicated. A frequent strategy for such problems is to construct the \emph{difference stream} $\diff{f}[t] = f(\xvec[0:t]) - f(\xvec[0:t-1])$ and to run the continual counting algorithms on this stream \cite{SongLMVC17, FichtenbergerHO21, JainKRSS23, RaskhodnikovaS24}.

One such work is that of Jain et al.~\cite{JainKRSS23}, which uses this strategy for counting distinct elements. Unfortunately, for this task, a change to the input stream at one time step $\xvec[t]$ could result in as many as $T$ changes to the difference stream (in fully dynamic streams where elements can be both inserted and deleted). Thus, a naive $\eps$-DP mechanism, one that uses an $\eps/T$-DP mechanism for continual counting to release the prefix sums of this difference stream (giving an $\eps$-DP guarantee via group privacy), would result in a \pj{factor $T$ blowup in error compared to continual counting.}\footnote{Note that such an algorithm would have worse error than the simple baseline of adding Gaussian noise to the function value at each time step and paying for composition over $T$ timesteps.} A similar phenomenon occurs for other cardinality estimation problems, including the estimation of triangle counts and degree lists for graph streams.

To get around this, Jain et al.~\cite{JainKRSS23} are the first to suggest parametrized accuracy guarantees, specifically for $\countdistinct$. They define the notion of \emph{maximum flippancy} of a stream, and show that maximum flippancy bounded by $k$ implies that a change to the input stream at one time step $\xvec[t]$ could result in at most $k$ changes to the difference stream. A similar group privacy argument would then result in a factor of $k$ error blowup for streams with a guaranteed bound $k$ on the maximum flippancy.\footnote{Note that this would give privacy and accuracy for streams with maximum flippancy bounded by $k$; extending to privacy for all streams can be done via a pre-processing step.} Jain et al. improve \pj{the dependence on maximum flippancy to a factor of $\sqrt{k}$}
via a non-black-box use of the binary tree mechanism. 

Unfortunately, their analysis is specific to the binary tree mechanism (which suffers from unfavorable constants in error bounds), and as raised in their paper, it is unclear if mechanisms that involve smaller noise addition (such as sophisticated Toeplitz factorization mechanisms) can be substituted in its place. Their analysis is also tailored to $\countdistinct$, and does not provide a roadmap to obtaining similar accuracy guarantees for other cardinality estimation problems like $\degreecount$. 
 This motivates the main question of this paper. \\

\noindent\fbox{%
    \begin{centering}
        \begin{minipage}{0.95\textwidth}
        {\it Can we improve the accuracy of differentially private cardinality estimation under continual observation by replacing the binary tree mechanism with other continual counting mechanisms?}
        \end{minipage}
    \end{centering}  
}

\vspace{3mm}
We answer this question affirmatively, giving significant improvements in accuracy for several cardinality estimation problems.

We note that mechanisms for cardinality estimation problems are used as a building block in more complicated tasks, and so achieving better error bounds for them is critical for encouraging more practical deployment of differential privacy. For example, $\countdistinct$ is used to detect abnormalities in networks \cite{Akella2003}, compute genetic differences between species \cite{Baker2018} etc., $\trianglecount$ is used to detect and monitor cohesiveness of communities in a social network \cite{EckmannM02}, detect presence of spamming activity in Web graphs \cite{BecchettiBCG08} etc., $\degreecount$ is used to track information about the spread of STDs \cite{Chandra2022}, compute parameters of online social networks \cite{DasguptaKS14} etc. Indeed, for standard continual counting, the constant factor improvements in accuracy via novel factorization mechanisms \cite{kairouz_practical_2021,denissov_improved_2022,fichtenberger_constant_2023,HenzingerUU23,DvijothamMPST24,mcmahan2024hasslefree} have encouraged their adoption in private machine learning systems \cite{2ZACKMRZ23,jaxprivacy2025}.




\section{Our Results}

\begin{table}[ht!]
\centering
\footnotesize
\setlength{\tabcolsep}{4pt}
\renewcommand{\arraystretch}{1.8}

\caption{Comparison of our (MaxSE) accuracy bounds with that from prior work for $1/2$-zCDP mechanisms over streams of length $T$. These results are obtained using \Cref{cor:toep-error} in conjunction with \Cref{cor:err-countdist}, \Cref{cor:err-degcount}, and \Cref{cor:err-triang-count}.}
\label{tab:results}
\begin{tabular}{|
    p{2.4cm}|
    p{1.6cm}|
    p{3.2cm}|
    p{3.5cm}|
    p{1.2cm}|
}
\hline
\textbf{Problem} & \textbf{Paper} & \textbf{Parameter used in Accuracy} & \textbf{Accuracy Upper Bounds}&
\textbf{Privacy for all streams?} \\
\hline

$\countdistinct$
& \cite{JainKRSS23}
& Maximum flippancy $k_{\mathrm{cd}}$
& $\sqrt{k_{\mathrm{cd}}}\log T, \sqrt{T}$ & 
 Yes \\ \Cline{2-5}{0.25pt}

\textbf{}
& \cellcolor{green!12} \textbf{Our Work}
& \cellcolor{green!12} Maximum flippancy $k_{\mathrm{cd}}$ 
& \cellcolor{green!12} \textbf{$\frac{1}{\pi \log_2 e}\sqrt{k_{\mathrm{cd}}}\log T$} 
& \cellcolor{green!12} Yes \\ 
\hline


$\degreecount$
& \cite{RaskhodnikovaS24}
& Stream length $T$
& $\sqrt{T}$
& Yes \\ \Cline{2-5}{0.25pt}

\textbf{}
& \cellcolor{green!12} \textbf{Our Work}
& \cellcolor{green!12} Degree contribution $k_{\mathrm{deg}}$ (\Cref{def:degcont})
& \cellcolor{green!12} \textbf{$\frac{\sqrt{2}}{\pi \log_2 e}\sqrt{k_{\mathrm{deg}}}\log T$}
& \cellcolor{green!12} Yes \\ 
\hline


$\trianglecount$
& \cite{RaskhodnikovaS24}
& Stream length $T$
& $T^{3/2}$  
& Yes \\ \Cline{2-5}{0.25pt}

\textbf{}
& \cellcolor{green!12} \textbf{Our Work}
& \cellcolor{green!12} Triangle contribution $k_{\mathrm{tri}}$ (\Cref{def:maxtrianglecontrib}), degree $D$
& \cellcolor{green!12} \textbf{$\frac{1}{\pi \log_2 e}\sqrt{k_{\mathrm{tri}} D}\log T$}
& \cellcolor{green!12} No \\ 
\hline


\end{tabular} 
\end{table}

\noindent Our results focus on two notions of error with respect to a function\footnote{The definition here is for one-dimensional functions $f$. If $f$ is vector-valued (for example, $\degreecount$), root mean squared error and root max squared error are defined as the maximum over coordinates of the root (mean or maximum) squared error for the output of the mechanism restricted to each coordinate.} $f$ for a continual release mechanism $\cM$ that processes the stream at each time step $t \in [T]$ and produces answer $a_t$: the root mean (expected) squared error $\meanse(
\cM, T
) = \sqrt{\mathbb{E}\left[ \frac{1}{T}\sum_{t\in[T]}(f(\xvec[0:t]) - a_t)^2 \right]}$, and the root maximum (expected) squared error  $\maxse(\cM,T) =  \sqrt{\max_{t \in [T]} \mathbb{E}\left[ (f(\xvec[0:t]) - a_{t})^2 \right]}$. In this summary of our results, we focus on $\maxse$; the associated results are summarized in Table~\ref{tab:results}. The results for root mean squared error are similar and are discussed later in the paper. Our privacy guarantees are expressed in terms of zero-concentrated differential privacy (zCDP), see \Cref{sec:prelims} for a formal definition. 

\paragraph{Counting Distinct Elements.}
For counting distinct elements under item level privacy (neighboring streams differ in occurrences of a single item), prior work of \cite{JainKRSS23} used the binary tree mechanism to obtain a $1/2$-zCDP algorithm with root maximum squared error approximately $\sqrt{k}\log T$, where $k$ denotes the \emph{maximum flippancy} of the stream (the maximum number of times any item toggles between present and absent). For realistic streams, we expect $k \ll T$—for example, a user account on a streaming service switches between “logged in’’ and “logged out’’ far fewer times than the total number of status updates. 

Our work gives bounds in terms of the same parameter, but improves the constant by roughly a factor of $4$---see Figure~\ref{fig:error_vs_k} for an exact relative comparison.%

\begin{figure}[h]
    \centering
    \subcaptionbox{Relative improvement in $\mathrm{MaxSE}$.\label{fig:maxse-vs-k}}{\includegraphics[width=0.49\linewidth]{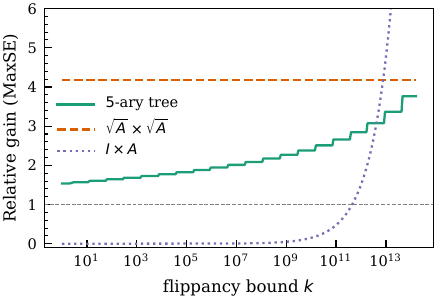}}%
    \hfill
    \subcaptionbox{Relative improvement in $\mathrm{MeanSE}$.\label{fig:meanse-vs-k}}{\includegraphics[width=0.49\linewidth]{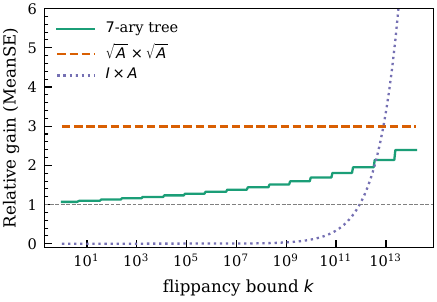}}%
    \hfill
    \caption{
        Root maximum and root mean expected squared error comparison between different mechanisms for $\countdistinct$ over $T=2^{50}$ time steps for a fixed flippancy bound $k$.
        Errors are plotted relative to the exact error bound achieved by~\cite{JainKRSS23}.
        The upper bounds for $b$-ary trees with subtraction are plotted using~\Cref{thm:leading-const-tree-approx} (including lower order terms), each for the asymptotically optimal choice of $b$.
        The upper bound for the square-root factorization is using \Cref{cor:toep-error} (including lower order terms), and the error for the naive factorization is exact (see \Cref{thm:naive}).\label{fig:error_vs_k}
    }
\end{figure}

\paragraph{Degree Histograms.}
For releasing degree histograms under item-level privacy (neighboring graph streams differ in all occurrences of a single edge), the best previously known algorithm relied on naive composition (recomputing the histogram at each time step), giving error that scaled polynomially with $T$. In addition, for notions of error closely related to $\maxse$ \footnote{They work with the $\ell_{\infty}$ notion of error- $\mathbb{E}[\max_{t \in [T]} \|f(\xvec[0:t]) - a_t \|_{\infty}]$. Our upper bounds give improved error in this notion of error as well; we discuss more in \Cref{sec:relatedwork}.}, prior work of \cite{RaskhodnikovaS24} established strong polynomial accuracy lower bounds in terms of the time horizon~$T$. Hence, our work instead achieves accuracy in terms of a new parameter we introduce: the \emph{maximum degree contribution}, the maximum number of times an edge toggles between present and absent. 
\pj{This quantity is}
much smaller than the number of time steps $T$ in 
\pj{many natural graph streams}.
For instance, in a stream capturing user similarity based on overlapping Amazon carts over time, the number of time steps during which two given users transition between sharing and not sharing items is typically far smaller than the total number of time steps. 

\paragraph{Triangle Counting.}
For releasing triangle counts under item-level privacy (neighboring graph streams differ in all occurrences of a single edge), the best known algorithm prior to our work relied on composition (giving polynomial-in-$T$ error), and prior work of \cite{RaskhodnikovaS24} proved strong accuracy lower bounds in terms of the time horizon~$T$ for notions of error closely related to $\maxse$. Our work instead obtains accuracy in terms of a new parameter we define: the \emph{maximum triangle contribution}, the maximum total change in the number of triangles containing a particular edge across the stream. For graphs of degree~$D$ over $T$ updates, we expect this quantity to be much smaller than $DT$. 

A limitation, reflected in Table~\ref{tab:results}, is that our guarantee does not provide privacy for \emph{all} graphs: it holds only for graphs whose degree is bounded by a public parameter~$D$ and triangle contribution by a public parameter~$k_{\mathrm{tri}}$. In contrast, using naive composition gives significantly worse accuracy, but gives privacy for all streams. We note that this result mirrors a line of work on graph privacy for insertion-only streams \cite{FichtenbergerHO21}, where privacy is established only for graphs with bounded degree. \\ 

\noindent We achieve the above results through a unified analysis that isolates a set of properties of sensitivity vectors for cardinality estimation problems, and shows via a novel analysis how state-of-the-art Toeplitz factorizations can be leveraged for problems with such sensitivity vectors. In particular, we expect that our results will apply to give improved accuracy bounds for cardinality estimation problems beyond those presented above (for example, counting the number of input entries satisfying a particular property (say, the number of movies rated 5 stars by more than 1000 users), counting the number of occurrences of a particular subgraph in a graph etc.) 


\subsection{Low-Space Error Bounds}\label{sec:lowspaceresults}

We note that one drawback of the above results using a Toeplitz factorization-based mechanism is that it requires $\Omega(T)$ space to implement, in contrast with tree-based methods that can be implemented in $O(\log T)$ space.\footnote{There is empirical work~\cite{mcmahan2024hasslefree} indicating that for a wide range of values of $T$, Toeplitz factorizations with the properties we need can be implemented in polylogarithmic space. Accompanying theoretical results are also known in some settings~\cite{DvijothamMPST24}, but the factorizations in this work do not satisfy the properties we need.}
Hence, we give results that leverage improved tree-based methods and novel analysis techniques to maintain $O(\log T)$ space while obtaining improved accuracy. The results are summarized in Table~\ref{tab:resultstree}.

\begin{table}[ht!]
\centering
\footnotesize
\setlength{\tabcolsep}{4pt}
\renewcommand{\arraystretch}{1.8}

\caption{Comparison of our ($\maxse$) accuracy bounds using $1/2$-zCDP tree aggregation mechanisms  over streams of length $T$. These results are obtained using \Cref{thm:leading-const-tree-approx} in conjunction with \Cref{cor:err-countdist}, \Cref{cor:err-degcount}, and \Cref{cor:err-triang-count}. } 
\label{tab:resultstree}
\begin{tabular}{|
    p{2.4cm}|
    p{1.6cm}|
    p{3.2cm}|
    p{4.5cm}|
    p{4cm}|
}
\hline
\textbf{Problem} & \textbf{Paper} & \textbf{Parameter used in Accuracy} & \textbf{Accuracy Bounds} \\
\hline

$\countdistinct$
& \cite{JainKRSS23}
& Maximum flippancy $k_{\mathrm{cd}}$
& $\sqrt{k_{\mathrm{cd}}}\log T, \sqrt{T}$ \\
\Cline{2-4}{0.25pt}

\textbf{}
& \cellcolor{green!12} \textbf{Our Work}
& \cellcolor{green!12} Maximum flippancy $k_{\mathrm{cd}}$
& \cellcolor{green!12} $0.609 \sqrt{k_{\mathrm{cd}} \log T \log(T/k_{\mathrm{cd}})}$ \\
\hline


$\degreecount$
& \cite{RaskhodnikovaS24}
& Stream length $T$
& $\sqrt{T}$
\\ \Cline{2-4}{0.25pt}

\textbf{}
& \cellcolor{green!12} \textbf{Our Work}
& \cellcolor{green!12} Degree contribution $k_{\mathrm{deg}}$ (\Cref{def:degcont})
& \cellcolor{green!12} $0.609 \sqrt{2k_{\mathrm{deg}} \log T \log(T/k_{\mathrm{deg}})}$ \\ 
\hline


$\trianglecount$
& \cite{RaskhodnikovaS24}
& Stream length $T$, degree $D$
& $T^{3/2}$ 
\\
 \Cline{2-4}{0.25pt}

\textbf{}
& \cellcolor{green!12} \textbf{Our Work}
& \cellcolor{green!12} Triangle contribution $k_{\mathrm{tri}}$ (\Cref{def:maxtrianglecontrib}), degree $D$
& \cellcolor{green!12} $0.609 \sqrt{k_{\mathrm{tri}} D \log T \log(DT/k_{\mathrm{tri}})}$  \\
\hline


\end{tabular} 
\end{table}

One important caveat to note is that obtaining privacy for all streams for $\countdistinct$ and $\degreecount$ while maintaining superior accuracy for realistic streams is done with a pre-processing step that tracks the parameter of interest for the stream and modifies the stream accordingly to satisfy bounds on the parameter. It is unclear how to track the parameters of interest while using space that is sub-polynomial in the time horizon $T$. Hence, using tree-based methods alone will not allow us to simultaneously achieve privacy for all streams and logarithmic space for these problems. 

\pj{Nonetheless, these results are significant since they reduce the problem to that of tracking parameters in low space.}
We note that for $\countdistinct$, recent work \cite{CummingsEMMOZ25} gives a method to track a variant of maximum flippancy\footnote{Cummings et. al. \cite{CummingsEMMOZ25} call this quantity `occurrency', it roughly corresponds to the maximum number of updates for a single item.} with space $\tilde{O}(T^{1/3})$, and using the tree-based methods presented in this section would give improved accuracy while also preserving sublinear space, whereas using Toeplitz factorizations as in the previous section would not give sublinear space.

We make the following notes about the results in Table~\ref{tab:resultstree}, which give bounds in terms of the same parameters as in Table~\ref{tab:results}. Firstly, the leading constant is approximately $0.6$, better than that obtained via the binary tree mechanism (see~\Cref{fig:error_vs_k}). In contrast, the constant is worse than that obtained by Toeplitz Factorizations by roughly a factor of $2$. However, our analysis of tree-based algorithms gives an asymptotic improvement over previous analyses of tree-based algorithms (as well as over our upper bounds via Toeplitz factorizations) with the dependency on $T$ being $\sqrt{\log T \log T/k}$ instead of $\log T$. For small $k$, this difference is insignificant, but for larger $k$ (polynomial in $T$), it means the tree-based bounds are asymptotically better than our upper bounds for Toeplitz factorizations. We compare our bounds in more detail in \Cref{sec:comparison}. 


We also show that tree-aggregation-based factorization mechanisms give state-of-the-art bounds for $(\eps, 0)$-DP (giving asymptotically better bounds than what one would get via group privacy). We refer to \Cref{sec:trees-error} for a detailed description of these bounds.

We discuss our techniques to obtain all the above results in the next section.

\section{Technical Overview}\label{sec:contributions}

\subsection{Overview of Framework}

In order to prove the results in the previous section, we start by characterizing properties of sensitivity vectors for our problems of interest. We will then prove general results for problems whose sensitivity vectors satisfy these properties, which will give a unified framework from which we obtain results for all the cardinality estimation problems.

For a function $f$, let $\xvec$ and $\yvec$ be two neighboring streams and let $\diff{f}$ be a function that takes a stream $\xvec$ and produces the corresponding difference stream $\diff{f}(\xvec)[t] = f(\xvec[0:t]) - f(\xvec[0:t-1])$. Since we are interested in applying continual counting mechanisms to the difference streams, a key object in our study is the sensitivity vector $\vec{\Delta}_f =  \diff{f}(\xvec) -  \diff{f}(\yvec)$. Note that if the function $f$ is vector valued, so are $\diff{f}[t]$ and  $\vec{\Delta}_f[t]$. \\


 \noindent {\bf Sensitivity Vector Sets for Cardinality Estimation Problems (\Cref{sec:senssets}).}
Our first main contribution is characterizing properties of sensitivity vectors that we can leverage to analyze the continual counting mechanisms of interest. 
\begin{itemize}[leftmargin=*]
 \item {\it $\mathbf{\countdistinct}$, $\degreecount$, and $\trianglecount$ have sensitivity vectors with \uline{bounded interval sums}}. For $\countdistinct$ and $\degreecount$, we identify a crucial property of the sensitivity vectors that we leverage in our proofs of error bounds---they are \textbf{alternating}.
That is, the non-zero entries alternate; if a non-zero entry of $\vec{\Delta}_{\countdistinct}$ is $1$, the next non-zero entry has to be $-1$ (and vice versa). Similarly, for $\degreecount$, we observe that there are two indices $a$ and $b$ such that the sensitivity vector $\vec{\Delta}_{\degreecount}(\xvec, \yvec)$ restricted to either of these two indices also has the same alternating property. For every other index $c$, $\vec{\Delta}_{\degreecount}(\xvec, \yvec)[c] = 0^T$. We note that the `alternating' property is also equivalent to all interval sums of the sensitivity vector being bounded between $-1$ and $1$ \pj{(\cref{claim:boundedvecs-are-alternating})}. For the function $\trianglecount$, we show that for graph sequences with degree $D$, while the sensitivity vectors of the difference streams for this function are not `alternating', they satisfy the more general property that all the interval sums are bounded between $-D$ and $D$. 

\item \noindent {\it Sensitivity vectors for realistic $\mathbf{\countdistinct}$, $\degreecount$, and $\trianglecount$ streams have \uline{bounded $\ell_1$-norm}}. 
A key point we leverage to get better bounds for realistic streams is that the set of updates for a single element has a \textit{bounded contribution} on the output statistic. We show that such bounded contribution 
\pj{is}
captured by the $\ell_1$ norm of the sensitivity vector.

\noindent For counting distinct elements in the fully dynamic streaming model, prior work \cite{JainKRSS23} introduced the \textit{maximum flippancy} $\Flip(\xvec)$ of the stream, defined as the maximum number of time steps an item goes from present to absent (or vice versa) in the stream. For streams with maximum flippancy at most $k$, we observe that the corresponding sensitivity vectors have at most $k$ non-zero entries ($k$-sparse). For $\degreecount$, we define a new parameter that we term the `maximum degree contribution' of an edge, and show that maximum degree contribution bounded by $k$, implies that for all nodes $c$, $\vec{\Delta}_{\degreecount}(\xvec, \yvec)[c]$ is $k$-sparse. Note that for these vectors with $\vec{\Delta} \in \{-1,0,1\}^T$, $k$-sparsity implies that the $\ell_1$ norm of the vector is bounded by $k$. For $\trianglecount$, we also define a new parameter, the `maximum triangle contribution' of the stream, which we prove is exactly equal to the largest $\ell_1$ norm of a sensitivity vector.

\end{itemize}

\noindent Inspired by the properties we uncover, we define the set of integer vectors with interval sums bounded between $-D$ and $D$, and $\ell_1$ norm bounded by $k$ to be the \textit{sensitivity vector set} $\cS_{D,k}$. We further define two 
\pj{integer-valued} sum streams to be $\boundedvecs{D,k}$-neighboring if their difference lies in $\boundedvecs{D,k}$. The technical arguments of the paper will focus on continual counting on $\boundedvecs{D,k}$-streams, which we define as continual counting with differential privacy defined with respect to the $\boundedvecs{D,k}$-neighboring relation. \\

\noindent {\bf Leveraging Toeplitz Factorization Mechanisms (\Cref{sec:toep}).} 
We now study continual counting on $\boundedvecs{D,k}$-streams via \textit{factorization mechanisms}, thereby giving bounds for our cardinality estimation problems of interest.
In such mechanisms, the prefix sums over time horizon $T$ are represented by the lower triangular $T \times T$ all-ones matrix $A$, and hence the prefix sums of stream $\diff{f}$ correspond to the matrix-vector product $A \cdot \diff{f}$.
Factorization mechanisms consider factorizations $A = LR$, and release $L(R \cdot \diff{f} + \zvec)$ where $\zvec$ is an i.i.d. Gaussian noise vector with variance scaled proportionally to the \textit{sensitivity} $\sens_2(R,\boundedvecs{D,k}) \equiv \max_{\vec{\Delta} \in \boundedvecs{D,k}} \|R \vec{\Delta} \|_2$.
It is known from the wide study of such mechanisms (see \Cref{sec:prelims}) that the above approach is differentially private with appropriate parameters and can be implemented in an online fashion.
Additionally, the error of such mechanisms is completely captured by the sensitivity $\sens_2(R,\boundedvecs{D,k})$, and matrix norms on $L$.
Our technical arguments reason about the first quantity, and rely on past work for the second.

We prove the following main theorem bounding $\sens_2(R,\boundedvecs{D,k})$ under some restrictions on the right factorization matrix $R$.

\begin{theorem}\label{thm:sensSdk}
Let $T, k, D \in \mathbb{N}$.
Let $R$ be a lower-triangular $T \times T$ Toeplitz matrix with non-increasing and non-negative lower-diagonal values. Let $\vec{\Delta} \in \boundedvecs{D,k}$. Then,\footnote{We note that if $k \leq D$, then the interval sum bound on $\dvec$ is vacuous. In this case, for any matrix $R$, $\sens_2(R,\boundedvecs{D,k}) = k \| R \|_{1\to 2} \leq \sqrt{kD}\| R \|_{1\to 2}$ (see e.g., proof of \cite[Fact 2.1]{HenzingerUU23} for the first step). Hence, the interesting range for the theorem is when $k > D$.}
$$\|R \vec{\Delta} \|_2 \leq \sqrt{k D} \cdot \|R\|_{1 \to 2}\,.$$
where $\|R\|_{1 \to 2}$ is the maximum $\ell_2$ norm of any column in matrix $R$.
\end{theorem}

Note that the case when $k=D=1$, has been studied extensively in prior work~\cite{fichtenberger_constant_2023, HenzingerUU23,henzinger2024unifying,DvijothamMPST24,henzinger2025normalizedsquarerootsharper}), and it is known from this work that the square-root factorization where $L = R = \sqrt{A}$ satisfies the properties needed to apply \Cref{thm:sensSdk}, giving the upper bound in the following result.

\begin{theorem}[Informal Error Bounds for Continual Counting on $\boundedvecs{D,k}$-streams]\label{cor:errtoepbsintro}
Fix $T,k,D \in \mathbb{N}$. Let $A$ be the lower-triangular all ones matrix (of size $T \times T$). Then, the factorization mechanism with $L = R = \sqrt{A}$ achieves the following error bounds for $1/2$-zCDP continual counting on $\cS_{D,k}$-streams of length $T$:
\begin{align*}
    \maxse(L, R, \cS_{D,k}) &\leq \left( \frac{1}{\pi\log(e)} + o(1) \right)  \sqrt{kD} \log T\,,\\
    \meanse(L, R, \cS_{D,k}) &\leq \left( \frac{1}{\pi\log(e)}  + o(1) \right)  \sqrt{kD} \log T\,.
\end{align*}
Additionally, if $k = O(T^{1/3})$ and $k \geq D$, we have that
\begin{align*}
    \maxse(L, R, \cS_{D,k}) &\geq \left(\frac{1}{\pi\log(e)} - o(1) \right)  \sqrt{kD \log(DT/k) \log T} \,,\\
    \meanse(L, R, \cS_{D,k}) &\geq \left(\frac{1}{\pi\log(e)} - o(1)\right)  \sqrt{kD \log(DT/k) \log T }\,.
\end{align*}
\end{theorem}

Note that the lower bound indicates that our analysis of the square-root factorization is tight up to lower order terms for a wide range of values of $k$ and $D$. We conjecture that the dependence on $T$ in the lower bound (specifically $\sqrt{\log (DT/k) \log T}$ vs $\log T$) is tight and that directly analyzing the square root matrix (instead of factoring through our general result for a broad class of Toeplitz matrices) might give this improved bound. We also note that when $k=D=1$, our results recover the tightest bounds on the square-root factorization for standard continual counting~\cite{henzinger2025normalizedsquarerootsharper}. 

Carefully applying the above general result to our problems of interest gives the results in Table~\ref{tab:results}. 
\\

\noindent {\bf Leveraging Tree-Based Factorization Mechanisms (\Cref{sec:trees})}. As mentioned in \Cref{sec:lowspaceresults}, to obtain improved error bounds while preserving low (logarithmic in the time horizon) space, we study tree aggregation factorization mechanisms in addition to dense Toeplitz factorizations. Building on work by \cite{AnderssonP23, AnderssonP25} on improved tree-aggregation mechanisms, as well as exploiting additional redundancy in the tree in our setting, we introduce a factorization 
$A = \hat{L}_b \hat{R}_b$, and study the quantity $\sens_2(\hat{R}_b, \boundedvecs{D,k})$, thereby characterizing the error notions of interest.

\begin{theorem}[Informal error bounds for tree-based mechanisms]\label{thm:intro-tree-eb}
    Let $k>D , T, b \in \mathbb{N}$. The $b$-ary tree mechanism, with subtraction (with associated factors $\hat{L}_b, \hat{R}_b$) achieves the following error bounds for $1/2$-zCDP continual counting on $\cS_{D,k}$-streams of length $T$:
   \begin{align*}
       \maxse(\hat{L}_b, \hat{R}_b, \cS_{D,k}) &\leq \left(\frac{\sqrt{b-1}}{\sqrt{2}\log b} + o(1)\right)\sqrt{Dk\log(DT/k)\log(T)}\,,\\
       \meanse(\hat{L}_b, \hat{R}_b, \cS_{D,k}) &\leq \left(\frac{\sqrt{b(1-1/b^2)}}{2\log b} + o(1)\right)\sqrt{Dk\log(DT/k)\log(T)}\,.
   \end{align*}
   where the constant in the parenthesis is minimized for $b=5$ and $b=7$ respectively, with corresponding values $0.609$ and $0.466$.
\end{theorem}

We give closed form lower bounds for $\sens_2(\hat{R}_b, \boundedvecs{D,k})$ that match our upper bounds up to lower order terms. We also give an exact polytime approach to computing $\sens_2(\hat{R}_b, \boundedvecs{D,k})$ via dynamic programming, which can be leveraged in practice to obtain exact bounds. Full details appear in \Cref{sec:trees}

\subsection{Techniques Used in Proofs}

\noindent \textbf{Reduction from $S_{D,k}$ Streams to $S_{1,k}$ Streams.} A key ingredient used throughout our proofs is a reduction from reasoning about a stream in $S_{D,k}$ to reasoning about streams from $S_{1,k}$. We show the following in \Cref{sec:trees-reduction}:

\begin{theorem}\label{thm:alternating-to-general-intro}
    Let $1 \leq k\leq D \leq T$, $1 \leq m \leq T \in \mathbb{N}$. Let $R$ be any $m \times T$ matrix. 
    Then,
    \begin{equation*}
      D \cdot \sens_2(R, \boundedvecs{1, \lfloor k/D \rfloor})\, \leq \sens_2(R, \boundedvecs{D, k}) \leq \max_{k_1,\dots,k_D: \sum_{i=1}^D k_i = k} \sum_{i=1}^D  \sens_2(R, \boundedvecs{1, k_i})
    \end{equation*}
\end{theorem}

The lower bound in this result is straightforward and follows from simple facts about the scaling of norms. The upper bound is more intricate and involves breaking up a sensitivity stream $\vec{\Delta}$ from $S_{D,k}$ into up to $D$ different sensitivity streams $\vec{\Delta}^{(i)}$ from $S_{1,k_i}$, where $\sum_{i=1}^D k_i \leq k$. Each stream corresponds to a distinct value that a prefix sum of $\vec{\Delta}$ can take, and the streams are carefully created such that $\sum_{i=1}^D \vec{\Delta}^{(i)} = \vec{\Delta}$. Combining structural properties of this decomposition with basic norm inequalities yields the upper bound. Details appear in \Cref{sec:trees-reduction}.

This reduction is crucial throughout our analysis because streams in $S_{1,k}$ enjoy a much more rigid structure than those in $S_{D,k}$: they are \emph{alternating}, meaning their non-zero entries have magnitude $1$ and alternate in sign. \\

\noindent\textbf{Techniques for \pj{Toeplitz Factorization Mechanisms}.}
We overview our techniques for Toeplitz factorizations here (full details can be found in \Cref{sec:toep}.) \\
\noindent \textit{Overview of Upper Bound.} We give two proofs of Theorem~\ref{thm:sensSdk} of different flavors. One proof is combinatorial and the other is more algebraic. Via the known link between sensitivity $\sens_2(R,\boundedvecs{D,k})$ and factorization error, applying this to the square-root factorization yields the error upper bounds of \Cref{cor:errtoepbsintro}.

Let $c(i,j) = (R^\top R)[i,j]$. The starting point of both proofs is analyzing the quadratic form $\|R\vec{\Delta}\|_2^2$. As a precursor, we show that monotonicity of lower diagonal values of $R$ induces useful monotonicity properties of $c(\cdot,\cdot)$ that we use in both proofs.

The combinatorial proof proceeds by first considering the alternating case $S_{1,k}$. We decompose the quadratic form into (i) a “squared term’’ of the form $c(i,i)\vec{\Delta}[i]^2$ and (ii) a “cross term’’ involving $c(i,j)\vec{\Delta}[i]\vec{\Delta}[j]$ for $i\neq j$. The key step is to show that the cross term is non-positive. Using the alternating structure, whenever $\vec{\Delta}[i]\vec{\Delta}[j]=1$ for $i<j$, we can match this pair with another $i<j'$ where $\vec{\Delta}[i]\vec{\Delta}[j']=-1$ and $j'$ lies closer to $i$. Monotonicity properties of $c(\cdot,\cdot)$ then let us show that each matched pair contributes a non-positive amount to the cross term. This helps us prove Theorem~\ref{thm:sensSdk} for $D=1$, which we then extend to all $D$ using the reduction from $S_{D,k}$ streams to $S_{1,k}$ streams in Theorem~\ref{thm:alternating-to-general-intro}. Details appear in \Cref{sec:toep-upper-bound}.

The second proof is linear algebraic and reasons about $S_{D,k}$ directly; we find it more elegant, but much less intuitive. A key quantity that we work with is $\sum_{j=0}^{T-1}c(i,j) \vec{\Delta}[j]$. Writing $SUM_j = \sum_{i=0}^j \vec{\Delta}[i]$, we write $\vec{\Delta}[j]$ as a difference of adjacent prefix sums $SUM_{j+1} - SUM_j$, and then change indices of the sum- $\sum_{j=0}^{T-1}c(i,j) (SUM_{j+1} - SUM_j) = \sum_{j=0}^{T-1}[c(i,j) - c(i,j+1) ] \cdot SUM_j$, a transformation commonly called \emph{Abel summation}. It turns out that the latter quantity is easier to reason about via the monotonicity of $c(\cdot,\cdot)$, and we leverage this carefully to prove the theorem. The details can be found in \Cref{sec:toep-upper-bound}.

\noindent \textit{Overview of Lower Bound.} 
For the lower bound on the sensitivity $\sens_2(R,\boundedvecs{D,k})$ of the square-root factorization, we use the probabilistic method. We first treat $S_{1,k}$ streams: draw sensitivity vector $\vec{\Delta}$ uniformly at random from $S_{1,k}$ (conditioned on $\|\dvec\|_1=k$) and analyze $\mathbb{E}\left[\|\sqrt{A} \cdot \vec{\Delta}\|_2^2\right]$. We split the expectation into a squared term involving $\mathbb{E}[\vec{\Delta}[i]^2]$ and a cross term involving $\mathbb{E}[\vec{\Delta}[i]\vec{\Delta}[j]]$ for $i\neq j$. The former is easy to reason about because each coordinate of $\vec{\Delta}$ is equally likely to be non-zero.
Bounding the cross term ends up being the key challenge of the proof, which boils down to analyzing the probability that $\vec{\Delta}[i] \neq \vec{\Delta}[j]$, conditioned on them being non-zero.
This corresponds to the number of non-zero values falling between $i$ and $j$ being even, which is exactly the parity of a hypergeometric random variable. We hence prove a result characterizing this probability (which may be of independent interest) and use it to get a handle on the cross term. Controlling the absolute value of the cross term ends up requiring intricate analytical manipulation; we refer readers to \Cref{sec:toep-lowerbound} for the details. This gives the results for $S_{1,k}$ streams; extending to $S_{D,k}$ streams again uses Theorem~\ref{thm:alternating-to-general-intro}. \\

\noindent \textbf{Techniques for Tree-Aggregation\pj{-Based Factorization Mechanisms}.}
For tree aggregation, we work directly with the combinatorial structure of the tree and again begin with $S_{1,k}$ streams, exploiting their alternating property. We reduce the sensitivity analysis (for $\sens_2(R,\boundedvecs{D,k})$) to a new combinatorial problem we call \emph{parity counting on trees}: given a complete $b$-ary tree with $T$ leaves and $k$ balls placed on leaves, how many internal nodes have an odd number of balls in their subtree? Due to the alternating structure of $S_{1,k}$ streams, maximizing this quantity yields an exact characterization of the sensitivity $\sens_2(R,\boundedvecs{D,k})$ of the tree-aggregation matrix.

We use this characterization to derive upper and lower bounds that are tight up to lower order terms and to design a dynamic programming algorithm that computes the $\sens_2(R,\boundedvecs{D,k})$ exactly in polynomial time. Our combinatorial view also motivates a sparsified tree factorization tailored to our setting: since standard tree factorizations use only a $(1 - 1/b)$ fraction of nodes, an equivalent fraction of nodes with an odd number of balls in their subtree can be ignored when computing the sensitivity. Combining these insights with the reduction to $S_{1,k}$ streams yields Theorem~\ref{thm:intro-tree-eb} and related results; see \Cref{sec:trees}.
\subsection{Challenges in Analyzing Factorization Mechanisms}\label{sec:challenges}

For continual counting on $S_{D,k}$-streams, the case when $k=D=1$ has been studied extensively in prior work. Here, the square-root factorization achieves near-optimal bounds ~\cite{fichtenberger_constant_2023,HenzingerUU23}. One might therefore wonder whether those techniques extend naturally to $k > 1$ (even when $D = 1$). In this section, we explain why our different analytical approaches are necessary.

The error of any factorization mechanism depends on the value
$\max_{\dvec} \| R \dvec \|_2$,
where $R$ is the right factor, giving a natural optimization problem. 
When $k=1$, this optimization problem has a simple solution; $\max_{\dvec} \| R \dvec \|_2$ is exactly equal to $\| R\|_{1\to 2}$, the largest $\ell_2$ norm  over columns of $R$, a value that can be computed efficiently by iterating over the columns. Additionally, for structured factorizations like the square root mechanism, $R$ is a lower-triangular Toeplitz matrix with monotonically decreasing diagonals. This means that the worst-case sensitivity vector is always $[1,0,\dots,0]$, allowing for direct computation of the sensitivity $\sens_2(R,\boundedvecs{1,1})$. A natural first attempt to handle larger $k$ might be to reduce to $k=1$ by appealing to the operator norm of $R$. This would however give a suboptimal linear dependence on $k$ (as opposed to the square root dependence we achieve). 

\medskip \noindent
\textbf{Computational Challenges for Higher $k$.} To make progress for larger $k$, we first restricted attention to the class of alternating sensitivity vectors—those with integral entries of magnitude at most $1$ whose non-zero values alternate in sign, motivated by the structure arising from our cardinality estimation problems. 
However, even in this simplified setting, directly computing $\max_{\vec{\Delta}} \|R \vec{\Delta}\|_2$ is computationally prohibitive: A naive brute-force enumeration would have to iterate over $\Theta\bigl(T^{\,\min(k,\, n-k)}\bigr)$ candidate sensitivity vectors, which is infeasible as $k$ gets larger. In general, a convex quadratic \emph{maximization} problem over a discrete set can be NP-hard. Thus, unlike the $k=1$ setting, we cannot hope to solve the optimization problem by enumeration.

\medskip \noindent
\textbf{Identifying the Hardest Sensitivity Vector via Structure of Matrix.} One may instead hope to identify the ‘hardest’ sensitivity vector directly, as is possible for the square root factorization when $k=1$. Indeed, experiments on the square-root matrix for various values of $k$ exhibit an encouraging heuristic pattern: the nonzero entries of the optimal alternating $k$-sparse vector tend to be somewhat evenly spaced. (See \Cref{fig:sqrt-opt-vecs} for $T=30$ and $k \in [1,10]$.)

\begin{figure}[h!]
    \centering
    \includegraphics{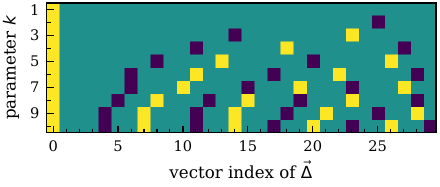}
    \caption{Heatmap showing the alternating $k$-sparse sensitivity vectors $\dvec$ which maximize $\| \sqrt{A} \dvec\|_2$ for each $k\in[1,10]$ and $n=30$. The yellow squares represent $1$, the purple squares represent $-1$ and the teal squares represent $0$. 
    \label{fig:sqrt-opt-vecs}}
\end{figure}

Unfortunately, there does not appear to be a clean closed-form description of these vectors. The alternating constraint introduces subtle interactions: note that to optimize $\|R \dvec \|_2$, we are interested in the norm of the linear combination of $k$ columns (where some columns are added and others are subtracted). Inserting a $-1$ at a particular position in $\vec{\Delta}$ can increase the absolute value of some coordinates of $R\vec{\Delta}$ (e.g., by making a previously small negative entry more negative), while simultaneously decreasing others (e.g., by cancelling positive contributions). How these effects trade off depends opaquely on the values of the matrix $R$.
Lacking a usable closed-form characterization, we ultimately abandoned the direct optimization approach from prior work and instead developed the combinatorial and algebraic upper-bound techniques and the probabilistic lower bound techniques described earlier.



\section{Related Work}\label{sec:relatedwork}


\paragraph{\bf Differential Privacy under Continual Observation.}

As described in the introduction, the seminal works of Dwork et al. \cite{dwork_differential_2010} and Chan et al. \cite{chan_private_2011} defined the continual observation model of differential privacy where inputs are received as a stream and the goal is to track a statistic over time while maintaining privacy. The version of the model they considered involves a single element arriving per time step and a single output being released per time step. We consider a more general model where many elements can arrive/updates can occur per time step, with a single output still released over each time step. This more general model has been studied in papers on private continual observation for graph problems \cite{SongLMVC17, FichtenbergerHO21, JainSW24}.

\paragraph{\bf Continual Counting and the Binary Tree Mechanism.}

Continual counting (or summation of bits) is a fundamental primitive used in many data analysis tasks. The papers defining the continual observation model (\cite{chan_private_2011} and \cite{dwork_differential_2010}) also gave the binary tree mechanism, a mechanism for continual counting with error scaling polylogarithmically in the time horizon $T$.

Mechanisms for continual counting are the main building block used to give mechanisms for many tasks in the continual observation literature, for example, continual observation of histograms \cite{CardosoR22, JainRSS23, HenzingerSS25}, state-of-the-art mechanisms for differentially private machine learning \cite{kairouz_practical_2021, denissov_improved_2022, choquette-choo_multi-epoch_2022, Choquette-ChooG23, 2ZACKMRZ23, PillutlaU2025}, various frequency statistics on a stream \cite{BolotFMNT13, Ghazi0NM23, JainKRSS23, EpastoMMMVZ23, CummingsEMMOZ25}, continual observation of various graph statistics \cite{SongLMVC17, FichtenbergerHO21, JainSW24, RaskhodnikovaS24, EpastoLMZ24}, and continual observation for clustering problems \cite{EpastoMZ23, TourHS24}. Hence, recent works that we will discuss next attempt to give mechanisms for continual counting that improve upon the constants of the binary tree mechanism (though they achieve the same asymptotic error). Most of these mechanisms are factorization mechanisms that exploit clever factorizations to achieve improved error bounds for private continual counting. Note that the neighboring notion studied in all these works is that adjacent streams differ in a single bit at one time-step. 
\paragraph{\bf Factorization Mechanisms for Continual Counting.}

Inspired by applications to cardinality estimation, we identify new neighboring relations under which we study private continual counting. Our work is effectively a study of the factorization mechanism~\cite{li_matrix_2015} applied to the lower-triangular matrix of all-ones, $A\in\mathbb{R}^{T\times T}$ under these new neighboring notions.
There is an extensive literature studying optimal matrix factorizations for private continual counting under the regular neighboring notion ~\cite{denissov_improved_2022,AnderssonP23,fichtenberger_constant_2023,HenzingerUU23,DvijothamMPST24,AnderssonP25,HenzingerU25,HenzingerKU25,AnderssonY2025,JacobsenF2025}. We now describe these in more detail. 

Tree-based factorizations~\cite{dwork_differential_2010,chan_private_2011,AnderssonP23,AnderssonPST25} were the first class of factorizations to achieve $O(\log T)$ error, and do so efficiently with $O(\log T)$ space and $O(T)$ time.
Their key structure is sparsity, which allows for not only efficient computation, but also for a low $\ell_2$ and $\ell_1$ sensitivity, making them also state-of-the-art for the separate problem of $\epsilon$-differentially private continual counting.
Techniques of \enquote{denoising} tree-based mechanisms have also been proposed, where the left factor is made denser for a constant factor improvement at the cost of efficiency~\cite{honaker2015,dp_histograms_2010,kairouz_practical_2021}.

More general dense factorizations have subsequently been proposed~\cite{denissov_improved_2022,fichtenberger_constant_2023,HenzingerUU23,HenzingerU25,henzinger2025normalizedsquarerootsharper}, and give improved leading constants of the error, often at a cost in efficiency.
We especially highlight the \emph{square-root factorization} $L=R=\sqrt{A}$, where $\sqrt{A}$ is Toeplitz and lower-triangular with entries $\sqrt{A}[i,j] = \binom{2(i-j)}{i-j}/4^{i-j}$ for $i\geq j$.
It attains an error of $\frac{\ln(T)}{\pi\sqrt{2\rho}} \pm O(1)$, which is optimal up to the additive constant for both of our error notions~\cite{fichtenberger_constant_2023,HenzingerUU23,matouvsek2020factorization}.
Subsequent work showed how to approximate this error up to the additive constant with a lower-triangular Toeplitz factorization in $O(\log^2 T)$ space and $O(T\log^2 T)$ time using rational-function approximation~\cite{DvijothamMPST24}.

\paragraph{\bf Factorization Mechanisms under `Non-Standard' Neighboring Relations.} We finally note that our work is not the first to study factorization mechanisms for releasing prefixes under continual observation for a \enquote{non-standard} neighboring relation.
Motivated by allowing a user to contribute multiple data points in a machine learning pipeline~\cite{kairouz_practical_2021}, the notion of \emph{$b$-min separated $k$-repeated participation} was introduced~\cite{choquette-choo_multi-epoch_2022,Choquette-ChooG23}.
Here any training example can appear at most $k$ times, each example separated by at least $b$ time steps.
For the factorization mechanism design problem, this translates to a sensitivity vector set $\cS_{(k, b)}\subseteq[-1, 1]^{T}$ where vectors are $k$-sparse and with non-zero entries having a spacing of at least $b$.
As our $k$-sparse alternating sensitivity vector set ($S_{1,k}$) makes up a strict subset of $\cS_{(k, 1)}$, any upper bound on the error for a factorization mechanism under $\cS_{(k, 1)}$-participation extends to our setup.
Crucially, as there is no analogous structure on $\cS_{(k, 1)}$ that enforces bounded prefix sums, these upper bounds are asymptotically worse than what we show.
In particular, factorizations based on a banded (inverse) square-root factorization~\cite{KalLamp24,KalininM2025} that are asymptotically optimal for this problem achieve a bound of $O(\sqrt{kT})$ ~\cite[Theorem 2]{KalininM2025}.
In contrast, we achieve a logarithmic dependence on $T$ in our setting.


\paragraph{\bf Comparison to Directly Related Work.} The most directly related work is that of \cite{JainKRSS23} who study counting distinct elements in the fully dynamic streaming model and introduce the maximum flippancy parameter. Their setting differs slightly from ours 
\pj{(neighboring streams in their notion are at most distance-2 away according to ours),}
but their algorithm extends to our model, and we compare 
\pj{the resulting bounds against ours.}
As shown in Figure~\ref{fig:error_vs_k}, our error improves over theirs across the full range of the maximum flippancy parameter. An alternative parameterization based on total flippancy was proposed by Henzinger et al.~\cite{HenzingerSS24}; while we focus on maximum flippancy—typically smaller for realistic streams—our improvements apply in that setting as well.  

For $\trianglecount$ and $\degreecount$, Raskhodnikova and Steiner~\cite{RaskhodnikovaS24} give strong polynomial (in $T$) lower bounds under item-level DP for expected $\ell_{\infty}$ error. Their model assumes only a single update per time step, and so while their lower bounds continue to apply in our setting, their upper bounds that are based on periodic recomputation do not apply. To the best of our knowledge, the previous best upper bounds for these problems come from naively running a batch DP algorithm at each time step and applying composition theorems. We introduce new parameters that are significantly smaller than the time horizon in realistic streams and obtain substantially improved error guarantees in terms of these parameters. Although our accuracy guarantees are stated in terms of $\maxse$ and $\meanse$, they also yield improvements for $\ell_{\infty}$ error via standard Gaussian tail bounds (see, for example, Section 2.4 of \cite{DvijothamMPST24}), providing a way around the lower bounds in \cite{RaskhodnikovaS24} for realistic streams. 

We also note that there is additional incomparable work on releasing $\trianglecount$ and\\ $\degreecount$~\cite{SongLMVC17,FichtenbergerHO21,JainSW24}).
These works focus on the insertion-only setting, which is less demanding than the fully dynamic case we study, though some of this work (e.g., \cite{JainSW24}) addresses the more challenging node-privacy model, whereas we focus on edge privacy. Moreover, \cite{JainSW24} gives a  $\trianglecount$ mechanism that is private for all insertion-only graph sequences, while our mechanism for $\trianglecount$ on fully dynamic streams is private only for a restricted class of sequences.

Finally, we emphasize that our focus differs from all of the above: we aim to obtain strong asymptotic guarantees and to optimize constants, whereas prior work targets asymptotics alone.

\section{Preliminaries}\label{sec:prelims}


Throughout the paper, we let $[i, j] = \{i, i+1,\dots, j-1, j\}$ and use the shorthand $[n] = [0, n-1]$.
We zero-index vectors, meaning for $\xvec\in\mathbb{R}^n$, $\xvec[i]$ is the $i$\textsuperscript{th} entry for $i\in[n]$.
We will also find it convenient to use the shorthand $\xvec[i:j]$, where $0\leq i \leq j\leq n-1$ for the vector in $\mathbb{R}^{i-j+1}$ with entries $(\xvec[i:j])[k] = \xvec[i+k]$.
Additionally we let $\xvec[:k] = \xvec[0:k]$ and $\xvec[k:] = \xvec[k:n-1]$.
Analogously for indexing a matrix $B\in\mathbb{R}^{m\times n}$, we use $B[i, j]$ for $(i,j)\in [m]\times[n]$.
We use $\log_b(x)$ for the logarithm of $x$ in base $b$, $\ln(x) = \log_e(x)$ for the natural logarithm and $\log(x) = \log_2(x)$.

Our main results are stated for zero-concentrated differential privacy, but we also discuss pure differential privacy.
We note that all notions of differential privacy quantifies the indistinguishability of outputs of a randomized algorithm $\cA : \mathcal{X}^n \to \mathcal{Y}$ when run on sensitive \emph{neighboring inputs}.
Throughout the paper we use the notation $x\sim x'$ for two such inputs, which are implicitly assumed be in $\mathcal{X}^n$ from context.

We give all relevant privacy definitions below.
\begin{definition}[$(\epsilon, \delta)$-Differential Privacy~\cite{dp_2006,dwork_algorithmic_2013}]
    A randomized algorithm $\cM : \mathcal{X}^n \to \mathcal{Y}$ is $(\epsilon, \delta)$-differentially private if for all $S\subseteq\mathrm{Range}(\cM)$ and all neighboring inputs $x, x'\in \mathcal{X}^n$, we have that
    \begin{equation*}
        \Pr[\cM(x)\in S] \leq \exp(\epsilon)\cdot\Pr[\cM(x')\in S] + \delta\,.
    \end{equation*}
    We refer to $(\epsilon, 0)$-differential privacy as $\epsilon$-differential privacy, or \emph{pure} differential privacy.
\end{definition}
\begin{definition}[$\rho$-Zero-Concentrated Differential Privacy~\cite{bun_concentrated_2016}]
    A randomized algorithm $\cM : \mathcal{X}^n \to \mathcal{Y}$ satisfies $\rho$-zCDP if for all neighboring inputs $x, x'\in \mathcal{X}^n$, we have that
    \begin{equation*}
        \forall \alpha > 1 : D_\alpha(\cM(x) \| \cM(x')) \leq \rho\alpha\,, 
    \end{equation*}
    where $D_\alpha(\cM(x) \| \cM(x'))$ is the $\alpha$-Rényi divergence between $\cM(x)$ and $\cM(x')$.
\end{definition}
We note that $\rho$-zCDP also implies $(\rho + 2\sqrt{\rho\ln(1/\delta)}, \delta)$-DP for every $\delta > 0$~\cite{bun_concentrated_2016}.
\begin{definition}[$\ell_p$ Sensitivity]
    Let $f : \mathcal{X}^n \to \bR^{n}$ be a function on sensitive inputs with a neighboring relation \enquote{$\sim$}.
    Then we define the (global) $\ell_p$ sensitivity of $f$ as
    \begin{equation*}
        \sens_p(f) = \max_{x\sim x'} \| f(x) - f(x')\|_p\,,
    \end{equation*}
    where we sometimes use $\sens(f) = \sens_2(f)$ as a shorthand.
\end{definition}
\begin{definition}[Gaussian Mechanism~\cite{bun_concentrated_2016}]
    Let $f : \mathcal{X}^n \to \bR^{n}$ be a function on sensitive inputs and define the (randomized) algorithm $\cM(x) = f(x) + \zvec$ where $\zvec\sim\cN(0, 0.5\cdot\sens_2(f)^2/\rho)^n$ is zero-mean, $n$-dimensional i.i.d.\ Gaussian noise. 
    Then $\cM$ satisfies $\rho$-zCDP.
\end{definition}
\begin{definition}[Laplace Mechanism~\cite{dp_2006}]
    Let $f : \mathcal{X}^n \to \bR^{n}$ be a function on sensitive inputs and define the randomized algorithm $\cM(x) = f(x) + \zvec$ where $\zvec\sim\mathrm{Lap}(0, \sens_1(f) / \epsilon)^n$ is zero-mean, $n$-dimensional i.i.d.\ Laplace noise.
    Then $\cM$ satisfies $\varepsilon$-DP.
\end{definition}

\subsection{Factorization Mechanism}
The \emph{factorization mechanism} (sometimes also referred to as \emph{matrix mechanism}~\cite{li_matrix_2015}) is central to our contributions.
Given a query matrix $A\in\mathbb{R}^{T\times T}$ and a set of sensitive inputs $\mathcal{D}\subseteq\mathbb{R}^{T}$ with neighboring relation \enquote{$\sim$}, it offers an approach for privately releasing $A\xvec$ for any $\xvec\in\mathcal{D}$.
For any factorization of $A$ into matrices $L\in\mathbb{R}^{T\times m}, R\in\mathbb{R}^{m\times T}$ such that $A=LR$, it induces a mechanism $\cM_{LR}(\xvec)$.
Given the $\ell_2$ sensitivity of $R$
\begin{equation*}
    \sens_2(R) =
    \max_{\substack{(\xvec, \yvec)\in\cD^2\\ \xvec\sim\yvec}} \| R\xvec -R\yvec\|_2\,,
    \max_{\substack{(\xvec, \yvec)\in\cD^2\\ \xvec\sim\yvec}} \| R(\xvec -\yvec)\|_2\,,
\end{equation*}
the mechanism outputs
\begin{equation*}
   \cM_{LR}(\xvec) = A\xvec + L\zvec\qquad\text{where}\qquad \zvec\sim\cN(0, 0.5\cdot\sens_2(R)^2/\rho)^m
\end{equation*}
for a $\rho$-zCDP guarantee.
The intuition behind the mechanism is that $R\xvec$ is released with the Gaussian mechanism as $R\xvec + \zvec$, and then post-processed by $L$ from the left into $\cM_{LR}(\xvec) = L (R\xvec + \zvec)$.

As our main contributions is to study factorization mechanisms under non-standard neighboring relations, we make the this dependence more explicit in our notation.
We define the \emph{sensitivity vector set} $\cS$ and express the associated $\ell_p$ sensitivity of $R$, $\sens_p(R, \cS)$, as
\begin{equation*}
   \cS = \{ \xvec - \yvec : \xvec, \yvec \in \cD\ \text{and}\ \xvec\sim\yvec \}\,,
   \quad \sens_p(R, \cS) = \max_{\dvec\in\cS}\| R\dvec\|_p = \sens_p(R)\,.
\end{equation*}
Throughout the paper we will leverage two common error metrics applied to the factorization mechanism:
\begin{align*}\label{eq:errormat}
    \maxse(L, R, \cS) &=
    \sqrt{\max_{t \in [T]} \mathbb{E}\left[ (A\xvec - \cM_{LR}(\xvec))_{t}^2 \right]}
    = \frac{\|L\|_{2\to\infty}\cdot \sens_2(R, \cS)}{\sqrt{2\rho}}\\
    \meanse(L, R, \cS)
     &= \sqrt{\mathbb{E}\left[ \frac{1}{T}\sum_{t\in[T]}(A\xvec - \cM_{LR}(\xvec))_{t}^2 \right]}
     = \frac{\|L\|_F \cdot \sens_2(R, \cS)}{\sqrt{2T\rho}}\,.
\end{align*}
where $\|\cdot\|_F$ is the \emph{Frobenius} norm, and $\| \cdot \|_{2\to\infty}$ is the largest $\ell_2$ norm over all rows of the matrix.
We will later also work with the $\|\cdot \|_{1\to 2}$ norm, which equals the largest $\ell_2$ norm taken over all columns of the matrix.
The two error metrics correspond to the \emph{root maximum (expected) squared error} and the \emph{root mean (expected) squared error} taken over all outputs.
As the factorization mechanism is exact on expectation, these errors are independent of the input and can be thought of as the root of the maximum and average variance over all outputs.

We also remark that the same framework can be used to derive a mechanism satisfying $\epsilon$-DP.
The only change is that the $R\xvec$ is now released with the Laplace mechanism which scales with the $\ell_1$ sensitivity.
The corresponding error metrics can be derived from replacing $\sens_2(R, \cS)/\sqrt{2\rho}$ by $\sqrt{2}\sens_1(R, \cS)/\epsilon$.
All results and statements throughout the paper are for $\rho$-zCDP unless explicitly stated otherwise.

\paragraph{Continual Observation}
Whenever the query matrix $A$ is lower-triangular, it is possible to study the problem of privately releasing $A\xvec$ \emph{under continual observation}, also referred to as \emph{(weighted) continual counting}~\cite{dwork_differential_2010,chan_private_2011}.
The same neighboring relation from before can be considered, but now $\xvec$ is streamed one entry at a time with updated objective:
\begin{equation*}
    \text{On receiving}\ \xvec[t]\ \text{at time}\ t,\,\text{immediately output}\ (A\xvec)[t]\,.
\end{equation*}
The corresponding factorization mechanism for the problem extends naturally:
\begin{equation*}
    \text{On receiving}\ \xvec[t],\,\text{immediately output}\ \mathcal{M}_{LR}(\xvec)[t] = (A\xvec)[t] + (L\zvec)[t]\,.
\end{equation*}
It satisfies $\rho$-zCDP under continual observation, even for a stronger notion of \emph{adaptive} DP~\cite{denissov_improved_2022}.
Here an input $\xvec[t]$ is allowed to depend on past private outputs $\mathcal{M}(\xvec)[t']$ for $t'< t$.
In particular, there are no structural constraints on either $L$ or $R$ for this stronger notion to hold.
Ultimately, our goal will be to reduce cardinality estimation problems to this setting.

\paragraph{Computational Model}
The time and space usage of the factorization mechanism are also of interest.
As $A\xvec$ is efficient to compute, the resource footprint is determined by the hardness of outputting $L\zvec$.
In line with past work on continual counting~\cite{AnderssonP23,DvijothamMPST24,AnderssonP25}, we restrict our attention to \emph{linear} algorithms for outputting $L\zvec$.
Space is measured by the number of real values that have to be stored, and time by the number of additions and multiplications.

\subsection{Problems of Interest}\label{sec:problems-of-interest}
In this section, we introduce the cardinality estimation problems we study in this paper, and the associated neighboring relations of interest. We note that the two error metrics with respect to a function $f$ that we consider are the \emph{root maximum (expected) squared error}
$$\maxse(\cM,T) =  \sqrt{\max_{t \in [T]} \mathbb{E}\left[ (f(\xvec[0:t]) - \cM(\xvec[0:t]))^2 \right]},$$
and the \emph{root mean (expected) squared error} $$\meanse(\cM,T) = \sqrt{\mathbb{E}\left[ \frac{1}{T}\sum_{t\in[T]}(f(\xvec[0:t]) - \cM(\xvec[0:t]))^2 \right]}.$$ For $\degreecount$, which has vector outputs, we will consider a natural variant of these error metrics, discussed in the section on graph functions below.
\paragraph{Counting Distinct Elements under Continual Observation:} We study the distinct element problem for turnstile (fully dynamic) streams. In this problem, elements come from a universe $\univ$. A stream $\xvec$ for this problem corresponds to a sequence of updates (each update is a set of insertions or deletions of elements, or a no op), that is for all $t \in [T]$, $\xvec[t] \in \{\{ \{+,-\} \times \univ \} \cup \bot\}^*$. We will use $\xvec[t][i]$ to represent the $i^{th}$ update at time step $t$. Informally, we will use $+u$ to represent the insertion of an element $u$ and $-u$ to represent the deletion of an element $u$. Now, we are ready to formally define the count distinct problem. $\countdistinct$ is a function that takes as input a stream $\xvec$ (a sequence of $T$ updates) and outputs a stream of $T$ non-negative integers. For $t \in [T]$,
    $$\countdistinct(\xvec)[t] = \sum_{u \in \univ}\mathbbm{1}(\text{$u$ has been inserted more often than deleted in $\xvec[0:t]$})\,.$$

    Note that prior work studying this problem under continual observation \cite{JainKRSS23, CummingsEMMOZ25} considered the situation where you only have one update per time step. We consider the more general model where you can have arbitrarily many updates per time step since it corresponds more directly to practical systems. 
    
    We will consider the following neighboring relation in our work. Privacy is defined with respect to this neighboring relation.
    \begin{definition}[Item Neighbors for $\countdistinct$]
        We say that two streams $\xvec$ and $\yvec$ for $\countdistinct$ are neighboring if there exists at most one element $u \in \univ$ such that replacing all updates corresponding to $u$ (i.e. $-u$ or $+u$) with $\bot$ in one of the streams results in the other.
    \end{definition}
    This is in (our opinion) a simpler definition than what has been used in prior work~\cite{JainKRSS23}. In that definition, neighboring corresponds to replacing a \textit{subset} of updates corresponding to $u$ with $\bot$. Semantically, we believe our neighboring notion captures the risks associated with revealing information about an element $u$ just as well as the neighboring relation studied in prior work. Additionally, we note that being neighboring in otheir notion corresponds to being $2$-neighbors in our neighboring notion, and so our results (with this extra factor of $2$) apply to their neighboring relation as well.
    \paragraph{ Releasing Graph Functions under Continual Observation:} We study the degree histogram problem and triangle count problems for fully dynamic graph streams under edge differential privacy. 

    A stream $\xvec$ for this problem can be thought of as a sequence of graphs $G_1,\dots,G_T$ all on the same $n$ nodes. Each element in the stream will correspond to a set of edge updates (corresponding to the insertion or deletion of an edge $e = (a,b)$ between two nodes $a$ and $b$ in the graph, or a $\bot$ operation indicating a no-op) that represent how the graph at time step $t$ is generated from the graph at time step $t-1$. We will use $\xvec[t][i]$ to represent the $i^{th}$ update at time step $t$. We will assume that $G_{-1}$ is the empty graph. Informally, we will use $+(a,b)$ to represent the insertion of an edge $(a,b)$ and $-(a,b)$ to represent the deletion of an edge $(a,b)$. We will also abuse notation to interchangeably refer to the stream as a sequence of edge updates or as a sequence of graphs depending on convenience. 

    Now, we are ready to formally define the degree estimation problem. $\degreecount$ is a function that takes as input a stream $\xvec$ (a sequence of $T$ graphs $G_0,\dots,G_{T-1}$) and outputs a stream of $T$ non-negative integer vectors each of size $n$. For $t \in [T]$,
    $$\degreecount(\xvec)[t][i] = \sum_{a \in [n]} \mathbbm{1}[\text{edge  }(i,a) \in G_t]\,.$$

    Next, we will define the triangle count problem. A triangle in a graph corresponds to the presence of a set of edges $\{(a,b),(b,c),(a,c) \}$ for three nodes $a,b,c$ in the graph. $\trianglecount$ is a function that takes as input a stream $\xvec$ (a sequence of $T$ graphs $G_0,\dots,G_{T-1}$) and outputs a stream of $T$ non-negative integers. For $t \in [T]$,
    $$\trianglecount(\xvec)[t] = \sum_{a,b,c \in [n]} \mathbbm{1}[a,b,c \text{ form a triangle in $G_t$}]\,.$$

    Next, we are ready to define the neighboring relation of interest, with respect to which privacy is defined.
    \begin{definition}[Edge Neighbors for Graph Problems]
        We say two streams $\xvec$ and $\yvec$ for $\trianglecount$ and $\degreecount$ are edge neighboring if there exists at most one edge $(a,b)$ such that replacing all updates corresponding to $(a,b)$ (i.e. $-(a,b)$ or $+(a,b)$) with $\bot$ in one of the streams results in the other.
    \end{definition}

    Note that prior work studying this problem under continual observation \cite{RaskhodnikovaS24} considered the situation where only one edge is updated per time step. We consider the case where an arbitary number of edges can be updated per time step.
 
    Note that we will be interested in the case where some parameter (for example, degree) of all graphs in the stream are bounded by a value $k$. We can define $k$-restricted neighbors as graph sequences that satisfy the above edge-neighboring requirement and additionally the restriction that all graphs in both graph sequences have maximum parameter bounded by $k$. $k$-restricted edge privacy can be defined with such a neighboring relation.
    

    Finally, we note that since $\degreecount$ is vector valued, our error notions are ill-defined for it. We will overload the error metrics as follows: we will define the root mean squared error for this problem as the max over coordinates of the root mean squared error over the outputs restricted to that coordinate, and the root max squared error similarly.

\section{Sensitivity Vector Sets for Fully Dynamic Streams}\label{sec:senssets}

In this section, we formally describe our sensitivity vector sets of interest, and explain how they relate to cardinality estimation problems. In specific, we will explain how error bounds for continual counting for our defined sensitivity vector sets correspond to error bounds for $\countdistinct$, $\degreecount$ and $\trianglecount$.

Common to the problems we consider will be that algorithms for them involve running continual counting mechanisms on difference streams. Formally, for a counting function $\propcount$, define difference streams $\diff{\propcount}(\xvec)[t] = \propcount(\xvec)[t] - \propcount(\xvec)[t-1]$. 

Firstly, in \Cref{sec:boundpref} we will characterize properties of sensitivity vectors for the difference streams for $\countdistinct, \degreecount$, and $\trianglecount$. Then, in \Cref{sec:boundedcontrib}, we will consider how bounding the `contribution' of elements in these streams adds additional structure to the sensitivity vector sets that we can leverage.




\subsection{Properties of Sensitivity Vectors for Cardinality Problems}\label{sec:boundpref}

Consider neighboring streams $\xvec,\yvec$ defined over a universe $\univ$ associated with a counting problem $\propcount$. The sensitivity vector associated with these neighboring streams is the subtraction of the corresponding difference streams:
    $$\vec{\Delta}_{\propcount}(\xvec,\yvec) = \diff{\propcount}(\xvec) - \diff{\propcount}(\yvec)\,.$$
The sensitivity vector set for $\propcount$ is the collection of all such vectors over neighboring stream pairs:
    $$\vec{\Delta}_{\propcount} = \{\vec{\Delta}_{\propcount}(\xvec,\yvec) \;|\; \xvec,\yvec\in\univ^T \text{ s.t. $\xvec$ \& $\yvec$ are neighbours} \}\,.$$\
We first focus on sensitivity vectors with values in $\{-1,0,1\}^T$, and define a vector as \emph{alternating} if its non-zero entries alternate in sign. For example, the vector $(1,0,0,-1,0,1)$ is alternating. In \Cref{thm:countdiffalt,thm:degdiffalt}, we prove that the sensitivity vector sets for $\countdistinct$ and $\degreecount$ are alternating.

We generalize this notion by defining the set $\boundedvecs{D}$ which contain all vectors with $D$-bounded interval sums,
    $$\boundedvecs{D} = \bigg\{\vec{s}\in\Z^T \;\bigg|\;\, \max_{i \le j \in [T]} \Big| \sum_{t=i}^j \vec{s}[t] \Big| \leq D \bigg\}\,.$$
For all cardinality estimation problems considered in this work ($\countdistinct,\degreecount$, and $\trianglecount$), we prove that their sensitivity vector sets are contained $\boundedvecs{D}$ for appropriately chosen $D$. In fact, when $D=1$, the set $\boundedvecs{1}$ is precisely the set of alternating vectors as proven in \Cref{claim:boundedvecs-are-alternating} below. 

\begin{claim}\label{claim:boundedvecs-are-alternating}
    $\vec{\Delta} \in \{-1,0,1\}^T\text{ is alternating }\;\iff\; \vec{\Delta} \in \boundedvecs{1}$
\end{claim}

\begin{proof}
Observe that if the sensitivity vector is alternating, then any interval sum is at most $1$, since adjacent $1$s and $-1$s in the interval will cancel each other out, leaving at most one excess uncanceled value. On the other hand,  if the sensitivity vector is not alternating, then there exists two non-zero values of the same sign with zeros between them and the interval sum corresponding to the indices of these non-zero values will have absolute value larger than $1$. Hence, this sensitivity vector will not be in $\boundedvecs{1}$. 
\end{proof}

\paragraph{\bf The Sensitivity Vector $\vec{\Delta}_{\countdistinct}(\xvec,\yvec) \in \boundedvecs{1}$ for All Neighboring Count Streams $\xvec,\yvec$.} Consider neighboring streams $\xvec$ and $\yvec$ for $\countdistinct$. We define the restricted stream $\xvec|_u$ to an item $u$ as follows.
    \begin{equation}\label{eq:rescountdist}
        \xvec|_u[t][i] = 
        \begin{cases}
        +u &\quad\text{if }\xvec[t][i] = +u\,,\\
        -u &\quad\text{if }\xvec[t][i] = -u\,,\\
        \perp &\quad\text{o/w}\,.
        \end{cases}
    \end{equation}



\begin{theorem}\label{thm:countdiffalt}\hspace{-6pt}
For all neighboring streams $\xvec,\yvec \hspace{-1.5pt}\in\hspace{-1.5pt} (\univ\cup \{\bot\})^T$\hspace{-2pt}, 
the sensitivity stream $\vec{\Delta}_{\countdistinct}(\xvec,\yvec)$ is alternating. (i.e., $\vec{\Delta}_{\countdistinct} \in \boundedvecs{1}$.)
\end{theorem}

\begin{proof}
Consider neighboring streams $\xvec$, $\yvec$ that differ in the data of some element $u$. Without loss of generality, we can assume that $u$ appears in $\xvec$ and not in $\yvec$.
First, note that $\vec{\Delta}_{\countdistinct}(\xvec,\yvec)$ depends only on the occurrences of $-u,u$ since those are the only inputs on which $\xvec$ and $\yvec$ differ. i.e., $$\vec{\Delta}_{\countdistinct}(\xvec,\yvec) = \countdiff(\xvec|_u) - \countdiff(\yvec|_u)\,.$$
Moreover, $\yvec|_u = \{\perp\}^T$ which means that $\countdiff(\yvec|_u) = \vec{0}$. This gives us the following:
\begin{equation*}
    \vec{\Delta}_{\countdistinct}(\xvec,\yvec) = \countdiff(\xvec|_u)\,.
\end{equation*}

We will now consider $\vec{\Delta}_{\countdistinct}(\xvec,\yvec)[t] = \countdiff(\xvec|_u)[t]$ and show that it is alternating. To this end, we construct \Cref{tab:distinct-sens} which shows the computation of $\countdiff(\xvec|_u)[t]$ based on $\countdistinct(\xvec|_u)[t]$ and $\countdistinct(\xvec|_u)[t-1]$.

\begin{table}[ht!]
\centering
\begin{tabular}{c@{\hskip 1em}c|c|c}
    & \multicolumn{3}{c}{$\countdistinct(\xvec|_u)[t]$} \\
    && $\mathbf{0}$ & $\mathbf{1}$ \\
  \cline{2-4}
  \multirow{2}{*}{$\countdistinct(\xvec|_u)[t-1]$} & $\mathbf{0}$ & $\qquad0\qquad$ & $\qquad1\qquad$ \\
  \cline{2-4}
  & $\mathbf{1}$ & $-1$ & $0$ \\
\end{tabular}
\caption{Computing $\countdiff(\xvec|_u)[t]$ based on $\countdistinct(\xvec|_u)$.}\label{tab:distinct-sens}
\end{table}

We can see from \Cref{tab:distinct-sens} that $\countdiff(\xvec|_u)[t] = 1$ only if $\countdistinct(\xvec|_u)[t] = 1$. In future time-steps, this count can either stay the same or reduce to $0$. While the count remains the same, $\countdiff(\xvec|_u)$ will be $0$, and when the count reduces to $0$ at time $t^*$, $\countdiff(\xvec|_u)[t^*]$ will be $-1$. 
Similarly, $\countdiff(\xvec|_u)[t] = -1$ only if $\countdistinct(\xvec|_u)[t] = 0$. In future time-steps, this count can either stay the same or increment to $1$. While the count remains the same, $\countdiff(\xvec|_u)$ will be $0$, and when the count increments to $1$ at time $t^*$, $\countdiff(\xvec|_u)[t^*]$ will be $1$. 
\end{proof}

\paragraph{\bf Characterizing Sensitivity Vectors for $\degreecount$.} Consider edge-neighboring streams $\xvec$ and $\yvec$ for $\degreecount$. We define the restricted stream $\xvec|_{(u,v)}$ to an edge $(u,v)$ as follows.
    \begin{equation*}
        \xvec_{(u,v)}[t][i] = 
        \begin{cases}
         \xvec[t][i] &\quad\text{if }\xvec[t][i] = \pm(u,v)\,,\\
        \perp &\quad\text{o/w}\,.
        \end{cases}
    \end{equation*}

\begin{theorem}\label{thm:degdiffalt}
For all edge-neighboring graph streams $\xvec,\yvec$, there are two nodes $u,v$ such that 
\begin{enumerate}
  \item $\vec{\Delta}_{\degreecount}[\cdot][u] \in \boundedvecs{1}$, $\vec{\Delta}_{\degreecount}[\cdot][v] \in \boundedvecs{1}$.
    \item For all nodes $c \neq u,v$, we have that $\vec{\Delta}_{\degreecount}[\cdot][c] = 0^T$.
\end{enumerate}

\end{theorem}

\begin{proof}
Consider edge-neighboring streams $\xvec$, $\yvec$ that differ in the data of some edge $(u,v)$. Without loss of generality, we can assume that $(u,v)$ appears in $\xvec$ and not in $\yvec$.

First, note that the non-zero values of $\vec{\Delta}_{\degreecount}(\xvec,\yvec)$ depend only on the occurrences of $-(u,v),+(u,v)$. This is because the edge $u,v$ only impacts the degree of nodes $u$ and $v$, and its impact on the degree of those nodes is independent of the state of the rest of the graph. Therefore 
$$\vec{\Delta}_{\degreecount}(\xvec,\yvec) = \degreediff(\xvec|_{(u,v)}) - \degreediff(\yvec|_{(u,v)})\,.$$
Moreover, $\yvec|_{(u,v)} = \big\{\vec{\perp}\big\}^T$ which means that $\degreediff(\yvec|_{(u,v)}) = \vec{0}$. This gives us the following:
\begin{equation}
    \vec{\Delta}_{\degreecount}(\xvec,\yvec) = \degreediff(\xvec|_{(u,v)})\,.
\end{equation}

Since the edge $(u,v)$ can effect only the degrees of the vertices $u$ and $v$ that it comes in contact with, we have that $\degreediff(\xvec|_{(u,v)})[t][w] = 0$ for all $w \neq u,v$. This proves the second item of the theorem.
Additionally, since the degrees of $u$ and $v$ will be impacted in the same way by updates to edge $(u,v)$, we have that
    $$\degreediff(\xvec|_{(u,v)})[t](u) = \degreediff(\xvec|_{(u,v)})[t](v)\,.$$
We now focus on the stream of updates $\degreediff(\xvec|_{(u,v)})[t][u]$ and show that it must be alternating. First, if $\degreediff(\xvec|_{(u,v)})[t][u] = 1$, we show that next non-zero entry in this stream of updates 
must be $-1$. We can then show the converse by a similar argument.

\begin{table}[h!]
\centering
\begin{tabular}{c@{\hskip 1em}c|c|c}
    & \multicolumn{3}{c}{$\degreecount(\xvec|_{(u,v)})[t][u] $} \\
    && $\mathbf{0}$ & $\mathbf{1}$ \\
  \cline{2-4}
  \multirow{2}{*}{$\degreecount(\xvec|_{(u,v)})[t-1][u] $} & $\mathbf{0}$ & $\qquad0\qquad$ & $\qquad1\qquad$ \\
  \cline{2-4}
  & $\mathbf{1}$ & $-1$ & $0$ \\
\end{tabular}
\caption{Computing $\degreediff(\xvec|_{(u,v)})[t][u]$ based on $\degreecount(\xvec|_{(u,v)})$.}\label{tab:degree-sens}
\end{table}

Note that $\degreediff(\xvec|_{(u,v)})[t][u] = 1$ only if $\degreecount(\xvec|_{(u,v)})[t][u]~=~1$ (see~\Cref{tab:degree-sens}). In future time-steps, this count can either stay the same or reduce to $0$. While $\degreecount$ remains the same, $\degreediff$ will be $0$, and when the $\degreecount$ reduces to $0$, $\degreediff$ will be $-1$. 

Similarly, $\degreediff(\xvec|_{(u,v)})[t][u] = -1$ only if $\degreecount(\xvec|_{(u,v)})[t][u] = 0$. In future time-steps, this count can either stay the same or increase to $1$. While the $\degreecount$ remains the same, $\degreediff$ will be $0$, and when the count increases to $1$, $\degreediff$ will be $1$. 
This completes the proof of the first part of the theorem.
\end{proof}

\paragraph{\bf The Sensitivity Vector $\vec{\Delta}_{\trianglecount}(\xvec,\yvec) \in \boundedvecs{D}$ for All Neighboring Graph Streams $\xvec,\yvec$ with Degree Bound $D$.} Consider neighboring streams $\xvec$ and $\yvec$ for $\trianglecount$ with degree bound $D$. We define the restricted stream $\xvec|_{(a,b)}$ to an edge $(a,b)$ as follows.
    \begin{equation*}
        \xvec_{(a,b)}[t][i] = 
        \begin{cases}
        \{ \xvec[t][i] &\quad\text{if }\xvec[t][i] \in \big\{ (a,\cdot), (b,\cdot), (\cdot,a), (\cdot,b)\big\}\,,\\
        \perp &\quad\text{o/w}\,.
        \end{cases}
    \end{equation*}

\begin{theorem}\label{thm:triangle-count-sensitivity}
Fix neighboring graph streams $\xvec,\yvec$ with degree bound $D$. The sensitivity stream $\vec{\Delta}_{\trianglecount}(\xvec,\yvec)\in \boundedvecs{D}$.
\end{theorem}

\begin{proof}
Consider neighboring graphs $\xvec$, $\yvec$ that differ in edge $(a,b)$. Without loss of generality, we can assume that $(a,b)$ appears in $\xvec$ and not in $\yvec$. We prove \Cref{thm:triangle-count-sensitivity} using the following three claims.

\begin{claim}\label{claim:sensvectc}
    $\vec{\Delta}_{\trianglecount}(\xvec,\yvec) = \trianglediff(\xvec) - \trianglediff(\yvec) = \trianglediff\big(\xvec|_{(a,b)}\big) - \trianglediff\big(\yvec|_{(a,b)}\big) = \trianglediff\big(\xvec|_{(a,b)}\big)$.
\end{claim}

This follows from the fact that $\trianglediff(\xvec) - \trianglediff(\yvec)$ depends only on the triangles that contain edge $(a,b)$, and the fact that such triangles only exist in $\xvec$ (i.e., $\trianglediff(\yvec|_{(a,b)}) = \vec{0}$).

\begin{claim}
    Then, for all $i<j \in [T]$:
$$\sum_{t=i}^{j} \trianglediff(\xvec|_{(a,b)})[t] = \trianglecount(\xvec|_{(a,b)})[j] - \trianglecount(\xvec|_{(a,b)})[i-1]\,.$$
\end{claim}

This follows from the fact that for all $i<j \in [T]$:
\begin{align*}
    &\sum_{t=i}^{j} \trianglediff(\xvec|_{(a,b)})[i] 
    = \sum_{t=i}^{j} - \trianglecount(\xvec|_{(a,b)})[t-1] + \trianglecount(\xvec|_{(a,b)})[t]\\
    = &- \trianglecount(\xvec|_{(a,b)})[i-1] + \cancel{\trianglecount(\xvec|_{(a,b)})[i}] - \cancel{\trianglecount(\xvec|_{(a,b)})[i}]\dots\\
    & + \cancel{\trianglecount(\xvec|_{(a,b)})[j-1}] - \cancel{\trianglecount(\xvec|_{(a,b)})[j-1}] + \trianglecount(\xvec|_{(a,b)})[j]\,.
\end{align*}

The following claim will complete this proof.
\begin{claim}
    Then, for all $i<j \in [T]$:
$$\left|\trianglecount[j] - \trianglecount[i-1]\right| \leq D\,.$$
\end{claim}
This follows directly from the fact that $\trianglecount[j],\trianglecount[i-1] \in \{0,\dots,D\}$.
\end{proof}

\subsection{Sensitivity Vector Sets under Bounded Contribution}\label{sec:boundedcontrib}

In this section, we characterize the effects of bounding the change in contribution of specific elements or items to the sensitivity vector sets.

Let $\boundedvecs{D,k} \subseteq \boundedvecs{D}$ be the set of vectors with the additional guarantee of $k$-bounded $\ell_1$-norms:
    $$\boundedvecs{D,k} = \big\{\vec{s}\in\Z^T \;\big|\;\, \vec{s}\in S_D,
\; \|\vec{s}\|_1 \leq k \big\}\,.$$

We will prove that the notion of bounded maximum flippancy for counting distinct elements introduced in \cite{JainKRSS23} corresponds to restricting the sensitivity vector set to a subset of $\boundedvecs{1,k}$ where $k$ is a bound on maximum flippancy. For degree histograms, and triangle counting, we will suggest new notions that capture bounded contribution that we term \textit{maximum degree contribution}, and \textit{maximum triangle contribution} respectively, and argue that it has a similar effect of restricting the sensitivity vector set to a subset of $\boundedvecs{D,k}$ where $k$ is a bound on the respective parameter.

{\bf Bounded Contribution Changes for $\countdistinct$ via Maximum Flippancy:} 
Observe that for arbitrary neighboring streams $\xvec, \yvec$ for $\countdistinct$, the $\ell_1$ norm of the sensitivity vector could be as large as $T$ ($\xvec = \{+u, -u, +u, -u, \dots,+u, -u \}, \yvec = \{ \bot, \bot, \dots,\bot\}$). However, such a stream $\xvec$ seems unrealistic. For this purpose, prior work \cite{JainKRSS23} introduced the concept of maximum flippancy  to capture what practical streams looked like, and to prove that for such streams, one could achieve better error bounds.

We start by recalling the definition of maximum flippancy from \cite{JainKRSS23}. The flippancy $\mathsf{flip}(u)$ of an item $u \in \univ$ in a stream $\xvec$ for $\countdistinct$ is defined as the number of time steps the item $u$ goes from present to absent or absent to present in the stream $\xvec$. This corresponds to the number of time steps there is a change in the contribution of the item to the $\countdistinct$ value for $\xvec$. The maximum flippancy is defined as $\max_{u \in \univ} \mathsf{flip}(u)$.

\begin{theorem}\label{thm:countdist-flip}
    Let $\xvec, \yvec$ be neighboring streams for $\countdistinct$ with maximum flippancy at most $k$. Then, $\vec{\Delta}_{\countdistinct}(\xvec, \yvec) \in \boundedvecs{1,k}$.
\end{theorem}
\begin{proof}
Consider a stream $\xvec$ with maximum flippancy bounded by $k$, and consider a neighboring stream $\yvec$ with maximum flippancy bounded by $k$ as well. Without loss of generality, assume $\xvec$ has updates corresponding to item $u \in \univ$ and $\yvec$ does not. 

Consider the restricted stream $\xvec|_u$ restricted to item $u$ (see \Cref{eq:rescountdist} for a reminder of the definition). Observe that $\countdiff(\xvec|_u)$ has at most $k$ non-zero entries. On the other hand $\countdiff(\yvec|_u)$ is $0^T$. Hence, the sensitivity vector $\vec{\Delta}_{\countdistinct}(\xvec,\yvec)$ is $k$-sparse. 

Since every value in $\vec{\Delta}_{\countdistinct}(\xvec,\yvec)$ is in $\{ 1,0,-1 \}$, this further implies that the $\ell_1$ norm of  $\vec{\Delta}_{\countdistinct}(\xvec,\yvec)$ is bounded by $k$. 
This combined with \Cref{thm:countdiffalt} ($\vec{\Delta}_{\countdistinct}(\xvec,\yvec) \in \boundedvecs{1}$) gives us that $\vec{\Delta}_{\countdistinct}(\xvec,\yvec) \in \boundedvecs{1,k}$ as required. 
\end{proof}

{\bf Bounded Contribution Changes for $\degreecount$ via Maximum Degree Contribution:} 
In this paper, we give a new definition of the degree contribution of an edge $(a,b)$ in a graph sequence over $n$ nodes. Note that with no further assumptions, the sensitivity vector restricted to the indices corresponding to node $a$ or node $b$ could have $\ell_1$ norm as large as $T$. To see this consider $\xvec$ consisting of alternating insertions and deletions of the edge $(a,b)$, and a neighboring stream $\yvec$ being $\bot^T$. We define the maximum degree contribution of a stream to avoid such pathological cases.

\begin{definition}[Maximum Degree Contribution]\label{def:degcont}
    The degree contribution of an edge $(a,b)$ in a stream $\xvec$ is defined as the number of times it is updated throughout the stream $\xvec$. The maximum degree contribution of a stream is the maximum degree contribution over its edges.
\end{definition}

We next argue that a bound on the maximum degree contribution of a stream gives further structure to the sensitivity vectors.

\begin{theorem}\label{thm:degree-contrib}
    Let $\xvec$ and $\yvec$ be neighboring streams for $\degreecount$ with maximum degree contribution bounded by $k$. Then, there are two nodes $a,b$ such that:
    \begin{enumerate}
        \item $\vec{\Delta}_{\degreecount}[\cdot][a] \in \boundedvecs{1,k}$, $\vec{\Delta}_{\degreecount}[\cdot][b] \in \boundedvecs{1,k}$.
        \item For all nodes $c \neq a,b$, we have that  $\vec{\Delta}_{\degreecount}[\cdot][c] = 0^T$.
    \end{enumerate}
\end{theorem}
\begin{proof}
WLOG let $\xvec$ contain updates corresponding to edge $(a,b)$ and $\yvec$ contain no updates corresponding to edge $(a,b)$.
Item 2 follows directly from \Cref{thm:degdiffalt}.
For Item 1, for nodes $a$ and $b$, we have by Item 1 of \Cref{thm:degdiffalt}, that $\vec{\Delta}_{\degreecount}[\cdot][a] \in \boundedvecs{1}$ and $\vec{\Delta}_{\degreecount}[\cdot][b] \in \boundedvecs{1}$. 

Finally, observe that if $\vec{\Delta}_{\degreecount}[t][a] = \diff{\degreecount}(\xvec)[t][a] - \diff{\degreecount}(\yvec)[t][a] \neq 0$, then there is an update corresponding to node $a$ in $\xvec$ and no corresponding update in $\yvec$. This implies that there is an update corresponding to edge $(a,b)$. Hence, since the maximum degree contribution is bounded by $k$, the number of non-zero entries in $\vec{\Delta}_{\degreecount}[\cdot][a]$ is bounded by $k$. A similar argument applies for node $b$ as well. This implies that $\vec{\Delta}_{\degreecount}[\cdot][a] \in \boundedvecs{1,k}$ and $\vec{\Delta}_{\degreecount}[\cdot][b] \in \boundedvecs{1,k}$, as required.
\end{proof}

{\bf Bounded Contribution Changes for $\trianglecount$ via Maximum Triangle Contribution:} 
In this paper, we give a new definition of the triangle contribution of an edge $(a,b)$ in a graph sequence over $n$ nodes with maximum degree $D$. We start by motivating the definition. Firstly, note that with no further assumptions, there are neighboring graph sequences $\xvec, \yvec$ such that $\| \vec{\Delta}_{\trianglecount}(\xvec,\yvec) \|_1 = DT$. \footnote{Consider for example that $\xvec$ corresponds to adding an edge $(a,b)$ at the first time step and then adding all other edges $(a,\dot), (b,\dot)$ at even time steps and deleting all edges of the form $(a,\dot), (b,\dot)$ (other than edge $(a,b)$) at odd time steps that are not the first time step. $\yvec$ is identical to $\xvec$ other than the fact that at the first time step its update is $\bot$.}

However, we would expect that for realistic streams $\xvec$, we would not observe such an extreme case. We define maximum triangle contribution to capture this.

{\bf First attempt at a definition:} Similarly to flippancy, we could attempt to define triangle contribution of an edge as the number of time steps that there is a change in the number of triangles including this edge $(a,b)$. One subtlety in such a definition is that such a change in number of triangles could occur even for updates that \textit{do not involve the edge $(a,b)$}. Once triangle contribution is defined, maximum triangle contribution could be defined analogously to maximum flippancy by taking a maximum over edges. 

Consider two neighboring streams $\xvec$ and $\yvec$ for $\trianglecount$, such that $\xvec$ contains updates involving edge $(a,b)$ and $\yvec$ does not. Observe that by \Cref{claim:sensvectc}, we have that $\vec{\Delta}_{\trianglecount}(\xvec,\yvec) = \trianglediff(\xvec|_{(a,b)})$. 
Now, since the maximum degree of the graph is $D$, observe that at most $D$ triangles involving edge $(a,b)$ could be affected at a time step where there is a change in the number of triangles including $(a,b)$. Additionally, observe by definition that $\vec{\Delta}_{\trianglecount}(\xvec, \yvec)$ is $k$-sparse. Hence, if the maximum triangle contribution is bounded by $k$, we have that $\| \vec{\Delta}_{\trianglecount}(\xvec, \yvec) \|_1 \leq kD$. 

{\bf Our  definition:} In the above argument, when we bound the $\ell_1$ norm, we have to always consider the worst case where $D$ triangles are affected each time there is a change in the number of triangles involving edge $(a,b)$. To avoid considering this worst case, we come up with a different definition that directly bounds the $\ell_1$ norm.

\begin{definition}[Maximum Triangle Contribution]\label{def:maxtrianglecontrib}
    At any time step $t$, consider the change in number of triangles including an edge $(a,b)$, i.e. $\trianglediff(\xvec|_{(a,b)})[t]$. Then, the triangle contribution of edge $(a,b)$ is defined as the sum of the absolute values of these changes over all time steps i.e. $\sum_{t \in [T]} |\trianglediff(\xvec|_{(a,b)})[t]|$. The maximum triangle contribution is defined as the maximum over edges of the triangle contributions of the edges.
\end{definition}

Now, we can directly bound the $\ell_1$ norm of the sensitivity vector.

\begin{theorem}\label{thm:triangle-count-bounded}
Consider neighboring graph streams $\xvec,\yvec$ with degree bound $D$ and maximum triangle contribution bounded by $k$. The sensitivity stream $\vec{\Delta}_{\trianglecount}(\xvec,\yvec)\in \boundedvecs{D,k}$.
\end{theorem}
\begin{proof}
Consider two neighboring streams $\xvec$ and $\yvec$ for $\trianglecount$, such that $\xvec$ contains updates involving edge $(a,b)$ and $\yvec$ does not.
Using \Cref{thm:triangle-count-sensitivity}, we have that $\vec{\Delta}_{\trianglecount}(\xvec,\yvec)\in S_{D}$. 

Additionally, note that by \Cref{claim:sensvectc}, we have that $\vec{\Delta}_{\trianglecount}(\xvec,\yvec) = \trianglediff(\xvec|_{(a,b)})$. This implies that 
    $$\| \vec{\Delta}_{\trianglecount}(\xvec,\yvec) \|_1 = \| \trianglediff(\xvec|_{(a,b)}) \|_1.$$

By \Cref{def:maxtrianglecontrib} for maximum triangle contribution, we have $\| \trianglediff(\xvec|_{(a,b)}) \|_1 \leq k$. This completes the proof.
\end{proof}

\subsection{Error bounds for Cardinality Estimation from Error bounds for Continual Counting}\label{sec:errorsensvec}

Consider streams $\vec{s}, \vec{t}$ with integer values. We say that these two streams are $\boundedvecs{}$-neighboring if $\vec{s} - \vec{t} \in \boundedvecs{}$. We say that a mechanism $M$ is a $\rho$-zCDP continual counting mechanism on $\boundedvecs{}$-streams if it is $\rho$-zCDP with respect to the $\boundedvecs{}$-neighboring relation. 

The main algorithm we use for all our problems of interest ($\countdistinct$, $\degreecount$, $\trianglecount$) uses the same idea to reduce to continual counting. From the original stream $\xvec$, we will construct the difference stream $\diff{\propcount}(\xvec)$ for the cardinality estimation problem $\propcount$, and then use an appropriate differentially private continual counting mechanism on $\boundedvecs{}$-streams where $\boundedvecs{}$ is the set of sensitivity vectors associated with the problem $\propcount$. See \Cref{alg:redcontcount} for a formal description.

\begin{algorithm}[h!]
    \caption{Algorithm $\cA$ reducing from $\propcount$ to continual counting mechanism $M$}
    \label{alg:redcontcount}
    \begin{algorithmic}[1]
        \Statex \textbf{Input: } cardinality estimation problem $\propcount$,  input stream $\xvec$ for $\propcount$, time horizon $T$, continual counting mechanism $M$.
        \Statex \textbf{Output:} Stream of outputs $o_1,\dots,o_T \in \mathbb{R}^T$
        \State Let $\xvec[-1]= \bot$.
        \For{$t \in [T]:$}
        \State Let $\diff{\propcount}(\xvec)[t] = \propcount(\xvec[-1:t]) -\propcount(\xvec[-1:t-1])$.
        \State Send $\diff{\propcount}(\xvec)[t]$ to mechanism $M$ and output the result $o_t$.
        \EndFor
    \end{algorithmic}
\end{algorithm}

We prove the following theorem connecting the privacy and accuracy of the continual counting mechanism $M$ to that of the associated algorithm for $\propcount$. We will be interested in the root mean (expected) squared error and root max (expected) squared error of mechanisms $\cM$ for the problem $\propcount$ which we will represent by $\maxse(\cM, T)$ and $\meanse(\cM, T)$. For $\degreecount$, we will have vector valued outputs at each time step (and the function is also vector valued), and we will be interested in the maximum root mean (expected) squared error and maximum root max (expected) squared error (where the maximum in both cases is over the indices of the vector). We will overload the notation for $\maxse(\cM, T)$ and $\meanse(\cM, T)$ to mean this when relevant.

We will often be interested in this error for a certain set of inputs and will make this explicit in our statements.

\begin{theorem}\label{thm:rederror}
Let $\boundedvecs{}$ be the set of sensitivity vectors associated with a cardinality estimation problem $\propcount$. Let $\cA$ refer to \Cref{alg:redcontcount}.
\begin{enumerate}
    \item For any $\rho>0$, $\cA$ is $\rho$-zCDP if $M$ is $\rho$-zCDP with respect to the $\boundedvecs{}$-neighboring relation.
    \item Let $M$ be a factorization mechanism for continual counting using factorization matrices $L,R$ for the sensitivity vector set $\boundedvecs{}$. Then,
    $$\maxse(\cA, T) = \maxse(L,R,\boundedvecs{})\,,$$
    $$\meanse(\cA, T) = \meanse(L,R,\boundedvecs{})\,.$$
\end{enumerate}
\end{theorem}
\begin{proof}
    For item 1 of the theorem, consider streams $\xvec, \yvec$ for the problem $\propcount$ that are neighboring. Then, by definition of the sensitivity vector set, we have that $\diff{\propcount}(\xvec) - \diff{\propcount}(\yvec) \in \boundedvecs{}$. Observe that when \Cref{alg:redcontcount} is run on streams $\xvec, \yvec$, the continual counting mechanism $M$ would be run on the difference streams $\diff{\propcount}(\xvec)$, and $\diff{\propcount}(\yvec)$ respectively and hence, since $M$ is $\rho$-zCDP with respect to the $\boundedvecs{}$-neighboring relation, we have that \Cref{alg:redcontcount} is $\rho$-zCDP.

    We now prove item 2 of the theorem (for root max squared error, the same argument applies for root mean squared error). Observe that for all $t \in [T]$, we have that $\sum_{i=0}^t \diff{\propcount}(\xvec)[t] = \propcount(\xvec)[t]$. Hence, since the outputs of $\cA$ and continual counting mechanism $M$ are identical (in \Cref{alg:redcontcount}), the root max squared error of $\cA$ is exactly that of mechanism $M$, which is $\maxse(L,R,S)$ by definition.
\end{proof}

We now explain the consequences for our problems of interest by combining the previous subsections characterizing the sensitivity vector sets with \Cref{thm:rederror}.

    {\bf Error Bounds for $\countdistinct$:} We note that a direct application of \Cref{thm:rederror}, combined with \Cref{thm:countdist-flip} (characterizing sensitivity vector sets for $\countdistinct$) would give privacy  for neighboring streams that both had flippancy at most $k$, but would not necessarily give privacy for all streams.

    However, we note that this can be extended to give privacy for all streams via a technique of \cite{JainKRSS23}. We recap this technique below. We preprocess the stream to ensure that the processed stream has flippancy bound at most $k$. The first step is to maintain a counter tracking the flippancy of all items in the stream. For any update in $\xvec$ that causes the flippancy of an item to cross the threshold $k$, we simply replace it with $\bot$. Call the resulting stream $\mathsf{trunc}(\xvec)$. Observe that $\mathsf{trunc}(\xvec)$ has maximum flippancy at most $k$. Additionally, if $\xvec$ itself has maximum flippancy at most $k$, then $\mathsf{trunc}(\xvec) = \xvec$. Now, we can use this preprocessing step to ensure privacy for all streams while still getting the same accuracy for streams with flippancy at most $k$.
    \begin{corollary}\label{cor:err-countdist}
    Fix $k>0$. Let $M$ be a factorization mechanism for continual counting using factorization matrices $L,R$ for the sensitivity vector set $\boundedvecs{1,k}$. Let $\cA$ reference \Cref{alg:redcontcount}. 
    We have that:
    \begin{itemize}
        \item For any $\rho>0$, $\cA(\mathsf{trunc}(\cdot))$is $\rho$-zCDP if mechanism $M$ is $\rho$-zCDP with respect to the $\boundedvecs{1,k}$-neighboring relation.
        \item For streams for $\countdistinct$ with maximum flippancy at most $k$, we have that 
        $$\maxse(\cA(\mathsf{trunc}(\cdot)), T) = \maxse(L,R,\boundedvecs{1,k})\,,$$
    $$\meanse(\cA(\mathsf{trunc}(\cdot)), T) = \meanse(L,R,\boundedvecs{1,k})\,.$$
    \end{itemize}
    \end{corollary}
We further note that exactly as in \cite{JainKRSS23}, once we have a mechanism that is accurate for a given flippancy threshold but private for all streams, we can extend to the case where we do not know a flippancy bound in advance via a careful use of the sparse vector technique. This gives accuracy bounds that \textit{adapt} to the maximum flippancy of the stream. We refer the readers to Section 3.2 of \cite{JainKRSS23} for details.

{\bf Error Bounds for $\degreecount$:} To run \Cref{alg:redcontcount} for $\degreecount$, we can run $n$ parallel copies, one corresponding to the degree of each node. Observe that if run directly, we would only get privacy when restricted to streams with bounded maximum degree contribution. To extend this to get privacy for all streams, we implement a preprocessing step similar to that for $\countdistinct$. Specifically, we track the degree contribution of all edges, and if the degree contribution of an edge crosses the bound $k$, then we replace all future updates corresponding to this edge with $\bot$. Call the resulting stream $\mathsf{trunc}(\xvec)$. Observe that $\mathsf{trunc}(\xvec)$ has maximum degree contribution at most $k$. Additionally, if $\xvec$ itself has maximum degree contribution at most $k$, then $\mathsf{trunc}(\xvec) = \xvec$.

Combining \Cref{thm:rederror} with \Cref{thm:degree-contrib} (and by the discussion above), this then gives the following result  (where we have used the fact that an edge can only affect the degrees of its endpoints, and additionally used composition to deal with $2$ nodes having sensitivity vectors in $S_{1,k}$, with the others identically being $0^T$.). 

\begin{corollary}\label{cor:err-degcount}
    Fix $k>0$. Let $M$ be a factorization mechanism for continual counting using factorization matrices $L,R$ for the sensitivity vector set $\boundedvecs{1,k}$. Let $\xvec$ be a stream corresponding to $\degreecount$, and let $\cA$ reference \Cref{alg:redcontcount} being run in parallel $n$ times, one to estimate the degree of each node of the graph. 
    We have that:
    \begin{itemize}
        \item For any $\rho>0$, $\cA(\mathsf{trunc}(\cdot))$ is $2\rho$-zCDP if mechanism $M$ is $\rho$-zCDP with respect to the $\boundedvecs{1,k}$-neighboring relation.
        \item For streams for $\degreecount$ with maximum degree contribution at most $k$, we have that 
        $$\maxse(\cA(\mathsf{trunc}(\cdot)), T) = \maxse(L,R,\boundedvecs{1,k})\,,$$
    $$\meanse(\cA(\mathsf{trunc}(\cdot)), T) = \meanse(L,R,\boundedvecs{1,k})\,.$$
    \end{itemize}
    \end{corollary}

{\bf Error Bounds for $\trianglecount$:} A direct application of \Cref{thm:rederror}, combined with \Cref{thm:triangle-count-bounded} (characterizing sensitivity vector sets for $\trianglecount$) gives the following result.
\begin{corollary}\label{cor:err-triang-count}
    Fix $k, D >0$. Let $M$ be a factorization mechanism for continual counting using factorization matrices $L,R$ for the sensitivity vector set $\boundedvecs{D,k}$. Let $\cA$ reference \Cref{alg:redcontcount}. Let streams $\xvec, \yvec$ for $\trianglecount$ be $(D,k)$-neighboring if every graph in both has maximum triangle contribution bounded by $k$ and degree bounded $D$.
    We have that:
    \begin{itemize}
        \item For any $\rho>0$, $\cA$ is $\rho$-zCDP with respect to the $(D,k)$-neighboring relation if mechanism $M$ is $\rho$-zCDP with respect to the $\boundedvecs{D,k}$-neighboring relation.
        \item For streams for $\trianglecount$ with maximum triangle contribution at most $k$, and maximum degree at most $D$, we have that 
        $$\maxse(\cA, T) = \maxse(L,R,\boundedvecs{D,k})\,,$$
    $$\meanse(\cA, T) = \meanse(L,R,\boundedvecs{D,k})\,.$$
    \end{itemize}
    \end{corollary}

A drawback of the above corollary is that privacy is only with respect to $(D,k)$-neighboring graph streams. Unlike the previous cases, it is unclear how to extend this algorithm to give an algorithm that is private for all graphs. The key difficulty is that ignoring updates to an edge after the maximum triangle contribution for that edge reaches its bound could have a cascading effect on other edges, since it will affect the triangle contributions of those edges too. This could result in neighboring graphs becoming vastly different after a truncation operation. This is the key question studied in the work of \cite{JainSW24} who design \textit{stable projections} for exactly this purpose. However, their techniques are only for insertion-only streams and do not extend readily to the fully dynamic setting. We leave it as a very intriguing question to extend our algorithm for $\trianglecount$ to be private for all graphs. 

We once again emphasize that existing bounds for $\trianglecount$ in our setting scale polynomially with $T$. Our parametrization allows us to use factorization mechanisms to carry out continual counting over the difference stream and achieve significantly better accuracy of the form $\sqrt{kD} \log T$. Hence, we view our improved accuracy guarantees for this problem (even for a restricted privacy notion) to be a valuable contribution.



\section{A Reduction from \texorpdfstring{$\boundedvecs{D, k}$}{Sdk} to \texorpdfstring{$\boundedvecs{1, k}$}{S1k}}\label{sec:trees-reduction}

To recap, let $\boundedvecs{D}$ be the set of vectors $D$-bounded partial sums,
    $$\boundedvecs{D} = \bigg\{\vec{s}\in\Z^T \;\bigg|\;\, \max_{i \le j \in [T]} \Big| \sum_{t=i}^j \vec{s}[t] \Big| \leq D \bigg\}\,,$$
and let $\boundedvecs{D,k} \subseteq \boundedvecs{D}$ be the set of vectors with the additional guarantee of $k$-bounded $\ell_1$-norms,
    $$\boundedvecs{D,k} = \big\{\vec{s}\in\Z^T \;\big|\;\, \vec{s}\in \boundedvecs{D},\; \|\vec{s}\|_1 \leq k \big\}\,.$$
We have shown how mechanism design for cardinality estimation problems, parameterized in $D$ and $k$, can be reduced to the design of continual counting mechanisms, but for streams with sensitivity vector set $\boundedvecs{D, k}$.
In particular, our focus is on \emph{factorization mechanisms} in this space, for which the sensitivity analysis now is non-trivial.
Bluntly, for reasonable choices of matrix $R\in\mathbb{R}^{m\times T}$:
\begin{equation*}
    \textit{how do we compute }\sens_2(R, \boundedvecs{D, k}) = \max_{\vec{\Delta}\in \boundedvecs{D, k}} \| R\vec{\Delta}\|_{2}\,\textit{?}
\end{equation*}
We note that, for standard continual counting of bits where two streams are neighboring if they differ at a single position, the corresponding sensitivity vector set is $\boundedvecs{1,1}$.
Here the answer is simple:
\begin{equation*}
    \sens_2(R, \boundedvecs{1,1}) = \| R \|_{1\to 2}\,,
\end{equation*}
i.e., the $\ell_2$ sensitivity equals the largest $\ell_2$ norm of any column in $R$.

To answer the question for other values of $k$ and $D$, we will show a reduction.
In this section, we show that it often suffices to study $\boundedvecs{1, k}$ streams, that is, $k$-sparse streams from $\{0,\pm 1\}^{T}$ with alternating non-zero entries.
We can tightly estimate the sensitivity of $R$ for $\boundedvecs{D, k}$ streams, by studying the simpler and more structured $\boundedvecs{1,k}$ streams.
We begin by showing a lower bound in this direction.
\begin{lemma}\label{lem:alternating-to-general-lower}
    Let $1 \leq k \leq T$ and $D,p\geq 1$ where $\boundedvecs{D, k}\subset\mathbb{R}^{T}$ and $R\in\mathbb{R}^{m\times T}$ for $m\geq 1$.
    Then 
    \begin{equation*}
        \sens_p(R, \boundedvecs{D, k}) \geq D \cdot \sens_p(R, \boundedvecs{1, \lfloor k/D \rfloor})\,.
    \end{equation*}
\end{lemma}
\begin{proof}
    Consider $\vec{\Delta}^{*}\in\boundedvecs{1, \lfloor k / D\rfloor}$ attaining the maximum norm, i.e.,
    \begin{equation*}
        \sens_p(R, \boundedvecs{1, \lfloor k / D\rfloor})
        = \max_{\vec{\Delta}\in\boundedvecs{1, \lfloor k / D\rfloor}} \| R\vec{\Delta} \|_p = \|R\vec{\Delta}^{*} \|_p\,.
    \end{equation*}
    Noting that $\vec{\Delta}' = D\vec{\Delta}^{*} \in \boundedvecs{D, D\lfloor k/D\rfloor} \subseteq \boundedvecs{D, k}$, we can deduce
    \begin{equation*}
        \sens_p(R, \boundedvecs{D, k})
        \geq \sens_p(R, \boundedvecs{D, D\lfloor k/D\rfloor})
        \geq \| R\vec{\Delta}' \|_p
        = D\| R\vec{\Delta}^* \|_p
        = D\cdot\sens_p(R, \boundedvecs{1, \lfloor k / D \rfloor})\,,
    \end{equation*}
    finishing the proof.
\end{proof}

We will next show that a similar bound holds in the other direction.
The following lemma will be crucial.
\begin{lemma}\label{lem:decompose-vectorset}
    Let $T\geq k \geq 1$ and $D\geq 1$ where $\boundedvecs{D, k}\subset\mathbb{R}^{T}$.
    Then for any $\vec{\Delta}\in\boundedvecs{D, k}$,
    \begin{align*}
        \qquad\exists k_1,\dots,k_d \geq 0:
        \quad\sum_{d=1}^D k_d = k
        \qquad\text{such that}
        \qquad\vec{\Delta} = \sum_{d=1}^{D} \vec{\Delta}^{(d)}
        \qquad\text{where}
        \qquad\vec{\Delta}^{(d)} \in \boundedvecs{1, k_d}\,.
    \end{align*}
\end{lemma}
\begin{proof}
   The decomposition can be identified in multiple ways, e.g., via a greedy \enquote{stack coloring} argument.
   We will argue via the procedure given in \Cref{alg:decompose} for a fixed $\vec{\Delta}\in\boundedvecs{D, k}$.
    \begin{algorithm}[tb]
        \caption{Procedure for decomposing $\vec{\Delta}$ into $\vec{\Delta}^{(1)}, \dots, \vec{\Delta}^{(D)}$}
        \label{alg:decompose}
        \begin{algorithmic}[1]
            \Statex \textbf{Input: } Vector $\vec{\Delta} \in \boundedvecs{D, k}\subseteq\mathbb{R}^{T}$.
            \Statex \textbf{Output:} Vectors $\vec{\Delta}^{(1)}, \dots, \vec{\Delta}^{(D)}\in\mathbb{R}^{T}$.
            \State Initialize $\forall d\in[-D, D] : \vec{\Delta}^{(d)} \leftarrow \vec{0}$ and $p \leftarrow 0$.
            \For{ $t\leftarrow 0$ to $T-1$ } 
                \For{ $i \leftarrow 1$ to $\lvert \vec{\Delta}[t] \rvert$ } 
                    \If{$\mathrm{sign}(\vec{\Delta}[t]) > 0$}
                        \State Set $p \leftarrow  p + 1$
                        \State Set $\vec{\Delta}^{(p)}[t] = 1$
                    \Else
                        \State Set $\vec{\Delta}^{(p)}[t] = -1$
                        \State Set $p \leftarrow  p - 1$
                    \EndIf
                \EndFor
            \EndFor
            \State Re-label the set $\{ \vec{\Delta}^{(d)} : d\in[-D, D] \}$ such that $\vec{\Delta}^{(1)}, \dots, \vec{\Delta}^{(D)}$ includes all non-zero vectors.
            \State \Return $\vec{\Delta}^{(1)},\dots,\vec{\Delta}^{(D)}$.
        \end{algorithmic}
    \end{algorithm}
    By the definition of $\boundedvecs{D, k}$, any interval sum on $\vec{\Delta}$ is in the range $[-D, D]$ \emph{and} the prefix takes on at most $D+1$ distinct values.
    To see the second claim, note that if the prefix were to take on more distinct values, then there would necessarily have to exist an interval sum taking on a value outside of $[-D, D]$.
    Immediately after the double for-loop has been executed and before Line~10 has been executed, we make a few remarks about the variables:
    \begin{enumerate}
        \item There are at most $D$ distinct vectors in $\{ \vec{\Delta}^{(d)} : d\in[-D, D] \}$ which are non-zero.
        \item All $\vec{\Delta}^{(d)}\in\{0, \pm 1\}^{T}$ are alternating vectors.
        \item $\sum_{d\in[-D, D]} \vec{\Delta}^{(d)} = \vec{\Delta}$.
    \end{enumerate}
    The first remark is immediate from there being at most $D+1$ distinct prefix values combined with the order in which $p$ and $\vec{\Delta}^{(p)}$ is updated on line~4-5 vs.\ line~8-9.
    E.g., if the prefix takes on values on the range $[a, b]$, then only $\{\vec{\Delta}^{(d)} : d\in[a, b+1]\}$.
    The second remark is immediate: if $\vec{\Delta}^{(d)}$ has a $+1$ inserted at time $t$, then necessarily the next element inserted at time $t' > t$ must be a $-1$ since the prefix must decrease in value to reach $d$ again.
    The third remark is also immediate: on reading $\vec{\Delta}[t]$, the same (signed) unit weight is spread equally over $\lvert \vec{\Delta}[t]\rvert$ vectors.

    After the re-labeling on line~10, It follows that \Cref{alg:decompose} will return $D$ alternating vectors $\vec{\Delta}^{(1)}, \dots, \vec{\Delta}^{(D)}$ from $\boundedvecs{1, k}$.
    We can refine the statement to say that each $\vec{\Delta}^{(d)} \in \boundedvecs{1, k_d}$ where $k_d = \| \vec{\Delta}^{(d)}\|_1$, from which we necessarily can claim that $\sum_{d\in[D]} k_d = k$, finalizing the proof.
\end{proof}
\begin{lemma}\label{lem:alternating-to-general-upper}
    Let $1\leq k \leq T$ and $D,p\geq 1$ where $\boundedvecs{D, k}\subset\mathbb{R}^{T}$ and $R\in\mathbb{R}^{m\times T}$ for $m\geq 1$.
    Then
    \begin{equation*}
        \sens_p(R, \boundedvecs{D, k})
        \leq \max_{k_1,\dots,k_D \geq 0 : \sum_{d=1}^D k_i = k} \sum_{d=1}^D  \sens_p(R, \boundedvecs{1, k_d})
        \leq D \cdot U(\lceil k/D \rceil)\,\,.
    \end{equation*}
    where the last bound holds for any concave and nondecreasing function $U : [0, T] \to \mathbb{R}$ such that $\sens_p(R, \boundedvecs{1, k}) \leq U(k)$.
\end{lemma}
\begin{proof}
    We start by bounding the target sensitivity.
    \begin{align*}
        \sens_p(R, \boundedvecs{D, k})
        &= \max_{\vec{\Delta}\in\boundedvecs{D, k}} \| R\vec{\Delta} \|_p
        = \| R\vec{\Delta}^* \|_p
        \leq \sum_{d=1}^{D} \| R\vec{\Delta}^{(d)} \|_p
        \leq \sum_{d=1}^{D} \sens_p(R, \boundedvecs{1, k_d})\,.
    \end{align*}
    The first inequality uses \Cref{lem:decompose-vectorset} applied to $\vec{\Delta}^{*}$ decomposing it into a sum over $\{\vec{\Delta}^{(d)}\in\boundedvecs{1, k_d} : d\in[D]\}$, and then uses the triangle inequality.
    The second inequality upper bounds each summand by the sensitivity.
    To get the first upper bound in the lemma, we maximize the upper bound over all choices of integer $k_1,\dots k_d\geq 0$ such that $\sum_{d\in[D]} k_d = k$.
    To finish the proof, we use $U$:
    \begin{equation*}
        \sens_p(R, \boundedvecs{D, k})
        \leq \max_{k_1,\dots,k_D \geq 0 : \sum_{d=1}^D k_d = k} \sum_{d=1}^D  \sens_p(R, \boundedvecs{1, k_d})
        \leq \max_{k_1,\dots,k_D \geq 0 : \sum_{d=1}^D k_d = k} \sum_{i=1}^D U(k_d)
    \end{equation*}
    By $U$ being concave and non-decreasing, we get that a near-even split maximizes the expression, i.e.,
    \begin{equation*}
        \sum_{d=1}^{D} U(k_d) \leq (D-r)U(q) + r U(q+1)
    \end{equation*}
    where $q = \lfloor k/D \rfloor$ and $r = k - Dq$.
    The intuition for why this is true, is that decrementing a larger argument to increment a smaller argument can never decrease the value of the sum.
    Formally, for integers $0 \leq a < b$
    \begin{equation*}
       U(a+1) + U(b-1) - (U(a) + U(b)) = \underbrace{U(a+1) - U(a)}_{\Delta U(a)} - \underbrace{U(b) - U(b-1)}_{\Delta U(b-1)} \geq 0
    \end{equation*}
    where the inequality holds due to forward difference $\Delta U(\cdot)$ being nonincreasing (concavity of $U$).
    To finish the proof, we have that
    \begin{align*}
        (D-r)U(q) + r U(q+1)
        \leq D\cdot U(q+1)\,.
    \end{align*}
    by $U$ being non-decreasing.
    Moreover, $q+1 = \lceil k / D\rceil$, unless $r=0$, but in this case $q = \lfloor k/D \rfloor = \lceil k/D\rceil$ and
    \begin{align*}
        (D-r)U(q) + r U(q+1) = D U(q) = DU(\lceil k/D\rceil)\,.
    \end{align*}
    Hence,
    \begin{equation*}
        \sens_p(R, \boundedvecs{k, D}) \leq D\cdot U(\lceil k/D \rceil)\,.\qedhere
    \end{equation*}
\end{proof}
It is natural to ask if $\sens_p(R, \boundedvecs{1, k}) = f(k)$ itself is a concave function in $k$, and so if the lemma could be simplified.
Unfortunately, there are matrices $R$ for which this is not the case: e.g.,
\begin{equation*}
    R=\begin{bmatrix} 1 & 0 & 1 \end{bmatrix}\in\mathbb{R}^{1\times 3}\,,\qquad\text{where}\qquad f(1) = 1\,,\quad f(2)=1\,,\quad f(3)=2^{1/p}\,,
\end{equation*}
which implies that $\frac{f(1) + f(3)}{2} = \frac{1 + 2^{1/p}}{2} \geq 1 = f(2)$, disproving concavity of $f$.
Nevertheless, for all factorizations and right factors $R$ that we consider, such concave, and tight, upper bounds exist.

We finish this section by summarizing \Cref{lem:alternating-to-general-lower}~and~\ref{lem:alternating-to-general-upper} in a single theorem.
\begin{theorem}\label{thm:alternating-to-general}
    Let $1\leq k \leq T$ and $D,p\geq 1$ where $\boundedvecs{D, k}\subset\mathbb{R}^{T}$ and $R\in\mathbb{R}^{m\times T}$ for $m\geq 1$.
    Then
    \begin{equation*}
        D \cdot \sens_p(R, \boundedvecs{1, \lfloor k/D \rfloor})
        \leq \sens_p(R, \boundedvecs{D, k})
        \leq \max_{k_1,\dots,k_D \geq 0 : \sum_{d=1}^D k_i = k} \sum_{d=1}^D  \sens_p(R, \boundedvecs{1, k_d})
        \leq D \cdot U(\lceil k/D \rceil)\,\,.
    \end{equation*}
    where the last bound holds for any concave and nondecreasing function $U : [0, T] \to \mathbb{R}$ satisfying that $\sens_p(R, \boundedvecs{1, k}) \leq U(k)$.
\end{theorem}

\section{Toeplitz Factorizations}\label{sec:toep}

In this section, we will prove error bounds for continual counting on $\boundedvecs{D,k}$ streams, using factorization mechanisms where the factorization satisfies some special properties.
As before, let $A$ be the $T \times T$ all-ones lower-triangular matrix (representing the query matrix for continual counting).
We show novel error bounds for continual counting (under new neighboring relations) using the square-root factorization~\cite{bennett77,fichtenberger_constant_2023,HenzingerUU23,DvijothamMPST24,henzinger2025normalizedsquarerootsharper} for $A$ where $A = LR$, and $L = R = \sqrt{A}$ (see \Cref{thm:sqrtprops} for a formal definition).
The sensitivity vector sets of interest (corresponding to these neighboring relations) are those described in \Cref{sec:senssets}, and follow naturally from releasing functions under continual observation in the fully dynamic streaming model (such as $\countdistinct$ and $\trianglecount$). 

The main results we prove are the following error bounds, which follow as corollaries to the sensitivity bounds on $R$ (i.e., bounds on $\sens_2(R,\boundedvecs{D,k})$) that we establish. 
See \Cref{sec:prelims} for a definition of the error metrics.

\begin{corollary}\label{cor:toep-error}
    Let $A\in\mathbb{R}^{T\times T}$ be an arbitrary query matrix and consider a factorization $A=LR$ where
    $R$ is a lower-triangular Toeplitz matrix with non-increasing diagonal values.
    Fix $D, k \geq 1$.
    Then
    \begin{align*}
        \maxse(L, R, \boundedvecs{D, k}) &\leq \sqrt{Dk}\cdot \maxse(L, R, \boundedvecs{1,1})\,,\\
        \meanse(L, R, \boundedvecs{D, k}) &\leq \sqrt{Dk}\cdot \meanse(L, R, \boundedvecs{1,1}),.
    \end{align*}
    In particular, if $A$ is the lower-triangular matrix of all-ones and $L = R = \sqrt{A}$, then
    \begin{align*}
        \maxse(L, R, \boundedvecs{D,k}) &\leq \left( \frac{\ln(T)}{\pi} + 1.067  \right)  \sqrt{\frac{Dk}{2\rho}}\,,\\
        \meanse(L, R, \boundedvecs{D,k}) &\leq \left( \frac{\ln(T)}{\pi} + 0.908 + o(1) \right)  \sqrt{\frac{Dk}{2\rho}}\,.
    \end{align*}
\end{corollary}
In essence, our sensitivity analysis allows us to slot-in existing leading-constant-optimal factorizations for standard continual counting on $\boundedvecs{1,1}$-streams, at the cost of a multiplicative factor $\sqrt{Dk}$.
We are also able to show that, in the case of the square-root factorization in particular, our analysis is in fact tight for a large range of $D$ and $k$.
\begin{corollary}\label{cor:sqrt-lower-bound}
    Let $A\in\mathbb{R}^{T\times T}$ be the lower-triangular matrix of all-ones, and consider the factorization $L=R=\sqrt{A}$.
    Fix $D, k \geq 1$ where $k/D \leq O(T^{1/3})$ and $T$ is sufficiently large.
    Then
    \begin{align*}
        \maxse(L, R, \boundedvecs{D,k}) &\geq \left( \frac{1}{\pi} - o(1) \right)  \sqrt{\frac{D^2\lfloor k/D\rfloor \ln(T)\ln\left(\frac{T}{\lfloor k/D\rfloor}\right)}{2\rho}}\,,\\
        \meanse(L, R, \boundedvecs{D,k}) &\geq \left( \frac{1}{\pi} - o(1) \right)  \sqrt{\frac{D^2\lfloor k/D\rfloor \ln(T)\ln\left(\frac{T}{\lfloor k/D\rfloor}\right)}{2\rho}}\,.
    \end{align*}
\end{corollary}
\janote{decide if this and everywhere else should avoid the natural logarithm.}

The rest of the section is organized as follows.
Firstly, in \Cref{sec:sensquadform}, we prove properties of the matrix $R^TR$ that are crucial to our analysis.
Next, in \Cref{sec:toep-upper-bound}, we will prove \Cref{cor:toep-error} and argue in particular that the square-root factorization has the required properties.
We give two separate proofs for the key bound on $\sens_2(R, \boundedvecs{D, k})$: one combinatorial (based on a matching argument), and one based on direct manipulation of the quadratic form.
Finally, in \Cref{sec:toep-lowerbound}, we prove \Cref{cor:sqrt-lower-bound}, demonstrating that our bound on the sensitivity $\sens_2(R,\boundedvecs{D,k})$ is tight for a large range of $D$ and $k$.
We do so via the probabilistic method: we lower bound $\mathbb{E}[\| \sqrt{A} \dvec\|_2^2]$ for $\dvec$ uniformly sampled from a subset of $\boundedvecs{1, k}$.

\subsection{Sensitivity as Quadratic Form and Matrix Properties}\label{sec:sensquadform}

Let $R \in\mathbb{R}^{T\times T}$ be a lower-triangular matrix Toeplitz matrix with diagonals $r_0,\dots,r_{T-1}$ that are monotonically non-increasing and non-negative.
Note that in such a matrix, we have that $R[i, j] = r_{i-j}$ for $i\geq j$.
We will proceed by analyzing $\|R \vec{\Delta} \|_2^2$, since this corresponds to a quadratic form that is easier to reason about algebraically.
By definition, we have that
\begin{equation*}
    \| R\vec{\Delta} \|_2^2 = \vec{\Delta}^T R^T R \vec{\Delta}\,.
\end{equation*}
We first analyze properties of the matrix $R^T R$ and then reason about the above quadratic form to derive our bounds on the sensitivity $\sens_2(R,\boundedvecs{D,k})$.
Let $c(i,j) \equiv R^TR[i,j]$.
We establish some key monotonicity properties of $c(i,j)$ that will prove useful in our later analysis.

\begin{claim}\label{claim:monomat}
Let $c(i,j)$ be as defined above. Then, 
\begin{enumerate}
    \item {\bf Symmetry. }For all $0 \leq i,j \leq T-1$, we have that $c(i,j) = c(j,i).$
    \item {\bf Monotonicity along row (beyond diagonal entry). } Let $0 \leq i \leq j \leq T-1$. For $0 \leq m \leq T - j -1$, we have that $c(i,j) \geq c(i,j+m)$.
    \item {\bf Monotonicity of diagonal entries. } Let $0 \leq m \leq j \leq T- max(i,j) - 1$. We have that $c(i,j) \geq c(i+m,j+m)$.
\end{enumerate}
\end{claim}

\begin{proof}
First, we write a closed form expression for $c(i,j)$:
\begin{align}
    c(i,j) &= (R^T R)[i,j]
    = \sum_{k=0}^{T-1} {R^T}[i, k] R[k, j] \nonumber \\
    &= \sum_{k=\max(i, j)}^{T-1} {R}[k, i] R[k, j] = \sum_{k=\max(i, j)}^{n-1} r_{k-i} r_{k-j} \nonumber \\
    &=  \sum_{\ell=0}^{n-\max(i, j)-1} r_{\ell} r_{\ell + |i-j|}\,,\label{eq:cij}
\end{align}
where the second equality follows from the definition of matrix multiplication, the third equality follows because $R$ is lower-triangular and by the definition of transpose, the fourth equality follows because $R$ is Toeplitz and hence $R[i,j] = r_{i-j}$, and the final equality is substituting $\ell = k - \max(i,j)$ and rewriting the sum in terms of $\ell$. 

We now prove each item of the claim. 
\begin{enumerate}
    \item {\bf Symmetry.} $c(i,j) = c(j,i)$: this follows because $(R^TR)^T = R^T R$, by properties of the transpose operation.
    \item {\bf Monotonicity along row (beyond diagonal entry). } Let $0 \leq i \leq j \leq T-1$. For $0 \leq m \leq T-j-1$, consider $c(i,j) - c(i,j+m)$. By \Cref{eq:cij}, we have that
    \begin{align*}
    c(i,j) - c(i,j+m) & = \sum_{\ell=0}^{T-j-1} r_{\ell} r_{\ell + j-i} -  \sum_{\ell=0}^{T-j-m-1} r_{\ell} r_{\ell + j+m-i} \\
    & = \sum_{\ell=0}^{T-j-m-1} r_{\ell} [r_{\ell + j-i} - r_{\ell + j+m-i}] +  \sum_{\ell=T-j-m}^{T-j-1} r_{\ell} r_{\ell + j-i}\\
    & \geq 0\,,
    \end{align*}
    where the last inequality follows from the fact that $r_s \geq r_t$ for $s \leq t$ for our matrix of interest, and by the fact that all $r$ values are non-negative. 
    \item {\bf Monotonicity of diagonal entries. } Let $0 \leq i \leq j \leq T-1$. For $0 \leq m \leq T-\max(i,j)-1$, consider $c(i,j) - c(i+m,j+m)$. By \Cref{eq:cij}, we have that 
    \begin{align*}
    c(i,j) - c(i+m,j+m) & = \sum_{\ell=0}^{T-j-1} r_{\ell} r_{\ell + j-i} -  \sum_{\ell=0}^{T-j-m-1} r_{\ell} r_{\ell + j+m-i-m} \\
    & = \sum_{\ell=T-j-m}^{T-j-1} r_{\ell} r_{\ell + j-i} \geq 0\,,
    \end{align*}
where the last inequality follows since every $r$ value is non-negative.
\end{enumerate}
\end{proof}

We are now ready to analyze the quadratic form of interest.
\begin{align}\label{eq:quadraticform}
    \| R \dvec \|_2^2
    \vec{\Delta}^T R^T R \vec{\Delta} &= \sum_{i, j\in[T]} (R^T R)[i,j] \vec{\Delta}[i] \vec{\Delta}[j] = \sum_{i,j\in[T]} c(i,j) \vec{\Delta}[i] \vec{\Delta}[j]\,.
\end{align}

We now analyze this quantity for the sensitivity vector sets that are of interest (those described in \Cref{sec:senssets}).

\subsection{Upper Bounding the Sensitivity}\label{sec:toep-upper-bound}

The following theorem is the key result for this section, bounding the $\ell_2$ sensitivity for $\boundedvecs{D,k}$-streams.

\begin{theorem}\label{thm:toep-sens}
    Let $R\in\mathbb{R}^{T\times T}$ be a lower-triangular Toeplitz matrix with non-increasing, non-negative diagonal values.
    Fix $D, k \geq 1$.
    Then, 
    \begin{equation*}
        \sens_2(R,\boundedvecs{D,k}) \leq \sqrt{Dk}\cdot \| R \|_{1\to 2}\,.
    \end{equation*}
\end{theorem}

We will give two separate proofs for the theorem.
The first proof is combinatorial, relying on a matching argument.
It is arguably more intuitive than our second, shorter and more direct proof.
\begin{proof}[Proof of \Cref{thm:toep-sens} via matching argument]
First, consider the simpler case of $D=1$.
Consider any $\vec{\Delta} \in \pj{\boundedvecs{1,k}}$. From Equation~\eqref{eq:quadraticform}, 
\begin{align}
    \|R \vec{\Delta} \|_2^2 = \sum_{ i, j \in [T]} c(i,j) \vec{\Delta}[i] \vec{\Delta}[j]  & = \sum_{i=0}^{T-1} c(i,i) \vec{\Delta}[i]^2 + 2 \sum_{0\leq i < j \leq T-1} c(i,j) \vec{\Delta}[i] \vec{\Delta}[j]\enspace \nonumber\\
    & 
    \leq \sum_{i=0}^{k-1}c(i,i) + 2 \sum_{0\leq i < j \leq T-1} c(i,j) \vec{\Delta}[i] \vec{\Delta}[j]\,.\label{eq:secondterm}
\end{align}
where the second equality follows because of Item 1 of \Cref{claim:monomat} ($c(i,j)=c(j,i)$), and the second inequality follows because $\vec{\Delta} \in \{-1,0,1\}^T$, because  $\vec{\Delta}$ is $k$-sparse, and because of Item 3 in \Cref{claim:monomat} ($c(i,j) \geq c(i+m,j+m)$ for $0 \leq m \leq T-\max(i,j)-1$).

We will argue that the second term in the above sum is negative via a combinatorial argument, and thereby bound the result by the first term.

Consider the following matching operation $M$ that matches ordered pairs: for every $i<j$ such that $\vec{\Delta}[i] \neq \vec{\Delta}[j]$ and $\vec{\Delta}i, \vec{\Delta}[j] \neq 0$, match $(i,j)$ with $(i,j')$ where $j'$ is the smallest value larger than $j$, such that $\vec{\Delta}[i] = \vec{\Delta}[j']$. If such a $j'$ does not exist, leave $(i,j)$ unmatched.

Let the set of matches including $(i,.)$ be $M_i$. Let the set of unmatched ordered pairs $(i,.)$ be $\pj{N_i}$.\palak{$S_i$ is overloaded}

We observe two properties of this matching $M$:
\begin{itemize}
    \item {\bf All pairs $(i,j) \in \pj{N_i}$ satisfy $x_i \neq x_j$.} We first argue that every $(i,j)$ with $\vec{\Delta}[i], \vec{\Delta}[j] \neq 0$ and $\vec{\Delta}[i] = \vec{\Delta}[j]$ is matched. To see this, consider such an $i$ and $j$. Since $\vec{\Delta}$ is alternating, there exists an $i < j_1 < j$ such that $\vec{\Delta}[j_1] \neq 0$, $\vec{\Delta}[i] \neq \vec{\Delta}[j_1]$, and there exists no non-zero entry between $\vec{\Delta}[j_1]$ and $\vec{\Delta}[j]$. Note that this implies $(i,j)$ and $(i,j_1)$ are matched together. This implies that all pairs $(i,j) \in \pj{N_i}$ satisfy $x_i \neq x_j$.

   \item {\bf Contribution of matched pairs is non-positive.} We also argue that in every match $\{(i,j_1), (i,j_2) \} \in M_i$ (where $\vec{\Delta}[j_1] \neq \vec{\Delta}[i]$ and $\vec{\Delta}[j] = \vec{\Delta}[i] $) we have that the total contribution of the ordered pairs in this match to the sum of interest is non-positive, that is $c(i,j_1) \vec{\Delta}[i]\vec{\Delta}[j_1] + c(i,j_2)\vec{\Delta}[i]\vec{\Delta}[j_2] \leq 0$. This follows since $i < j_1 < j_2$ by the way the matching is done, which implies by Item 2 of \Cref{claim:monomat} that $c(i,j_2) \leq c(i,j_1)$). 

\end{itemize}

Then, analyzing the second term in the right hand side of \Cref{eq:secondterm}, and reducing this to reasoning about each index $i$ separately, we get that
\begin{align*}
    &2 \sum_{0\leq i < j \leq T-1} c(i,j) \vec{\Delta}[i] \vec{\Delta}[j] \\ & = 2 \sum_{0\leq i \leq T-1}  \left(\sum_{(i,a) \in \pj{N_i}} c(i,a) \vec{\Delta}[i] \vec{\Delta}[a] +\sum_{\{(i,s),(i,t)\} \in M_i} c(i,s)\vec{\Delta}[i] \vec{\Delta}[s] + c(i,t)\vec{\Delta}[i] \vec{\Delta}[t] \right) \\
    & \leq 0
\end{align*}
where the last inequality follows from the two properties observed about the matching and from the fact that $c(.,.)$ is always non-negative. Hence, in total, from \Cref{eq:secondterm}, we get that 
\begin{equation*}
   \|R \vec{\Delta}\|_2
   \leq \sqrt{\sum_{i=0}^{k-1}c(i,i)} \leq \sqrt{kc(0,0)}
   = \sqrt{k}\cdot \sens_2(R, \boundedvecs{1,1}),
\end{equation*}
where the last inequality follows from the monotonicity of the diagonal entries (Item 3 of \Cref{claim:monomat}).
In particular we have shown
\begin{equation*}
    \sens_2(R, \boundedvecs{1,k})
    \leq  \sqrt{k}\cdot \sens_2(R, \boundedvecs{1,1})\,.
\end{equation*}
To complete the proof, we now consider $D > 1$.
Using \Cref{lem:alternating-to-general-upper} at the first step, we have that
\begin{align*}
    \sens_2(R, \boundedvecs{D, k})^2
    &\leq \left(\max_{k_1,\dots,k_D: \sum_{i=1}^D k_i = k} \sum_{i=1}^D  \sens_2(R, \boundedvecs{1, k_i})\right)^2\\
    &\leq \max_{k_1,\dots,k_D: \sum_{i=1}^D k_i = k} \left(\sum_{i=1}^D   \sqrt{k_i}\| R \|_{1\to 2}\right)^2
    \leq Dk\cdot \| R \|_{1\to 2}^2
\end{align*}
where the second step used our result for $D=1$, and the last step used Cauchy-Schwarz.
Taking a square-root finishes the proof.
\end{proof}

In addition to this intuitive combinatorial proof, we also give a proof that is more analytical.
While we do not leverage it in this paper, the proof has the additional benefit that it also holds if $\boundedvecs{D, k}$ is defined over $\mathbb{R}^T$, rather than $\mathbb{Z}^T$.
\begin{proof}[Proof of \Cref{thm:toep-sens} via Abel summation.]
    Consider arbitrary $\vec{\Delta}\in\boundedvecs{D,k}\subseteq\mathbb{R}^{T}$.
    We can write
    \begin{align*}
        \| R \vec{\Delta}\|_{2}^2
        = \sum_{i, j\in[T]} c(i, j)\vec{\Delta}[i]\vec{\Delta}[j]
        = \sum_{i=0}^{T-1}\vec{\Delta}[i] \underbrace{\sum_{j=0}^{T-1} c(i, j)\vec{\Delta}[j]}_{Q_i}\,,
    \end{align*}
    where we continue working with the inner sum.
    Define the infinite sequences $(f_t)_{t\in\mathbb{Z}}, (g_t)_{k\in\mathbb{Z}}$ with bounded support:
    \begin{equation*}
        f_t = \begin{cases}
            c(i, t)\quad&\text{for}\ 0 \leq t\leq T-1\,,\\
            0\quad&\text{otherwise}\,,
        \end{cases}\quad
        g_t = \begin{cases}
            \vec{\Delta}[t]\quad&\text{for}\ 0 \leq t\leq T-1\,,\\
            0\quad&\text{otherwise}\,.
        \end{cases}
    \end{equation*}
    Defining $G_m = \sum_{t=-\infty}^{m} g_t$, we write
    \begin{align*}
        Q_i &= \sum_{j=0}^{T-1} c(i, j) \vec{\Delta}[j]
        = \sum_{t\in\mathbb{Z}} f_t (G_t - G_{t-1})
        = \sum_{t\in\mathbb{Z}} f_t G_t - \sum_{t\in\mathbb{Z}} f_t G_{t-1}
        = \sum_{t\in\mathbb{Z}} (f_t-f_{t+1}) G_t\,,
    \end{align*}
    where the last step changes the indexing in the second sum and then merges the sums.
    This transformation is sometimes called \emph{Abel summation}.
    By the monotonicity, non-negativity and symmetry of $c(i, j)$ from \Cref{claim:monomat}, we have that $c(i,j)-c(i, j+1) \geq 0$ for $j\geq i$, and $c(i, j) - c(i, j+1) \leq 0$ for $j < i$.
    We can thus write
    \begin{align*}
        Q_i &= \sum_{j\geq i} \left[c(i, j) - c(i, j+1)\right] G_j
        - \sum_{j < i} \left[c(i, j+1) - c(i, j)\right]G_j\\
        &\leq \sup_{j\in\Z}(G_j)\sum_{j\geq i} \left[c(i, j) - c(i, j+1)\right]
        - \inf_{j\in\Z}(G_j)\sum_{j < i} \left[c(i, j+1) - c(i, j)\right]\\
        &= \left(\sup_{j\in\Z}(G_j) - \inf_{j\in\Z}(G_j)\right)c(i, i)\,,
    \end{align*}
    where the inequality uses that $\inf_k(b_k)\sum_k a_k\leq \sum_k a_k b_k \leq \sup_k(b_k) \sum_k a_k$ for a nonnegative sequence $(a_k)_k$, and the final equality that the sums telescope.
    By an analogous argument (swap the role of the sums), we can also derive a lower bound:
    \begin{align*}
        Q_i 
        &\geq \left(\inf_{j\in\Z}(G_j) - \sup_{j\in\Z}(G_j)\right)c(i, i)\,.
    \end{align*}
    Combining the two bounds, we get
    \begin{align*}
        \lvert Q_i \rvert
        &\leq \left(\sup_{j\in\Z}(G_j) - \inf_{j\in\Z}(G_j)\right)c(i, i)
        \leq D\cdot c(i, i)\,,
    \end{align*}
    where the inequality uses that $G_j$ is a prefix on $\vec{\Delta}$, and so the difference $\sup_{j\in\Z}(G_j) - \inf_{j\in\Z}(G_j)$ is a (signed) \emph{interval sum on $\vec{\Delta}$}, and so bounded by $D$ from the definition of $\boundedvecs{D, k}$.

    We thus have (taking absolute value of the summand)
    \begin{align*}
        \| R \vec{\Delta} \|_{2}^2 &= \sum_{i=0}^{T-1}\vec{\Delta}[i] Q_i
        \leq  D\sum_{i=0}^{T-1}\lvert \vec{\Delta}[i]\rvert c(i,i)\,.
    \end{align*}
    Using that $c(i,i) \leq c(0,0) = \|R \|_{1\to 2}^{2}$ and that $\|\vec{\Delta}\|_1 \leq k$ by definition:
    \begin{align*}
        \| R \vec{\Delta} \|_{2}^2
        \leq D\cdot\| R\|_{1\to 2}^2 \sum_{i=0}^{T-1}\lvert\vec{\Delta}[i]\rvert
        \leq Dk\cdot\| R\|_{1\to 2}^2\,.
    \end{align*}
    Invoking the final bound for the sensitivity
    \begin{equation*} 
        \sens_2(R,\boundedvecs{D,k}) = \max_{\vec{\Delta}\in\boundedvecs{D,k}} \| R\vec{\Delta}\|_{2}
        \leq \sqrt{D k}\cdot \| R \|_{1\to 2}
    \end{equation*}
    finishes the proof.
\end{proof}


We note that the square-root factorization from previous work~\cite{bennett77,fichtenberger_constant_2023,HenzingerUU23,DvijothamMPST24,henzinger2025normalizedsquarerootsharper} where $L = R = \sqrt{A}$ (and $A$ is the lower-triangular all-ones matrix) satisfies the properties of $R$ described above.
This is formally stated next, together with useful properties of the factorization.
\begin{theorem}[\cite{henzinger2025normalizedsquarerootsharper}]\label{thm:sqrtprops}
    Let $A\in\mathbb{R}^{T\times T}$ be the lower-triangular all-ones matrix. Define its square-root $\sqrt{A}\in\mathbb{R}^{T\times T}$ as the lower-triangular Toeplitz matrix with entries $r_t = \binom{2t}{t} / 4^t$
    on its $t$\textsuperscript{th} diagonal.
    That is, $\sqrt{A}[i, j] = r_{i-j}$ for $i\geq j$.
    Let $c(i,j) = \left(\sqrt{A}^T\sqrt{A}\right)[i,j]$.
    Then, the following statements are true:
    \begin{enumerate}
        \item $(r_t)_{t\in\mathbb{N}}$ is a decreasing sequence with $r_0 = 1$, and for integer $t\geq 1$,
        \begin{equation*}
            \frac{1}{\sqrt{\pi(t+1)}} \leq r_t \leq \frac{1}{\sqrt{\pi t}}\tag*{\text{\cite[Lemma 4]{henzinger2025normalizedsquarerootsharper}}}
        \end{equation*}
        \item  For $i\in[T]$,
        \begin{equation*}
           1 + \frac{\ln(T-i)}{\pi} \leq c(i, i) \leq 1.067 + \frac{\ln(T-i)}{\pi}\tag*{\text{\cite[Theorem 3]{henzinger2025normalizedsquarerootsharper}}}
        \end{equation*}
        \item For $1/2$-zCDP, we have the following results for the error,
        \begin{align*}
            \maxse(\sqrt{A},\sqrt{A},\boundedvecs{1, 1}) &= \|\sqrt{A}\|_{1 \to 2}^2 = \|\sqrt{A} \|_{2\to \infty}^2 = c(0,0)\,,\\
            \meanse(\sqrt{A},\sqrt{A},\boundedvecs{1, 1}) &= 0.908 + \frac{\ln(T)}{\pi} + o(1) \leq c(0,0)\tag*{\text{\cite[Theorem 4]{henzinger2025normalizedsquarerootsharper}}}
        \end{align*}
        \item Additionally, \begin{equation*}
            \frac{\| \sqrt{A}\|_F}{\sqrt{T}} = \sqrt{\frac{\ln(T)}{\pi}} + O(1)\,.\tag*{\text{\cite[Proof of Theorem 4]{henzinger2025normalizedsquarerootsharper}}}
        \end{equation*}
    \end{enumerate}
\end{theorem}

Given~\Cref{thm:sqrtprops}, we can now prove \Cref{cor:toep-error}.
\begin{proof}[Proof of \Cref{cor:toep-error}]
    Invoking \Cref{thm:toep-sens} for the bound on $\sens_2(R,\boundedvecs{D,k})$ in the definition of $\maxse(L, R, \boundedvecs{D, k})$ and $\meanse(L, R, \boundedvecs{D, k})$ gives the first part of the corollary statement.
    From (1) of~\Cref{thm:sqrtprops}, we additionally have that $R=\sqrt{A}$ satisfies the monotonicity condition of \Cref{thm:toep-sens}.
    Plugging in the error bounds from (3) of \Cref{thm:sqrtprops} into the first statement of \Cref{cor:toep-error} proves the final statement.
\end{proof}

\subsection{A Lower Bound on the Square-Root Sensitivity}\label{sec:toep-lowerbound}
In the case of the square-root factorization, \Cref{thm:toep-sens} allows us to give a bound of
\janote{decide on a rule for capitalization in titles/paragraphs (perhaps check the PODS template for inspiration)}
\begin{equation*}
    \sens_2(\sqrt{A}, \boundedvecs{D, k}) \leq \sqrt{Dk} \| \sqrt{A} \|_{1\to 2}
     = \left(\frac{1}{\sqrt{\pi}} + o(1)\right)\sqrt{Dk\ln(T)}
\end{equation*}
for the $\ell_2$ sensitivity.
It is natural to ask if our bound is tight for general values of $D$ and $k$.
We show that, at least for sufficiently small values of $k/D$, it is.
\begin{theorem}\label{thm:sqrt-tight-sens}
    For $1 \leq D \leq k$, $T$ sufficiently large, and $k/D = O(T^{1/3})$, we have that
    \begin{equation*}
        \sens_2(\sqrt{A}, \boundedvecs{D, k}) \geq \left(\frac{1}{\sqrt{\pi}} - o(1) \right)\sqrt{D^2 \lfloor k/D\rfloor \ln\left(\frac{T}{\lfloor k/D \rfloor}\right)}\,,
    \end{equation*}
    where $o(1)$ tends to zero as $T\to\infty$.
    If additionally $k/D = \omega(1)$, then the following simpler bound holds
    \begin{equation*}
        \sens_2(\sqrt{A}, \boundedvecs{D, k}) \geq \left(\frac{1}{\sqrt{\pi}} - o(1) \right)\sqrt{Dk\ln(DT/k)}\,.
    \end{equation*}
\end{theorem}
Essentially, our upper bound on the sensitivity has the correct asymptotics for $k/D = O(T^{1/3})$.
Moreover, if $k/D$ is subpolynomial and either (1) $D=1$, or (2) $k/D$ is superconstant, then additionally the leading constant in the upper bound is tight.

The theorem is proved using the probabilistic method: for a $\vec{\Delta}$ drawn uniformly at random from $\boundedvecs{1,k}$ (conditioned on $\|\vec{\Delta}\|_1 = k$), we show that $\mathbb{E}[\|\sqrt{A}\vec{\Delta}\|_{2}^2] \leq \sens_2(\sqrt{A}, \boundedvecs{1, k})^2$ has to be large.
Using \Cref{lem:alternating-to-general-lower}, we are then able to extend this result to $\boundedvecs{D, k}$.

We need the following two lemmas to prove \Cref{thm:sqrt-tight-sens}.
The first lemma is a bound on $c(i, j)$ for the square-root whenever $i\neq j$.
The techniques used for proving it are standard, with similar statements having appeared in past work studying the square-root factorization, see e.g.,~\cite{fichtenberger_constant_2023,HenzingerUU23,KalLamp24,DvijothamMPST24,kalinin2025learningrateschedulingmatrix}.
The second lemma bounds the expected parity of a hypergeometric random variable.
Such a result is likely folklore, but we could find no reference for it.
\begin{lemma}\label{lem:sqrt-terms}
    Let $R=\sqrt{A}\in\mathbb{R}^{T\times T}$ be the square-root matrix, and define $c(i,j) \equiv (R^T R)[i,j]$ for $i,j\in[T]$.
    The following bounds hold for $i\neq j$:
    \begin{align*}
        c(i, j) &\leq \frac{1}{\pi}\ln\left(1+\frac{T}{\lvert i-j \rvert}\right) + \frac{3}{2}\,.
    \end{align*}
\end{lemma}
\begin{lemma}\label{lem:signed-hypergeom}
    Let $N,K,n \geq 1$ and $X\sim\mathrm{Hypergeom}(N,K,n)$.
    That is, $X$ describes the number of successful draws out of $n$, when sampling without replacement from a population of size $N$, out of which $K$ elements are successes.
    Then
    \begin{equation*}
        \lvert \mathbb{E}\left[ (-1)^{X}\right]\rvert
        \leq \exp(-2nK/N) + n(n-1)/N\,.
    \end{equation*}
\end{lemma}
\begin{proof}[Proof of \Cref{lem:sqrt-terms}]
    Recall that the square-root matrix has entries $\sqrt{A}[i,j] = r_{i-j}$ where $r_t = \binom{2t}{t} / 4^t$, and in particular satisfies the bounds
    \begin{equation*}
        \frac{1}{\sqrt{\pi (t+1)}}\leq r_t \leq \frac{1}{\sqrt{\pi t}}
    \end{equation*}
    for all $t\geq 1$~(\Cref{thm:sqrtprops}).

    Without loss of generality, let $j < i$.
    Define $d=i-j > 0$, and let $N=T-1-\max(i,j) = T-1-i$
    \begin{align*}
        c(i, j) &= \sum_{t=0}^{N} r_t r_{t+d}
        \leq \frac{1}{\sqrt{\pi d}} + \frac{1}{\pi}\sum_{t=1}^{N} \frac{1}{\sqrt{t(t+d)}}
        \leq \frac{1}{\sqrt{\pi d}} + \frac{1}{\pi\sqrt{d+1}} + \frac{1}{\pi}\int^{N}_{1} \frac{\mathrm{dz}}{\sqrt{z(z+d)}}
    \end{align*}
    We continue working on the integral:
    \begin{equation*}
        \int^{N}_{1} \frac{\mathrm{dz}}{\sqrt{z(z+d)}}
        = \left[2\cdot\mathrm{asinh}\left(\sqrt{\frac{z}{d}}\right)\right]_{z=1}^{z=N}
        \leq 2\cdot\mathrm{asinh}\left(\sqrt{N/d}\right)\,,
    \end{equation*}
    where $\mathrm{asinh}(x) \equiv \ln(x + \sqrt{1+x^2})$.
    Given $x \geq 0$, we have $\sqrt{1+x^2} \leq 1+x$, and so $\mathrm{asinh}(x) \leq \ln(1 + 2x)$.
    Further, note that $1+2\sqrt{x} \leq 3\sqrt{x}$ if $x>0$, and so also $\ln(1+2\sqrt{x}) \leq \ln(3) + \frac{1}{2}\ln(x) \leq \ln(3) + \frac{1}{2}\ln(1+x)$.
    Similarly, $1+2\sqrt{x} \leq 3$ for $0<x<1$, and so $\ln(1+2\sqrt{x}) \leq \ln(3) \leq \ln(3) + \frac{1}{2}\ln(1+x)$.
    We can thus write
    \begin{equation*}
        \mathrm{asinh}\left(\sqrt{N/d}\right) \leq \ln(1+2\sqrt{N/d})
        \leq \ln(3) + \frac{1}{2}\ln(1+N/d)\,.
    \end{equation*}
    Putting everything together
    \begin{align*}
        c(i, j) &\leq \frac{1}{\sqrt{\pi d}} + \frac{1}{\pi\sqrt{d+1}} + \frac{2}{\pi}\left(\ln(3) + \frac{1}{2}\ln(1+N/d) \right)\\
        &= \frac{1}{\sqrt{\pi d}} + \frac{1}{\pi\sqrt{d+1}} + \frac{2\ln(3)}{\pi} + \frac{1}{\pi}\ln(1+N/d)\\
        &\leq \underbrace{\frac{1}{\sqrt{\pi}} + \frac{1}{\pi\sqrt{2}} + \frac{2\ln(3)}{\pi}}_{\approx 1.489} + \frac{1}{\pi}\ln(1+N/d)\\
        &\leq \frac{3}{2} + \frac{\ln\left(1 + \frac{N}{d}\right)}{\pi}\,.
    \end{align*}
    Since $c(i, j) = c(j, i)$, the final lemma statement follows from setting $d=\lvert i-j\rvert$ and upper bounding $N=T-\max(i,j)-1 \leq T$.
\end{proof}
\begin{proof}[Proof of \Cref{lem:signed-hypergeom}]
    We will upper bound the expression by approximating it with a binomial distribution.
    Define $\hat{X}\sim\mathrm{Bin}(n, p)$ for $p=K/N$.
    We have that
    \begin{equation*}
        \mathbb{E}\left[(-1)^{\hat{X}}\right]
        = \sum_{\ell=0}^{n} (-1)^{\ell}\binom{n}{\ell}p^{\ell}(1-p)^{n-\ell}
        = (1-2p)^{n} \leq \exp(-2pn)
    \end{equation*}
    where the second equality follows from the binomial expansion of $(1-p + (-p))^n$.
    Note that $X$ describes sampling \emph{without} replacement, and $\hat{X}$ \emph{with} replacement.
    Consider the process where we instead draw $n$ samples with replacement from $N$.
    The probability of avoiding a collision (sampling an element more than once) is:
    \begin{equation*}
        \Pr[\text{\#collisions}=0] =
        \frac{N}{N}\cdot\frac{N-1}{N}\cdot\ldots\cdot\frac{N-n+1}{N} = \prod_{\ell=0}^{n-1} (1-\ell/N)
    \end{equation*}
    and so
    \begin{equation*}
        \Pr[\text{\#collisions}\geq 1] =
        1- \prod_{\ell=0}^{n-1} (1-\ell/N)
        \leq \sum_{\ell=0}^{n-1} \frac{\ell}{N}
        = \frac{n(n-1)}{2N}\,,
    \end{equation*}
    where the inequality is a union bound.
    We now use a coupling trick.
    Let $\hat{X}$ be the outcome of sampling these $n$ samples with replacement from $N$.
    If there is no collision among the $n$ samples, then we define $X=\hat{X}$.
    Otherwise, if there is a collision, we set $X$ according to the outcome of sampling $n$ fresh elements without replacement.
    Clearly $X$ and $\hat{X}$ follow the correct hypergeometric and binomial distributions respectively, but now $\Pr[X\neq \hat{X}] \leq \Pr[\text{\#collisions}\geq 1] \leq \frac{n(n-1)}{2N}$.
    
    To finish the proof, we establish that the expectations are close:
    \begin{align*}
       \big\lvert \mathbb{E}\big[(-1)^{X}\big] - \mathbb{E}\big[(-1)^{\hat{X}}\big]\big\rvert
       &= \big\lvert \mathbb{E}\big[(-1)^{X} - (-1)^{\hat{X}}\big]\big\rvert
       \leq \mathbb{E}\big[\big\lvert (-1)^{X} - (-1)^{\hat{X}}\big\rvert\big]\\
       &= 2\cdot\Pr[X + \hat{X}\ \text{is odd}]
       \leq 2\cdot\Pr[X \neq \hat{X}]\\
       &\leq n(n-1)/N\,.
    \end{align*}
    The final statement follows from invoking the triangle inequality:
    \begin{align*}
        \big\lvert \mathbb{E}\big[(-1)^{X}\big]\big\rvert
         \leq \lvert \mathbb{E}\big[(-1)^{\hat{X}}\big]\rvert + n(n-1)/N
         \leq \exp(-2nK/N) + n(n-1)/N\,.\qquad\qedhere
    \end{align*}
\end{proof}

We proceed to prove \Cref{thm:sqrt-tight-sens}.
\begin{proof}[Proof of \Cref{thm:sqrt-tight-sens}]
    Let $R=\sqrt{A}\in\mathbb{R}^{T\times T}$.
    Consider the case where $D=1$ and $k\geq 3$; we will return to the other cases towards the end of the proof.
    We will argue via the probabilistic method.
    Consider $\vec{\Delta}$ drawn uniformly at random from $\boundedvecs{1, k}$, conditioned on $\|\dvec\|_1 = k$.
    Call this distribution $\mathcal{U}$.
    If we can show a lower bound on the expectation of $\| R\vec{\Delta}\|_{2}^2$, then that implies a lower bound on the (squared) $\ell_2$ sensitivity, i.e.,
    \begin{equation*}
        \sens_2(R, \boundedvecs{1, k}) \geq \sqrt{\mathbb{E}_{\vec{\Delta}\sim\mathcal{U}}\left[ \| R \vec{\Delta} \|_2^2 \right]}\,.
    \end{equation*}
    We suppress the subscript of the expectation from now on; all expectations are taken over $\mathcal{U}$.
    We have that
    \begin{align*}
        \mathbb{E}\left[ \| R \vec{\Delta} \|_2^2 \right]
        &= \mathbb{E}\left[ \sum_{(i,j)\in[T]^2} c(i, j) \vec{\Delta}[i]\vec{\Delta}[j] \right]
        = \sum_{i=0}^{T-1} c(i, i) \mathbb{E}\left[\vec{\Delta}[i]^2\right]
        + 2\sum_{j<i} c(i, j) \mathbb{E}\left[\vec{\Delta}[i]\vec{\Delta}[j]\right]\,,
    \end{align*}
    by linearity of expectation.
    We begin with the first term.
    Noting that $\mathbb{E}[\vec{\Delta}[i]^2] = k/T$, we can directly bound it using the bound on $c(i,i)$ from \Cref{thm:sqrtprops}:
    \begin{align*}
        \sum_{i=0}^{T-1} c(i, i) \mathbb{E}\left[\vec{\Delta}[i]^2\right] 
        &\geq \frac{k}{T}\left(T + \sum_{i=1}^{T-1} \frac{\ln(T-i)}{\pi}\right)
        \geq \frac{k}{T}\left(\frac{(\pi - 1)T}{\pi} + \frac{T\ln(T)}{\pi}\right)
        = \frac{k(\ln(T) + \pi - 1)}{\pi}\,,
    \end{align*}
    where the last inequality uses Stirling's approximation applied to $\ln(T!) \geq T\ln(T) - T$.

    For the second term, we first reason about $\mathbb{E}\left[\vec{\Delta}[i]\vec{\Delta}[j]\right]$.
    Let $X\sim\mathrm{Hypergeom}(N,K,n)$ denote a random variable drawn from a hypergeometric distribution, where we are drawing $n=k-2$ samples from a population of $N=T-2$, out of which $K=i-j-1$ are labeled successes.
    We claim that
    \begin{equation*}
        \mathbb{E}\left[\vec{\Delta}[i]\vec{\Delta}[j]\right]
        = \mathbb{E}\left[\vec{\Delta}[i]\vec{\Delta}[j] \Big\vert \vec{\Delta}[i]\vec{\Delta}[j] \neq 0 \right]\Pr\left[\vec{\Delta}[i]\vec{\Delta}[j] \neq 0\right]
        = -\mathbb{E}\left[(-1)^X\right]\cdot\frac{k(k-1)}{T(T-1)}\,.
    \end{equation*}
    To see the last step, note that the probability of both $\vec{\Delta}[i]$ and $\vec{\Delta}[j]$ being non-zero equals the second factor, and that conditioning on this event, the sign of their product is decided by the parity of the number of non-zero elements in $\vec{\Delta}$ between $i$ and $j$, i.e., the parity of $X$.

    Invoking \Cref{lem:signed-hypergeom}, we arrive at the following bound.
    \begin{align*}
        \left\lvert\mathbb{E}\left[ \vec{\Delta}[i]\vec{\Delta}[j]\right]\right\rvert
        &\leq \frac{k(k-1)}{T(T-1)}\left(\exp\left(-\frac{2(k-2)(i-j-1)}{T-2}\right) + \frac{(k-2)(k-3)}{T-2}\right)\\
        &\leq \frac{k^2}{T^2}\exp\left(-\frac{2(k-2)(i-j-1)}{T-2}\right)  + \frac{k^4}{T^3} 
    \end{align*}
    We can now start working with the off-diagonal term.
    \begin{align*}
        \left\lvert2\sum_{j < i} c(i, j) \mathbb{E}\left[ \vec{\Delta}[i]\vec{\Delta}[j]\right]\right\rvert
        & \leq 2\sum_{j < i} c(i, j) \left\lvert\mathbb{E}\left[ \vec{\Delta}[i]\vec{\Delta}[j]\right]\right\rvert\\
        &\leq \underbrace{\frac{2k^2}{T^2}\sum_{j < i} \exp\left(-\frac{2(k-2)(i-j-1)}{T-2}\right)c(i, j)}_{S_{exp}}
        + \underbrace{\frac{2k^4}{T^3} \sum_{j < i} c(i,j)}_{S_{poly}}
    \end{align*}

    We will first argue that $S_{poly} = o(k\log k)$.
    Note that
    \begin{equation*}
        \sum_{j < i} c(i,j) \leq \sum_{i,j\in[T]} c(i,j) = \| \sqrt{A}\mathbf{1}\|_2^2 = O(T^2)\,.
    \end{equation*}
    The last step requires some computation, but follows from $(\sqrt{A}\mathbf{1})[\ell] = \sum_{t=0}^\ell r_t = O(\sqrt{\ell})$.
    Thus $S_{poly} = O(k^4/T) = O(k)$ from $k = O(T^{1/3})$.

    We proceed with $S_{exp}$.
    We will use that $c(i, j) \leq \frac{3}{2} + \frac{1}{\pi}\ln(1 + \frac{T}{\lvert i-j\rvert})$ from \Cref{lem:sqrt-terms}.
    \begin{align*}
        S_{exp} &= \frac{2k^2}{\pi T^2}\sum_{j < i} \exp\left(-\frac{2(k-2)(i-j-1)}{T-2}\right)\left(\frac{3\pi}{2} + \ln\left(1 + \frac{T}{i-j}\right)\right)\\
         &= \frac{2k^2}{\pi T^2}\sum_{j=0}^{T-2}\sum_{d=1}^{T-1-j} \exp\left(-\frac{2(k-2)(d-1)}{T-2}\right)\left(\frac{3\pi}{2} + \ln\left(1 + \frac{T}{d}\right)\right)\\
         &= \frac{2k^2}{\pi T^2}\sum_{d=1}^{T-1}(T-d)\exp\left(-\frac{2(k-2)(d-1)}{T-2}\right)\left(\frac{3\pi}{2} + \ln\left(1 + \frac{T}{d}\right)\right)\\
         &\leq \frac{2k^2}{\pi T}\sum_{d=1}^{T-1}\exp\left(-\frac{2(k-2)(d-1)}{T-2}\right)\left(\frac{3\pi}{2} + \ln\left(1 + \frac{T}{d}\right)\right)\\
         &= \frac{3k^2}{T}\underbrace{\sum_{d=1}^{T-1}e^{-\alpha(d-1)}}_{Q_1} + \frac{2k^2}{\pi T}\underbrace{\sum_{d=1}^{T-1} e^{-\alpha(d-1)} \ln\left(1 + \frac{T}{d}\right)}_{Q_2}\,,
    \end{align*}
    where $\alpha = \frac{2(k-2)}{T-2}$.
    We first bound $Q_1$:
    \begin{align*}
        Q_1 
        = \sum_{d=1}^{T-1}e^{-\alpha(d-1)}
        \leq \sum_{d=0}^{\infty}e^{-\alpha d}
        = \frac{1}{1-e^{-\alpha}}
        \leq 1 + \frac{1}{\alpha}
        = \Theta(T/k)\,.
    \end{align*}
    where the second inequality uses that $1-e^{-x} \geq x/(1+x)$ for any $x \geq 0$.
    For $Q_2$, we have to do more work.
    Let $f(x) = e^{-\alpha x}\ln\left(1+\frac{T}{x+1}\right)$, which is positive for $x \geq 0$.
    Computing its derivative,
    \begin{equation*}
        f'(x) = -e^{-\alpha x}\left[\alpha\ln\left(1+\frac{T}{x+1}\right) + \frac{T}{(x+1)^2 + T(x+1)}\right]\,,
    \end{equation*}
    we can verify that it is negative on $[0, \infty)$, and so $f$ is strictly decreasing on this range.
    We can thus bound $Q_2$ as
    \begin{align*}
        Q_2 &= \sum_{d=0}^{T-2} f(d)
         \leq \sum_{d=0}^{\infty} f(d)
         \leq f(0) + \int_{0}^{\infty} f(x)\,\mathrm{dx}
         = \ln(1+T) + \int_{0}^{\infty} e^{-\alpha x}\ln\left(1+\frac{T}{x+1}\right)\,\mathrm{dx}\,.
    \end{align*}
    We continue working with the integral:
    \begin{equation*}
         \int_{0}^{\infty} e^{-\alpha x}\ln\left(1+\frac{T}{x+1}\right)\,\mathrm{dx}
         \leq \int_{0}^{\infty} e^{-\alpha x}\ln(1+T/x)\,\mathrm{dx}
         \leq \frac{1}{\alpha}\int_{0}^{\infty} e^{-y}\ln(1+K/y)\,\mathrm{dy}\,,
    \end{equation*}
    where the last step makes the substitution $y=\alpha x$ and sets $K=\alpha T = \Theta(k)$.
    Re-writing the integrand, we have that
    \begin{equation*}
       \ln(1+K/y)  = \ln K - \ln y + \ln(1+y/K)\,.
    \end{equation*}
    Thus,
    \begin{equation*}
        \int_{0}^{\infty} e^{-y}\ln(1+K/y)\,\mathrm{dy}
        = \ln K \underbrace{\int_{0}^{\infty} e^{-y}\,\mathrm{dy}}_{=1} - \underbrace{\int_0^{\infty}e^{-y}\ln y\,\mathrm{dy}}_{=-\gamma} + \int_{0}^{\infty} e^{-y}\ln(1+y/K)\,\mathrm{dy}\,,
    \end{equation*}
    where $\gamma \approx 0.5772$ is the Euler-Mascheroni constant.
    For the last integral, we use the bound $\ln(1+u) \leq u$ for $u\geq 0$:
    \begin{equation*}
        \int_{0}^{\infty} e^{-y}\ln(1+y/K)\,\mathrm{dy}
        \leq \frac{1}{K}\int_{0}^{\infty} y e^{-y}\,\mathrm{dy} = \frac{1}{K}\,.
    \end{equation*}
    We can now bound $Q_2$:
    \begin{align*}
        Q_2 &\leq \ln(1+T) + \frac{1}{\alpha}\int_{0}^{\infty} e^{-y}\ln(1+K/y)\,\mathrm{dy}
        \leq \ln(1+T) + \frac{1}{\alpha}\left(\ln K + \gamma  + \frac{1}{K}\right)\,.
    \end{align*}
    As $\frac{2k}{T-2}\leq \alpha = \frac{2(k-2)}{T-2} \leq \frac{2k}{T}$ and so $2(k-2) \leq K = \alpha T \leq 2k$, we simplify the bound further:
    \begin{align*}
        Q_2 &\leq \frac{T-2}{2k}\bigg(\ln(2k) + \underbrace{\frac{2k\ln T}{T-2} + \gamma + \frac{1}{2(k-2)}}_{O(1)}\bigg)
        = \frac{T \ln k}{2k} + O(T/k)
    \end{align*}
    where $k\ln(T)/T = O(1)$ as $k = o(T)$.
    We proceed to bound $S_{exp}$:
    \begin{align*}
       S_{exp} \leq \frac{k^2}{\pi T}(3\pi Q_1 + 2 Q_2)
       \leq \frac{k^2}{\pi T}\left(\frac{T\ln k}{k} + O(T/k)\right)
       = \frac{k\ln k}{\pi} + O(k)\,.
    \end{align*}

    
    We are now ready to put all the pieces together.
    \begin{align*}
        \sens_2(R, \boundedvecs{1, k})^2
        &\geq \mathbb{E}\left[ \| R \vec{\Delta} \|_2^2 \right]
        = \sum_{i=0}^{T-1} c(i, i) \mathbb{E}\left[\vec{\Delta}[i]^2\right]
        + 2\sum_{j<i} c(i, j) \mathbb{E}\left[\vec{\Delta}[i]\vec{\Delta}[j]\right]\\
        &\geq \frac{k\ln(T)}{\pi} + O(k)
        - \left\lvert 2\sum_{j<i} c(i, j) \mathbb{E}\left[\vec{\Delta}[i]\vec{\Delta}[j]\right]\right\rvert\\
        &\geq \frac{k\ln(T)}{\pi} + O(k) - S_{exp} - S_{poly}\\
        &\geq \frac{k\ln(T)}{\pi} - \frac{k\ln(k)}{\pi} - O(k)\\
        &= (1-o(1))\cdot\frac{k\ln(T/k)}{\pi}\,.
    \end{align*}
    To get the final statement, take the square-root.
    \begin{equation*}
        \sens_2(R, \boundedvecs{1, k})
        \geq \sqrt{(1-o(1))\cdot\frac{k\ln(T/k)}{\pi}}
        = (1-o(1))\cdot\sqrt{\frac{k\ln(T/k)}{\pi}}\,.
    \end{equation*}
    
    Lastly, we address the more general settings of $k$ and $D$.
    Note that for $D=1, k=1$, the problem is trivial since $\sens_2(R, S_{1,1}) = \| R \|_{1\to 2}$.
    For $D=1, k=2$, a little more work is needed.
    The off-diagonal term can now be computed directly
    \begin{align*}
        2\sum_{j<i} c(i, j) \mathbb{E}\left[\vec{\Delta}[i]\vec{\Delta}[j]\right]
        = -\frac{2\cdot 2\cdot 1}{T(T-1)} \sum_{j<i} c(i, j)
        = O(1)
    \end{align*}
    where the final sum is $O(T^2)$ by the same argument as before, and this term is of lower order compared to the diagonal contribution of $k\ln(T)/\pi$.
    For $1 < D \leq k$, we invoke \Cref{lem:alternating-to-general-lower}:
    \begin{equation*}
        \sens_2(R, \boundedvecs{D, k})
        \geq D\cdot\sens_2(R, \boundedvecs{1, \lfloor k/D \rfloor})
        \geq \left(\frac{1}{\sqrt{\pi}}-o(1)\right)\cdot\sqrt{D^2 \lfloor k / D \rfloor \ln\left(T/\lfloor k/D\rfloor \right)}
    \end{equation*}
    which allows us to reduce to the previous setting where $D=1$.
    If additionally $k/D = \omega(1)$, then $\lfloor k/D \rfloor = (1-o(1))k/D$, and the expression can be simplified one step further.
    \begin{equation*}
        \sens_2(R, \boundedvecs{D, k})
        \geq \left(\frac{1}{\sqrt{\pi}}-o(1)\right)\cdot\sqrt{Dk \ln(D T/k)}\,.\qedhere
    \end{equation*}
\end{proof}

\begin{proof}[Proof of \Cref{cor:sqrt-lower-bound}]
    The corollary follows by multiplying the sensitivity  $\sens_2(R,\boundedvecs{D,k})$ by either $\| \sqrt{A} \|_{2\to\infty}$ or $\| \sqrt{A} \|_{F}/\sqrt{T}$, for $\maxse$ and $\meanse$ respectively, and then dividing by $\sqrt{2\rho}$.
    Since $\|\sqrt{A}\|_{2\to\infty}$ and $\|\sqrt{A}\|_{F} / \sqrt{T}$ are $(\frac{1}{\sqrt{\pi}}+o(1))\sqrt{\ln(T)}$ (See (2), (3) and (4) of \Cref{thm:sqrtprops}), the corollary statement is immediate.
\end{proof}



\section{Tree-based Factorizations}\label{sec:trees}

The first mechanism to achieve non-trivial error bounds for private continual counting was the \emph{binary tree mechanism}~\cite{dwork_differential_2010,chan_private_2011}.
Throughout this section we will consider the natural extension to $b$-ary tree mechanisms, which has been studied in prior work~\cite{qardaji_2013,cormode_range_queries_2019,cardoso_differentially_2022,AnderssonPST25}.
While our focus is on $\rho$-zCDP, as in the case of standard continual counting on $\boundedvecs{1,1}$-streams, tree-based factorizations also give state-of-the-art error for pure differential privacy.
Towards the end of this section, we give error bounds for both privacy variants.

In \Cref{sec:trees-intro} we introduce tree-based factorizations studied in the context of regular $\boundedvecs{1,1}$-counting.
This includes stating the necessary (known) bounds on the left factor, which we also prove for the sake of completeness.
Next in \Cref{sec:trees-alt-sens}, we investigate the sensitivity computation for $\boundedvecs{D,k}$-streams.
We show how, in the context of trees, the sensitivity computation with respect to the right factor for $\boundedvecs{1,k}$ reduces to a parity counting problem, which can then be extended to $\boundedvecs{D, k}$ using \Cref{thm:alternating-to-general}.
Finally, \Cref{sec:trees-error} states tight error bounds for tree-based factorizations on general $\boundedvecs{D, k}$-streams.

\subsection{An Introduction to \texorpdfstring{$b$-ary Tree Mechanisms}{b-ary tree mechanisms}}\label{sec:trees-intro}
We begin with a short summary in the case where the number of steps $T = b^h$ for some positive integer $h$.
Throughout we will use the notation
\begin{equation*}
    \cT_{b, h} = \bigcup_{\ell=0}^{h}\cT^{\ell}_{b,h}\quad\text{where}\quad \cT_{b,h}^\ell= \{ [j\cdot b^{\ell}, (j+1)\cdot b^{\ell}] :\ 0 \leq j \leq b^{h-\ell} - 1 \}\,,
\end{equation*}
for the \emph{complete $b$-ary tree of height $h$} with $m = (b^{h+1} - 1) / (b - 1)$ nodes.
The $b$-ary tree mechanism can be expressed as a factorization mechanism for matrices $L_b\in\{0, 1\}^{T\times m}, R_b\in\{0,1\}^{m\times T}$.
Defining the bijective node-to-index mapping $i : \cT_{b,h} \to [m]$, the factorization $L_b R_b = A$ satisfies 
\begin{enumerate}
    \item For each node $v\in\cT_{b,h} : (R\xvec)_{i(v)} = \sum_{j\in v} \xvec[j]$.
    \item For each time step $t\in[T]$, define $\cQ_b(t)$ as the (unique) minimal set of nodes $\cQ_b(t) \subseteq \cT_{b,h}$ satisfying
    \begin{equation*}
        \bigcup_{v\in \cQ_b(t)} v = [0, t]\qquad\text{and}\qquad\forall v,v'\in\cQ_b(t) : v\cap v' = \emptyset\quad\text{or}\quad v=v'\,.
    \end{equation*}
    Then for each $v\in\cQ_b(t) : L_b[t, i(v)] = 1$.
\end{enumerate}
Intuitively, the $b$-ary tree mechanism releases a noisy version of the nodes $R_b\xvec$, and decodes the tree into a prefix by adding together a minimal number of nodes per step.
As $\lvert \cQ_b(t)\rvert$ describes the number of nodes added together at time $t$, it follows that $\| L_b \|_{2\to\infty}$ and $\| L_b \|_F$ can be computed from the maximum and average value of $\lvert \cQ_b(t)\rvert$ respectively.
We sketch a proof of the following known lemma for completeness (see e.g.,~\cite{chan_private_2011,AnderssonPST25}).
\begin{lemma}\label{lem:simple-l-bounds}
    For $T = b^h$, $\| L_b \|_{2\to\infty} = \sqrt{(b-1)h}$ and $\| L_b \|_F / \sqrt{T} = \sqrt{(b-1)h/2 + b^{-h}}$\,.
\end{lemma}
We will use the following observation about $\cQ_b(t)$, see e.g., \cite{chan_private_2011,AnderssonP23}.
\begin{claim}\label{clm:enc}
    For $b\geq 2$, let $\wvec(i)\in[b]^{h}$ be the encoding of an integer $i\leq b^{h} - 1$ as an $h$-digit number in base $b$.
    Then $\lvert \cQ_b(t)\rvert = \| \wvec(t+1) \|_1$.
\end{claim}
\begin{proof}[Proof sketch of \Cref{lem:simple-l-bounds}.]\hspace{-2pt}
    Note that $\| L_b \|_{2\to\infty}^2 \hspace{-1pt}=\hspace{-1pt} \max_{t\in[T]}\lvert\cQ_b(t)\rvert$ and $\|L_b \|_F^2 \hspace{-1pt}=\hspace{-1pt} \sum_{t\in[n]} \lvert\cQ_b(t)\rvert$, and so the computation reduces to arguing about the maximum and average size of $\cQ_b(t)$.
    Both will be argued through \Cref{clm:enc}.
    The result for $\| L_b \|_{2\to\infty}$ follows by noting that the encoding $\wvec(b^h - 1)$ has $h$ digits with value $b-1$.
    For $\| L_b \|_F$, note that $\wvec(t+1)$ cycles through all $h$-digit numbers in the $T$ steps, except for $\wvec(0)$, plus the $(h+1)$-digit number $\wvec(T)$ which is zero except for a most significant digit of~1.
    As the average $h$ digit number has an $\ell_1$ weight of $h(b-1)/2$, we get $\sum_{t\in[T]} \lvert \cQ_b(t)\rvert = hT(b-1)/2 + 1$.
    Dividing by $T$ and taking the root finishes the proof.
\end{proof}
Combining this result with the fact that $\sens_2(R_b, \boundedvecs{1,1}) = \sqrt{h+1}$ allows for recovering the standard $O(\log T)$ bounds on $\meanse(L_b, R_b)$ and $\maxse(L_b, R_b)$ for regular continual counting.
By tightly analyzing $\sens_2(R_b, \boundedvecs{D, k})$, our goal will be to extend these error bounds to $\boundedvecs{D,k}$-streams.

\subsubsection{Constant Factor Improvement via Subtraction.} However, before doing that, we note that on fixing the right factor $R_b$, the choice of left factor $L_b$ is not unique.
Recent work~\cite{AnderssonPST25} suggested a different choice of left factor that improves leading constants at no cost in time or memory efficiency---they remain $O(T)$ and $O(\log T)$ for constant $b$.
While they studied pure differential privacy, the same idea extends to $\rho$-zCDP.
We summarize their construction next.

Let the branching factor $b$ be odd.
On each level, their construction either adds $\leq (b-1)/2$ left-most children or \emph{subtracts} $\leq (b-1)/2$ right-most children to form a prefix.
The root of the tree can only be added, and the middle child of each parent is never used.
This uniquely identifies the nodes used at each step.
Denoting the (unique) set of nodes used for the prefix at time $t$ by the union of the nodes added and subtracted, $\hat{\cQ}_b(t) = \hat{\cQ}_b^+(t) \cup \hat{\cQ}_b^-(t)$, we have that
\begin{equation*}
    (A\xvec)[t] = \sum_{i=0}^t \xvec[i] = \sum_{v \in \hat{\cQ}_b^+(t)}\sum_{i\in v} \xvec[i] - \sum_{v \in \hat{\cQ}_b^-(t)}\sum_{i\in v} \xvec[i]\,.
\end{equation*}
Each row of $\hat{L}_b$ is thus defined by $\hat{\cQ}_b(t)$ where $\hat{L}_b[t, i(v)] = 1$ if $v\in\hat{\cQ}_b^+(t)$, equals $-1$ if $v\in\hat{\cQ}_b^-(t)$ and is otherwise zero.
The following lemma is adapted from~\cite{AnderssonPST25}.
\begin{lemma}\label{lem:adv-l-bounds}
    Let $b\geq 3$ be odd and $T = b^h$ for $h\geq 1$.
    Then
    \begin{equation*}
       \| \hat{L}_b \|_{2\to\infty} = \sqrt{\frac{(b-1)h + 2}{2}}\,,\qquad \frac{\| \hat{L}_b \|_F}{\sqrt{T}} = \frac{\sqrt{b(1-1/b^2)h + 2(1+b^{-h})}}{2}\,.
    \end{equation*}
\end{lemma}
For proving it, we will use the following claim from \cite{AnderssonPST25}, relating the number of nodes used per step to the $\ell_1$ weight of \emph{offset} encodings in $b$-ary.
\begin{claim}\label{clm:enc-adv}
    For odd $b\geq 3$, let $\hat{\wvec}(i)\in[-(b-1)/2, (b-1)/2]^{h}$ be the encoding of an integer $i\leq (b^{h} - 1)/2$ as an $h$-digit number in base $b$ with offset digits.
    Then $\lvert \hat{\cQ_b}(t)\rvert = \| \hat{\wvec}(t+1) \|_1$.
\end{claim}
\begin{proof}[Proof of \Cref{lem:adv-l-bounds}]
    As before, the arguments reduces to analyzing
    \begin{equation*}
        \max_{t\in[T]} \lvert\hat{\cQ}_b(t)\rvert\quad \text{and}\quad\frac{1}{T}\sum_{t\in[T]} \lvert\hat{\cQ}_b(t)\rvert\,.
    \end{equation*}
    For the maximum number of nodes, it is equal to $1 + h(b-1)/2$, and takes place when the root is included and every level below subtracts $(b-1)/2$ nodes at $t=T-2$.
    The average number of nodes requires a more elaborate argument, but follows almost directly from the proof of Lemma~14 in \cite{AnderssonPST25}.
    By \Cref{clm:enc-adv}, the number of nodes is again related to the $\ell_1$ weight of encodings of $1,\dots, T$, but now they are encoded as $h+1$-digit numbers with digits in $[-(b-1)/2, (b-1)/2]$.
    Over all $T$ steps, the least significant $h$ digits will cycle through every expressible number in $h$ digits \emph{exactly} once.
    The average $\ell_1$ weight of a single digit is $\frac{b(1-1/b^2)}{4}$, so these lower digits contribute $\frac{b(1-1/b^2)h}{4}$ on average over $T$ steps.
    Taking into account that the root of the tree is only used for the last $(b^h + 1) / 2$ steps, we arrive at a final average number of nodes
    \begin{equation*}
       \frac{b(1-1/b^2)h}{4} + \frac{(b^h + 1) / 2}{b^h} = \frac{b(1-1/b^2)h + 2(1 + b^{-h})}{4}\,,
    \end{equation*}
    finishing the proof.
\end{proof}

For our final error bounds, we will also want to consider the case where the numbers of steps $T$ is not an exact integer power of $b$.
Here we consider $h = \lceil \log_b T \rceil$, and define $\hat{L}_b$ as only the first $T$ rows of the corresponding matrix for $T' = b^h$ steps.
The existence of bounds in this case is implicit from~\cite{AnderssonPST25}, but they are not formally proved in that work.
We do so next.
\begin{corollary}\label{cor:adv-l-bounds}
    Let $b\geq 3$ be odd, $T\geq b$, and $h = \lceil \log_b(T) \rceil$ such that $b^{h-1} < T \leq b^h$.
    Then
    \begin{equation*}
       \| \hat{L}_b \|_{2\to\infty} \leq \sqrt{\frac{(b-1)h + 2}{2}}\,,
       \qquad \frac{\| \hat{L}_b \|_F}{\sqrt{T}} \leq \frac{\sqrt{b(1-1/b^2)h + 2(1+ b^2 + b^{-h})}}{2}\,.
    \end{equation*}
    If additionally $T \geq b^{2b}$ we also have the lower bounds
    \begin{equation*}
       \| \hat{L}_b \|_{2\to\infty} \geq \sqrt{\frac{(b-1)(h-1) + 2}{2}}\,,
       \qquad \frac{\| \hat{L}_b \|_F}{\sqrt{T}} \geq \frac{\sqrt{b(1-1/b^2)(h-1) + 2(1 - b^2 + b^{-h})}}{2}\,.
    \end{equation*}
\end{corollary}
\begin{proof}
    We begin with the maximum number of nodes.
    It is easy to see that we can lower and upper bound this number by the number of nodes used at time $t' = b^{h-1} - 2 < T - 1$ and $t=b^h - 2$ respectively.
    Each corresponding to adding the root and subtracting $(b-1)/2$ nodes per level for a tree of height $h-1$ and $h$.
    Thus
    \begin{equation*}
        1 + \frac{(h-1)(b-1)}{2}
        \leq \max_{t\in[T]} \lvert \hat{Q}_b(t) \rvert
        \leq 1 + \frac{h(b-1)}{2}\,,
    \end{equation*}
    and the final bound follows from taking a square-root.
    
    For the mean error, we will argue about how much the contribution at each level can deviate from its average.
    As already stated, the lowest $h$ digits contribute $b(1-1/b^2)/4$ nodes per step on average over $b^h$ steps.
    Note that each of these $h$ digits appear cyclically, where the $\ell$\textsuperscript{th} least significant digit appears with period $b^{\ell+1}$.
    By choosing $T$ adversarially, we could end up with one incomplete period contributing more or less than the average.
    It follows that the average number of nodes contributed by this digit can shift the average by at most $0.5(b-1)b^{\ell+1} / T < 0.5(b-1)b^{\ell-h + 2}$.
    The total shift in average $\ell_1$ weight over the lowest $h$ digits is therefore $< \sum_{\ell=0}^{h-1}0.5(b-1)b^{\ell - h + 2} = b^2/2$.

    For the upper bound, note that the root of the tree is only used for the last $(b^h + 1)/2$ steps in the complete tree, and so its average contribution for the full tree already serves as an upper bound.
    Thus:
    \begin{equation*}
        \frac{1}{T}\sum_{t\in[T]} \hat{Q}_b(t)
        \leq \frac{b^{2}}{2} + \frac{1}{b^{h}}\sum_{t\in[b^h]} \hat{Q}_b(t)
        = \frac{b(1-1/b^2)h + 2(1 + b^2 + b^{-h})}{4}\,.
    \end{equation*}
    For the lower bound, we take the exact average at time $t=b^{h-1}$, and again use that the total shift in average $\ell_1$ weight at time $T$ is at most $b^2/2$, i.e.,
    \begin{equation*}
        \frac{1}{T}\sum_{t\in[T]} \hat{Q}_b(t)
        \geq -\frac{b^{2}}{2} + \frac{1}{b^{h-1}}\sum_{t\in[b^h]} \hat{Q}_b(t)
        = \frac{b(1-1/b^2)(h-1) + 2(1 - b^2 + b^{-(h-1)})}{4}\,.
    \end{equation*}
    The right-hand side can be checked to be positive for $h \geq 1 + 2b$. 
    This is satisfied when $T > b^{h-1} \geq b^{2b}$.
    Taking square-roots on each of the bounds gives the statement.
\end{proof}
\Cref{cor:adv-l-bounds} essentially says that the scaling in \Cref{lem:adv-l-bounds} holds when you relax to arbitrary $T$.
For sufficiently large $T$ and constant $b$, both expressions are $\Theta(\sqrt{\log(T)})$ with explicit, tight leading constants.


\subsection{Sensitivity via Parity Counting}\label{sec:trees-alt-sens}
Having stated bounds for the left factor matrix, we shift our focus to the challenge unique to our setting: computing the sensitivity $\sens_2(R_b, \boundedvecs{D, k})$.
Given the combinatorial structure of trees, it will prove tractable to study the problem for $\boundedvecs{1, k}$, where the sensitivity vectors are in $\{0, \pm 1\}^{T}$, alternating and $k$-sparse.
After giving tight bounds in this case, we extend them to $\boundedvecs{D, k}$ via \Cref{thm:alternating-to-general}.
Previous work \cite[Lemma B.1]{JainKRSS23} established a bound of
\begin{equation*}
   \sens_2(R_2, \boundedvecs{1, k}) \leq \sqrt{k(1 + \log T)}
\end{equation*}
for the binary tree and its $\ell_2$ sensitivity (after translating their neighboring relation to our setup).
We will tighten this bound asymptotically to $\Theta(\sqrt{k\log_b(T/k)})$ for a $b$-ary tree.

We take a combinatorial view, reducing the problem to parity counting in the tree.
For now, assume that $T = b^h$ for a positive integer $h$ so that we are working with a complete tree $\cT_{b, h}$.
Call the subset $\cB \subseteq [T]$ a configuration of \emph{balls}.
For any node $v\in \mathcal{T}_{b, h}$, we consider its subtree a \emph{bin}, and define $o_v(\cB) = \lvert \cB \cap v \rvert \pmod 2$ as the indicator variable for if node $v$ has an odd number of balls in its bin.
Define
\begin{equation*}
    F_b(h, k) = \max_{\cB \subseteq\{0,1\}^T :\, \lvert \cB \rvert = k} \sum_{v\in \mathcal{T}_{b,h}} o_v(\cB)\,,
\end{equation*}
that is, $F_b(h, k)$, is the maximum number of odd nodes in $\mathcal{T}_{b,h}$ taken over all assignments of $k$ balls.
The following lemma is central to our analysis.
\begin{lemma}\label{lem:parity-is-sens}
    For $p \geq 1$, $\sens_p(R_b, \boundedvecs{1,k}) = \sqrt[\uproot{2}p]{\max_{k'\leq k} F_b(h, k')}$.
\end{lemma}
\begin{proof}
    Recall that $\sens_p(R_b, \boundedvecs{1,k}) = \max_{\dvec\in\boundedvecs{1,k}} \| R_b \dvec \|_p$.
    For $\dvec\in\boundedvecs{1,k}$, observe that $R_b\dvec \in \{0,\pm 1\}^{m}$, due to each node in the tree being an interval sum on $\dvec$, and $\dvec$ having non-zero entries of alternating sign.
    Setting $\cB$ equal to the non-zero indices of $\dvec$ we therefore have that
    \begin{equation*}
        \| R_b\dvec \|_p^p = \sum_{v\in\cT_{b, h}} \lvert (R_b\dvec)[i(v)] \rvert^{p} = \sum_{v\in\cT_{b, h}} o_v(\cB)\,,
    \end{equation*}
    and so $\sens_p(R_b, \boundedvecs{1,k})^p = \max_{\cB\in\{0,1\}^T : \lvert \cB \rvert \leq k} \sum_{v\in\cT_{b,h}} o_v(\cB) = \max_{k'\leq k} F_b(h, k')$.
\end{proof}
In essence, parity counting exactly describes the sensitivity of the tree.
We next show how $F(h, k)$ can computed using dynamic programming.
\begin{lemma}[A dynamic program]\label{lem:full-tree-dp}
    Let $h\ge 1$ and $k\le b^{h}$.
    Then
    \begin{equation*}\label{eq:tree‑recurrence}
        F_b(h, k) = \mathbbm{1}_{[k\,\mathrm{odd}]} + \max_{\substack{k_1 + \dots + k_b = k\\ \forall i : 0 \leq k_i \leq b^{h-1}}}\left\{\sum_{i=1}^b F_b(h-1, k_i)\right\}\,,\quad \mathrm{where} \quad F_b(0, k) = k\,.
    \end{equation*}
    Moreover, $F_b(h, k)$ can be computed as a dynamic program with table size $O(k)$ and in time $O(hk^{b-1})$.
\end{lemma}
\begin{proof}[Proof sketch]
   The correctness of the recurrence follows immediately from the problem description.
   To maximize the number of odd nodes, we can enumerate over all possible assignments of balls to the root's children and recurse, and any leaf can contribute at most~1 to the final count.
   The recurrence can be solved as a dynamic program \enquote{bottom-up}, and would require $O(k)$ memory per level which can be re-used across levels.
   For the time-complexity, observe that the maximization is enumerating over all integer partitions of $k$ into $b$ (non-negative) integer parts.
   We can upper bound the number of such integer partitions by the number of weak compositions of $k$ into $b$ non-negative parts: $\binom{k+b-1}{b-1} = O(k^{b-1})$ (for constant $b$).
   This overcounts the number of integer partitions by at most a constant factor $b!\,$.
   Performing this maximization over all levels gives the stated $O(h k^{b-1})$ time complexity.
\end{proof}
We remark that estimating the sensitivity of a tree via dynamic programming is a natural idea, and that similar ideas have been explored in the related (but distinct) context of multi-participation~\cite{kairouz_practical_2021}.
While the time complexity is polynomial in $k$ for constant $b$, it quickly becomes infeasible once larger $b$ are considered. 
Motivated by this, we tightly bound $F_b(h, k)$ next.
\begin{lemma}\label{lem:tree-bounds}
   Let $b\geq 2, h \geq 1, 1 \leq k \leq b^h$ and $\alpha_k = \lceil \log_b k \rceil$. Then
   \begin{equation*}
       k(h - \alpha_k + 1) \leq F_b(h, k) \leq k(h-\alpha_k + 1) + \frac{b^{\alpha_k} - 1}{b-1}\,,
   \end{equation*}
   with the upper bound monotonically increasing in $k$.
\end{lemma}
\begin{proof}
   By definition, we have that level $\ell^* = h - \alpha_k$ is the highest level in the tree with $k$ or more nodes.
   This follows from $b^{h - (\ell^* - 1)} = b^{\alpha_k-1} < k \leq b^{\alpha_k} = b^{h-\ell^*}$.
   For the upper bound, observe that level $\ell$ in the tree can have at most $\min(k, b^{h-\ell})$ odd nodes, thus
   \begin{align*}
       F_b(h, k) &\leq \sum_{\ell=0}^{h}\min(k, b^{h-\ell})
       = \sum_{\ell=0}^{\ell^*} k + \sum_{\ell=\ell^*+1}^{h} b^{h-\ell}
       = k(h-\alpha_k+1) + \frac{b^{\alpha_k} - 1}{b - 1}\,,
   \end{align*}
   where the monotonicity in $k$ is immediate from the expression after the first inequality.
   For the lower bound, we assign balls to leaves by adding at most one ball to each subtree with root on level $\ell^*$, earning a contribution of $k$ from every level up to $\ell^*$:
   \begin{equation*}
       F_b(h, k) \geq  k(\ell^* + 1) = k(h-\alpha_k+1)\,,
   \end{equation*}
   and we are done.
\end{proof}

\begin{corollary}\label{cor:sens-tree-bounds}
   Let $b\geq 2, h \geq 1, 1 \leq k \leq b^h$ and $\alpha_k = \lceil \log_b k \rceil$.
   Then
   \begin{equation*}
       \max\left(k(h-\alpha_k +1), b^{\alpha_k - 1}(h-\alpha_k + 2)\right)
       \leq \sens_p(R_b, \boundedvecs{1, k})^p
       \leq k(h-\alpha_k + 1) + \frac{b^{\alpha_k} - 1}{b-1}\,.
   \end{equation*}
\end{corollary}
\begin{proof}
    By \Cref{lem:parity-is-sens}, we have that $\sens_p(R_b, \boundedvecs{1, k})^p = \max_{k'\leq k} F_b(h, k')$.
    Invoking the upper bound in \Cref{lem:tree-bounds}, we get
    \begin{equation*}
        \sens_p(R_b, \boundedvecs{1, k})^p
        \leq \max_{k'\leq k}\left\{k'(h-\alpha_{k'}+1) + \frac{b^{\alpha_{k'}} - 1}{b - 1}\right\}
        = k(h-\alpha_k+1) + \frac{b^{\alpha_k} - 1}{b - 1}\,,
    \end{equation*}
    \janote{again, systematize the use of curly or parenthesis for min/max.}
    where the last step follows from the monotonicity in $k'$.
    Analogously for the lower bound, we get
    \begin{equation*}
        \sens_p(R_b, \boundedvecs{1, k})^p
        \geq \max_{k'\leq k} k'(h-\alpha_{k'}+1)
        = \max_{k' \leq k} L(k')
    \end{equation*}
    but in this case $L(k')$ is not monotonically increasing in $k'$.
    Nevertheless, it is monotonically increasing on each of the intervals $(b^{j-1}, b^{j}]$ for integer $j\geq 1$, where it attains the maximum $M_j = b^{j}(h-j+1)$.
    It is easy to check that $M_j$ is non-decreasing in $j$, e.g.,
    \begin{equation*}
        \frac{M_{j+1}}{M_j} = b\cdot \frac{h-j}{h-j+1} \geq 2\cdot \frac{1}{2} \geq 1\,.
    \end{equation*}
    It follows that $L(k')$ is maximized either at $k'=k$, or at $k'=b^{\alpha_k-1}$, proving the statement.
\end{proof}
When $k=o(T)$ the second term in the upper bound is of lower order, in which case we have shown that
\begin{equation*}
    \sens_p(R_b \boundedvecs{1, k})^p = \left(1+o(1)\right)\cdot k(h - \lceil \log_b k\rceil + 1)\,,
\end{equation*}
which asymptotically behaves like $\Theta(k\log(T/k))$.

\subsubsection{A Sensitivity Reduction Trick.}
We note that it is possible to improve lower order terms further at no additional cost.
The key observation is that most tree-based factorizations in fact only use (roughly) a $1 - 1/b$ fraction of the nodes in the tree, so the same fraction of the odd node counts in $F_b(h, k)$ can be ignored in the sensitivity computation.
For example, the standard $b$-ary tree mechanism never uses the right-most child of any parent; the subtraction variant never uses the middle-child.
Alternatively phrased, a $1-1/b$ fraction of columns in $\hat{L}_b$ (and $L_b$) are identically zero, and so we can zero-out the corresponding rows in $R_b$, producing $\hat{R}_b$, to get a new factorization $\hat{L}_b \hat{R}_b = A$ with lower sensitivity.
This optimization is unique to our setting, and does not yield any improvement for standard $\boundedvecs{1,1}$-continual counting.

In what follows we derive tight bounds for $\hat{R}_b$.
The analysis is near-identical to the preceding case of the full tree, with \Cref{cor:sens-tree-bounds-adv} being the analog of \Cref{cor:sens-tree-bounds}.
Let $\hat{F}(h, k)$ be the corresponding odd node count where only $b-1$ out of $b$ children can contribute.
For completeness, we give a matching dynamic program for $\hat{F}(h, k)$.
\begin{lemma}[A second dynamic program]\label{lem:partial-tree-dp}
    Let $h\geq 1$, $k \leq b^{h}$ and $w\in\{0,1\}$.
    Define $F_b(h, k, w)$ via
    \begin{align*}\label{eq:tree‑recurrence}
        F_b(h, k, w) &= w\cdot\mathbbm{1}_{[k\,\mathrm{odd}]}
        + \max_{\substack{k_1 + \dots + k_b = k\\
        \forall i : 0 \leq k_i \leq b^{h-1}}}\left\{ F_b(k_1, h-1, 0) + \sum_{i=2}^b F_b(k_i, h-1, 1)\right\}\,,\\
        &\quad F_b(k, 0, w) = w\cdot k\,.
    \end{align*}
    Then for real $p\geq 1$, $\sens_p(\hat{R}_b, \boundedvecs{1,k}) = \sqrt[\uproot{2}p]{\max_{k'\leq k}\hat{F}_b(h, k')} = \sqrt[\uproot{2}p]{\max_{k'\leq k} F_b(k, h, 1)}$.
    Moreover, $F_b(k, h, 1)$ can be computed as a dynamic program with a table of size $O(k)$ and in time $O(hk^{b-1})$.
\end{lemma}
\begin{proof}[Proof sketch]
    Observe that the recurrence is similar to the one from \Cref{lem:full-tree-dp}, except that now there is an additional variable $w$.
    The variable is set to~1 whenever the corresponding node contributes to the total odd node count, and is otherwise set to~0.
    Exactly $b-1$ out of the $b$ children of any node are assigned $w=1$, and the $w\cdot\mathbbm{1}_{[k\,\mathrm{odd}]}$ will only contribute for any node with $w=1$.
    That $\sens_p(\hat{R}_b, \boundedvecs{1,k})^p = \max_{k'\leq k}\hat{F}_b(h,k')$ follows from an analogous argument to \Cref{lem:parity-is-sens}: $\hat{R}_b\dvec$ can still only take on values in $\{0, \pm 1\}$, so maximizing the odd node count for the reduced tree implies maximizing the norm.

    For the dynamic program implementation, note that the number of weakly (ordered) compositions of $k$ into $b$ parts is still an upper bound on the number of configurations to optimize over for a fix $k,b$.
    Moreover, the introduction of the variable $w$ blows up the table size by at most a factor 2.
    It follows that the space and time usage is asymptotically no different from that of \Cref{lem:full-tree-dp}.
\end{proof}
Before formally bounding $\hat{F}(h, k)$, we first empirically showcase the improved sensitivity in \Cref{fig:tree-comp}.
\begin{figure}[th]
    \centering
    \subcaptionbox{Comparison of odd node count.\label{fig:odd-node-comp}}{\includegraphics[width=0.49\linewidth]{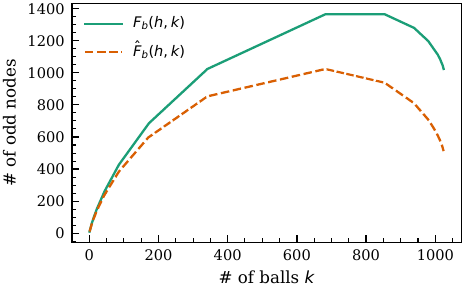}}%
    \hfill
    \subcaptionbox{Comparison of $\ell_2$ sensitivity.\label{fig:tree-sens-comp}}{\includegraphics[width=0.49\linewidth]{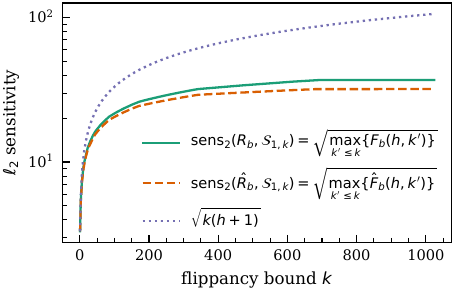}}%
    \hfill
    \caption{
        Comparison of the odd node count and the corresponding $\ell_2$ sensitivities for $b=2$ and $h=10$.
        $F_b(h, k)$ and $\hat{F}_b(h, k)$ are computed as dynamic programs using the formulations in \Cref{lem:full-tree-dp}~and~\ref{lem:partial-tree-dp}.
        \Cref{fig:odd-node-comp} directly plots the value of $F_b(h, k)$ and $\hat{F}_b(h, k)$, and \Cref{fig:tree-sens-comp} relates them to compute the exact sensitivity of $R_b$ and $\hat{R}_b$ respectively.
        The top green line in \Cref{fig:tree-sens-comp} is using the upper bound from \cite[Lemma B.1]{JainKRSS23}, translated to our setting.\label{fig:tree-comp}
    }
\end{figure}
Our treatment offers substantial an substantial improvement over the bound from \cite{JainKRSS23}, but the benefit from using a \enquote{reduced} factorization is subtle.
Moreover, the reduction is expected to be greater for smaller values of $b$, as then the fraction of nodes in the tree (rows) that can be discarded from $R_b$ (roughly $1/b$), is greater.
It turns out that this improvement, even for a binary tree, is of lower order unless $k=\Omega(T)$.
We show this next.
\begin{lemma}\label{lem:tree-bounds-adv}
   Let $h \geq 1$, $b\geq 2, 1 \leq k\leq b^{h}$ and $\alpha_k = \lceil \log_b k \rceil$, then
   \begin{equation*}
       k(h - \alpha_k + 1 - 1/b) \leq \hat{F}_b(h, k) \leq k(h - \alpha_k + 1) + b^{\alpha_k - 1}\,,
   \end{equation*}
   with the the upper bound monotonically increasing in $k$.
\end{lemma}
\begin{proof}
    For $k=1$, note that $\hat{F}_b(h,1) = h + 1$ satisfies both bounds.
    The rest of the proof focuses on $k\geq 2$.
    By our choice of $\alpha_k$, as in \Cref{lem:tree-bounds}, we again have that level $\ell^* = h - \alpha_k$ is the highest level in the tree with $k$ or more nodes.
    The restriction on $k$ enforces that $\ell^* < h$.
    For the upper bound, level $\ell < h$ in the tree can have at most $\min(k, (b-1)b^{h-\ell-1})$ odd nodes, thus
    \begin{align*}
        \hat{F}_b(k, h) &\le
        1 \hspace{-1pt}+\hspace{-1pt} \sum_{\ell=0}^{h-1}\min(k, (b-1)b^{h-\ell-1})
        = 1 \hspace{-1pt}+\hspace{-1pt} \sum_{\ell=0}^{\ell^*} k + (b-1)\sum_{\ell=\ell^*+1}^{h-1} b^{h-\ell-1}
        = k(h-\alpha_k+1) \hspace{-1pt}+\hspace{-1pt} b^{\alpha_k-1}
        \,.
    \end{align*}
    The monotonicity of the upper bound follows from its form after the first inequality.
    For the lower bound, analogously to \Cref{lem:tree-bounds}, we assign $k$ balls to each subtree rooted on level $\ell^*$, and we choose each leaf such that we get a contribution of $k$ from every level $< \ell^*$.
    On level $\ell^*$, however, we get a contribution of $\min(k, (b-1)b^{\alpha_k-1})$ since there might be fewer than $k$ contributing nodes.
    Thus,
    \begin{equation*}
        \hat{F}_b(h, k) \geq k(h-\alpha_k) + \min(k, (b-1)b^{\alpha_k-1}) \geq k(h-\alpha_k + 1 -1/b)\,,
    \end{equation*}
    where the last step uses that $(b-1)b^{\alpha_k-1} \geq (1-1/b)k$.
\end{proof}
\begin{corollary}\label{cor:sens-tree-bounds-adv}
   Let $b\geq 3, h \geq 1, 1 \leq k \leq b^h$ and $\alpha_k = \lceil \log_b k \rceil$.
   Then
   \begin{equation*}
       \max\left(k(h-\alpha_k +1-1/b), b^{\alpha_k - 1}(h-\alpha_k + 2-1/b)\right)
       \leq \sens_p(\hat{R}_b, \boundedvecs{1, k})^p
       \leq k(h-\alpha_k + 1) + b^{\alpha_k - 1}\,.
   \end{equation*}
\end{corollary}
\begin{proof}
    From \Cref{lem:partial-tree-dp}, we have that $\sens_p(\hat{R}_b, \boundedvecs{1, k})^p = \max_{k'\leq k} \hat{F}_b(h, k')$.
    Invoking the upper bound in \Cref{lem:tree-bounds-adv}, we have that
    \begin{equation*}
        \sens_p(\hat{R}_b, \boundedvecs{1, k})^p
        \leq \max_{k'\leq k}\left\{k'(h-\alpha_{k'}+1) + b^{\alpha_{k'} - 1}\right\} 
        = k(h-\alpha_k+1) + \frac{b^{\alpha_k} - 1}{b - 1}\,,
    \end{equation*}
    which follows from the maximization objective being monotonic in $k'$.
    Analogously for the lower bound, we get
    \begin{equation*}
        \sens_p(R_b, \boundedvecs{1, k})^p
        \geq \max_{k'\leq k} k'(h-\alpha_{k'}+1-1/b)
        = \max_{k' \leq k} L(k')
    \end{equation*}
    but in this case $L(k')$ is not monotonically increasing in $k'$.
    Nevertheless, it is monotonically increasing on each of the intervals $(b^{j-1}, b^{j}]$ for integer $1 \leq j \leq h$, where it attains the maximum $M_j = b^{j}(h-j+1)$.
    We directly verify that $M_j$ is non-decreasing in $j$,
    \begin{equation*}
        \frac{M_{j+1}}{M_j} = b\cdot \frac{h-j - 1/b}{h-j+1 -1/b} \geq b \cdot \frac{1-1/b}{2-1/b} \geq 3\cdot \frac{1-1/3}{2-1/3} = 6/5\,,
    \end{equation*} 
    where we have used that $b\geq 3$.
    It follows that $L(k')$ is maximized either at $k'=k$, or at $k'=b^{\alpha_k-1}$, proving the statement.
\end{proof}
Note that the corollary does not support $b=2$ (the upper bound does, however).


\subsubsection{Bounds for General $T$.}
Before proceeding, we also give bounds for the sensitivity in the case where $T$ is not an integer power of $b$.
We let $h = \lceil \log_b T \rceil$, and drop the last $b^h - T$ rows of $\hat{L}_b$ such that $\hat{L}_b \hat{R}_b = A = \mathbb{R}^{T\times T}$.
Note that this may lead to additional zero-columns, in $\hat{L}$ which allows for zeroing-out more entries in $\hat{R}_b$---we will not leverage this here, but note that it can yield an additional (lower order) reduction for certain values of $k$ and $T$.
\begin{corollary}\label{cor:sens-tree-bounds-adv-all-t}
    Let $3 \leq b \leq T$, $1 \leq k\leq T$, $h = \lceil \log_b T \rceil$ and $\tilde{k} = \min(k, b^{h-1})$.
    Then
    \begin{equation*}
       \big[\tilde{k}(\lceil \log_b T \rceil - \lceil \log_b \tilde{k}\rceil - 1/b)\big]^{1/p}
       \leq \sens_p(\hat{R}_b, \boundedvecs{1, k})
       \leq \big[k(\lceil \log_b (T/k)\rceil + 2)\big]^{1/p}\,.
    \end{equation*}
    If additionally $k \leq T/b$, then we get the following simplified lower bound
    \begin{equation*}
        \sens_p(\hat{R}_b, \boundedvecs{1, k})
        \geq \left[k(\lfloor \log_b(T/k) \rfloor - 1/b)\right]^{1/p}\,.
    \end{equation*}
\end{corollary}
\begin{proof}
    For the upper bound, note that the sensitivity cannot decrease if we consider the full tree over $b^h \geq T$ leaves, and so we can immediately invoke the upper bound in \Cref{cor:adv-l-bounds}.
    Thus
    \begin{align*}
        \sens_p(\hat{R}_b, \boundedvecs{1, k})^p
        &\leq k(\lceil \log_b T \rceil - \lceil \log_b k \rceil + 1) + b^{\lceil \log_b k\rceil-1}
        \leq k(\lceil \log_b (T/k)\rceil + 2)
    \end{align*}
    where the last inequality uses that $\lceil a \rceil - \lceil b \rceil \leq \lceil a - b\rceil$ and $\lceil a \rceil \leq a + 1$ for arbitrary $a,b\geq 0$.
    For the lower bound, we consider the full tree over $b^{h-1} < T$ leaves and let its associated (reduced) matrix be $\hat{R}_b'\in\mathbb{R}^{m\times b^{h-1}}$ for some $m \geq 1$.
    Define $\tilde{k} = \min(k, b^{h-1})$.
    Maximizing the norm over a smaller tree and over a restricted range of $k$ necessarily yields a smaller norm, and so
    \begin{align*}
        \sens_p(\hat{R}_b, \boundedvecs{1, k})^p
        &\geq \sens_p(\hat{R}_b', \boundedvecs{1, \hat{k}})
        \geq \tilde{k}(\lceil \log_b T \rceil - \lceil \log_b \tilde{k}\rceil - 1/b)
    \end{align*}
    where we invoke the lower bound in \Cref{cor:adv-l-bounds} and chose the first entry in the maximization.
    This lower bound is also strictly positive.
    For the more restricted case where $k \leq T/b$, we get that $\tilde{k} = k$ and our lower bound can be simplified further
    \begin{align*}
        \sens_p(\hat{R}_b, \boundedvecs{1, k})^p
        \geq k(\lceil \log_b T \rceil - \lceil \log_b k\rceil - 1/b)
        \geq k(\lfloor \log_b (T/k) \rfloor - 1/b)
    \end{align*}
    where we have used $\lceil a \rceil - \lceil b \rceil \geq \lfloor a-b \rfloor$.
    For this case we have $\log_b(T/k) \geq 1$, and so the lower bound is guaranteed to be a positive quantity.
    Taking the $p$-th root over all (positive) bounds give the final statement.
\end{proof}
    As long as $k \leq T/b$ (a practical regime), we have shown $\sens_p(\hat{R}_b, \boundedvecs{1, k})^p = \Theta(k \log(T/k))$.
    We will later see that for $k = \Omega(T)$, there is a better choice of factorization.

\subsubsection{Extending to $\boundedvecs{D, k}$.}
We are now ready to give our final result for the $\ell_p$ sensitivity of tree-based factorization over $\boundedvecs{D, k}$.
For convenience, we only state it for the case where $k \leq T/b$ and the lower bound has a (relatively) clean form.
We also assume $D \leq k$, since otherwise $\boundedvecs{D, k}$ collapses to the set of vectors with only bounded $\ell_1$ norm.
\begin{theorem}\label{thm:tight-sens-bound-trees}
    Let $3 \leq b \leq T$ and $1 \leq D \leq k \leq T/b$
    Then
    \begin{equation*}
       D\left[\left\lfloor \frac{k}{D} \right\rfloor\left(\left\lfloor \log_b\left(\frac{DT}{k}\right) \right\rfloor - \frac{1}{b}\right)\right]^{1/p}
       \leq \sens_p(\hat{R}_b, \boundedvecs{D, k})
       \leq D\left[\left\lceil \frac{k}{D} \right\rceil\left(\left\lceil \log_b\left(\frac{DT}{k}\right) \right\rceil + 2\right)\right]^{1/p}\,.
    \end{equation*}
\end{theorem}
\begin{proof}
    For the upper bound, we begin by arguing that the upper bound
    \begin{equation*}
        \sens_p(\hat{R}_b, \boundedvecs{1, k}) \leq U(k) = [k(\lceil \log_b(T/k) + 2)]^{1/p}
    \end{equation*}
    in \Cref{cor:sens-tree-bounds-adv-all-t} is a concave, non-decreasing function in $k$ on $[1, T]$.
    As $g(x) = x^{1/p}$ for $p\geq 1$ is concave and non-decreasing, it suffices to argue for $f(x) = x(\lceil \log_b(x/k) \rceil + 2)$ being concave to show concaveness for $U = g\circ f$.
    We show this via the non-increasing forward-difference $\Delta f(x) = f(x+1) - f(x)$:
    \begin{align*}
        \Delta f(x)
        &\leq (x+1)\left(\left\lceil \log_b \left(T/x\right)\right\rceil + 2\right) - x\left(\left\lceil \log_b\left(T/x\right)\right\rceil\ + 2\right)
        \leq \lceil \log_b T/x\rceil + 2\,,\\
        \Delta f(x)
        &\geq (x+1)\left(\left\lceil \log_b\left(\frac{T}{x+1}\right)\right\rceil + 2\right) - x\left(\left\lceil \log_b\left(\frac{T}{x+1}\right)\right\rceil + 2\right)
        = \left\lceil \log_b\left(\frac{T}{x+1}\right)\right\rceil + 2\,,\\
        \Delta^2 f(x) &= \Delta f(x+1) - \Delta f(x) \leq \left\lceil \log_b\left(\frac{T}{x+1}\right)\right\rceil + 2 - \left(\left\lceil \log_b\left(\frac{T}{x+1}\right)\right\rceil + 2\right)
        = 0\,.
    \end{align*}
    We can thus invoke \Cref{lem:alternating-to-general-upper}, yielding
    \begin{align*}
        \sens_p(\hat{R}_b, \boundedvecs{D, k})
        &\leq D\cdot U(\lceil k / D\rceil)
        \leq D\left[\left\lceil \frac{k}{D} \right\rceil\left(\left\lceil \log_b\left(\frac{T}{\lceil k/D \rceil}\right) \right\rceil + 2\right)\right]^{1/p}\,,
    \end{align*}
    where the theorem statement is the result after applying $T/\lceil k/D\rceil \leq DT/k$.
    For the lower bound, we immediately apply \Cref{lem:alternating-to-general-lower}, yielding
    \begin{align*}
        \sens_p(\hat{R}_b, \boundedvecs{D, k})
        &\geq D\cdot \sens_p(\hat{R}_b, \boundedvecs{1, \lfloor k/D\rfloor})
        \geq D\left[\left\lfloor \frac{k}{D} \right\rfloor\left(\left\lfloor \log_b\left(\frac{T}{\lfloor k/D \rfloor}\right) \right\rfloor - 1/b\right)\right]^{1/p}\,,
    \end{align*}
    where similarly the theorem statement comes from $T/\lfloor k/D \rfloor \geq DT/k$.
\end{proof}
A few remarks can be made about \Cref{thm:tight-sens-bound-trees}.
First, note that the bounds are \emph{tight} for the $\ell_1$ and $\ell_2$ sensitivities:
\begin{align*}
   \sens_1(\hat{R}_b, \boundedvecs{D, k}) &= \Theta\left(k\log\left( DT/k \right)\right)\,,\\
   \sens_2(\hat{R}_b, \boundedvecs{D, k}) &= \Theta\left(\sqrt{Dk\log\left( DT/k \right)}\right)\,.
\end{align*}
They also strictly generalize \Cref{cor:sens-tree-bounds-adv-all-t}, in the sense that they recover the bounds for $D=1$.

\subsection{Error Bounds}\label{sec:trees-error}
Having finished our treatment of sensitivity, we are finally ready to derive our error guarantees.
We give bounds for both $\rho$-zCDP and $\epsilon$-DP.
We settle for stating bounds for the case of $b$-ary trees with subtraction as it improves leading constants over using only addition.
Note that the conditions on $k/D$ can be removed at the cost of a messier statement.
\begin{theorem}[Trees for $\rho$-zCDP]\label{thm:leading-const-tree-approx}
    Let $b\geq 3$ be odd, $T\geq b^{2b}$, $1 \leq D \leq k \leq T/b$, and $\rho > 0$.
    Additionally assume $k/D = o(T)$ and $k/D = \omega(1)$.
    Then the $\rho$-zCDP $b$-ary tree mechanism with subtraction for counting on $\boundedvecs{D,k}$-streams producing $T$ outputs achieves the following asymptotic error bounds:
    \begin{align*}
       \maxse(\hat{L}_b, \hat{R}_b, \boundedvecs{D,k})
       &= \Theta\left(\sqrt{Dk\log(DT/k)\log(T)/\rho}\right)\,,\\
       \meanse(\hat{L}_b, \hat{R}_b, \boundedvecs{D,k})
       &= \Theta\left(\sqrt{Dk\log(DT/k)\log(T)/\rho}\right)\,.
    \end{align*}
    The following upper bounds also hold under the weaker restrictions $T\geq b$ and $k\leq T$:
   \begin{align*}
       \maxse(\hat{L}_b, \hat{R}_b, \boundedvecs{D,k})
       &\leq \left(\frac{\sqrt{b-1}}{\sqrt{2}\log(b)} + o(1)\right) \cdot \sqrt{\frac{Dk\log(DT/k)\log(T)}{2\rho}}\,,\\
       \meanse(\hat{L}_b, \hat{R}_b, \boundedvecs{D,k})
       &\leq \left(\frac{\sqrt{b(1-1/b^2)}}{2\log(b)} + o(1)\right)\cdot \sqrt{\frac{Dk\log(DT/k)\log(T)}{2\rho}}\,,
   \end{align*}
   where the leading constants are minimized for $b=5$ and $b=7$ respectively, attaining minima of $0.609$ and $0.466$ respectively.
   The $o(1)$ terms go to zero as $T\to\infty$ for constant $b$.
   It also runs in space $O(\log T)$ and time $O(T)$ for constant $b$.
\end{theorem}
\begin{proof}
    The asymptotic error bounds follow directly from \Cref{cor:adv-l-bounds} together with \Cref{thm:tight-sens-bound-trees}.
    For the finer upper bound, we begin by upper-bounding the sensitivity \Cref{thm:tight-sens-bound-trees}, but first note that we do not impose $k\leq T/b$ here as that is only used for the lower bound of that theorem.
    We have that
    \begin{align*}
        \sens_p(\hat{R}_b, \boundedvecs{D, k})
        \leq D\cdot\sqrt{\left\lceil \frac{k}{D} \right\rceil\left(\left\lceil \log_b\left(\frac{DT}{k}\right) \right\rceil + 2\right)}
        \leq (1+o(1))\cdot\sqrt{Dk\log_b\left(\frac{DT}{k}\right)}
    \end{align*}
    where the second inequality uses $\sqrt{\lceil x \rceil + a} \leq \sqrt{x + a + 1} \leq \sqrt{x} + \sqrt{a+1} = (1+o(1))\sqrt{x}$ for constant $a \geq 0$ and $x = \omega(1)$, twice.
    Once for $x = \log_b(DT/k)$ and $a=2$, and once for $x=k/D$, $a=0$, combined with $(1+o(1))^2 = 1 + o(1)$.
    For the norms on the left factor (\Cref{cor:adv-l-bounds}), we get through similar manipulation
    \begin{align*}
        \| \hat{L}_b \|_{2\to \infty} &\leq \sqrt{\frac{(b-1)\lceil\log_b(T)\rceil+2}{2}}
        \leq (1+o(1))\cdot\sqrt{\frac{(b-1)\log_b(T)}{2}}\,,\\
        \frac{\| \hat{L}_b\|_{F}}{\sqrt{T}} &\leq \frac{\sqrt{b(1-1/b^2)\lceil\log_b(T)\rceil+2(1+b^2+b^{-\lceil\log_b T\rceil})}}{2}
        \leq (1+o(1))\cdot\frac{\sqrt{b(1-1/b^2)\log_b(T)}}{2}\,.
    \end{align*}
    For the finer upper bounds in the theorem statement, multiply the corresponding norm upper bounds together, divide by $\sqrt{2\rho}$ and extract the leading constant.
    The leading constants are plotted against $b$ in \Cref{fig:zcdp-tree-constant-opt}.
    \begin{figure}[th]
        \centering
        \subcaptionbox{Leading constant for $\rho$-zCDP vs b.\label{fig:zcdp-tree-constant-opt}}{\includegraphics[width=0.49\linewidth]{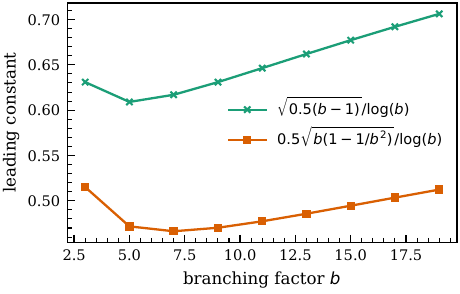}}%
        \hfill
        \subcaptionbox{Leading constant for $\varepsilon$-DP vs $b$.\label{fig:puredp-tree-constant-opt}}{\includegraphics[width=0.49\linewidth]{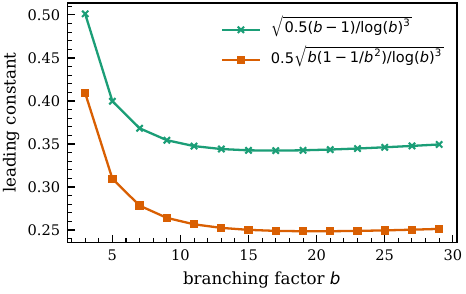}}%
        \hfill
        \caption{
            Leading constants for the $\rho$-zCDP (\Cref{thm:leading-const-tree-approx}) and $\epsilon$-DP (\Cref{thm:leading-const-tree-pure}) bounds on $\maxse(\hat{L}_b, \hat{R}_b, \cS_{1, k})$ and $\meanse(\hat{L}, \hat{R}_b, \cS_{1, k})$ plotted versus odd $b$.
            In \Cref{fig:zcdp-tree-constant-opt}, they are minimized for $b=5$ and $b=7$ respectively, attaining minima of $0.609$ and $0.466$.
            In \Cref{fig:puredp-tree-constant-opt}, they are minimized for $b=17$ and $b=19$ respectively, attaining minima of $0.342$ and $0.249$.\label{fig:tree-constant-opt}
        }
    \end{figure}
    The statements on time and space usage are proved in~\cite{AnderssonPST25}.
\end{proof}

To get a result for pure DP, we simply instead release $\hat{R}_b \xvec$ using the \emph{Laplace mechanism}, so the added noise has scale $\sqrt{2}\cdot\sens_1(\hat{R}_b, \boundedvecs{D, k}) / \varepsilon$, rather than $\sens_2(\hat{R}_b, \boundedvecs{D, k}) / \sqrt{2\rho}$.
We settle for only giving the finer upper bound in this case, but analogous tight asymptotics can be claimed for this case.
\begin{theorem}[Trees for $\epsilon$-DP]\label{thm:leading-const-tree-pure}
    Let $b\geq 3$ be odd, $T\geq b$, $1 \leq D \leq k \leq T$, and $\varepsilon > 0$.
    Additionally assume $k/D = o(T)$ and $k/D = \omega(1)$.
    Then the $\varepsilon$-DP $b$-ary tree mechanism with subtraction for counting on $\boundedvecs{D,k}$-streams producing $T$ outputs achieves the following error bounds:
   \begin{align*}
       \maxse(L, R, \boundedvecs{D, k}) &\leq \left(\sqrt{\frac{b-1}{2\log(b)^3}} + o(1) \right)\cdot \frac{\sqrt{2\log(T)}k\log(DT/k)}{\epsilon}\,,\\
       \meanse(L, R, \boundedvecs{D, k}) &\leq \left(\sqrt{\frac{b(1-1/b^2)}{4\log(b)^3}} + o(1)\right)\cdot \frac{\sqrt{2\log(T)} k\log(DT/k)}{\varepsilon}\,.
   \end{align*}
   where the leading constants are minimized for $b=17$ and $b=19$ respectively, with corresponding leading constants $0.342$ and $0.249$.
   The $o(1)$ terms go to zero as $T\to\infty$ for constant $b$.
   It also runs in space $O(\log T)$ and time $O(T)$ for constant $b$.
\end{theorem}
\begin{proof}
    The proof proceeds identically to that of \Cref{thm:leading-const-tree-approx}, except for the reliance on the $\ell_1$ sensitivity.
    We invoke \Cref{thm:tight-sens-bound-trees} (again, without enforcing $k\leq T/b$ as it only impact the lower bound of that theorem) for upper-bounding the sensitivity:
    \begin{align*}
        \sens_1(\hat{R}_b, \boundedvecs{D, k})
        \leq D\cdot\left\lceil \frac{k}{D} \right\rceil\left(\left\lceil \log_b\left(\frac{DT}{k}\right) \right\rceil + 2\right)
        \leq (1+o(1))\cdot k\log_b\left(\frac{DT}{k}\right)
    \end{align*}
    where the second inequality uses $\lceil x \rceil + a \leq x + a + 1 \leq (1+o(1))x$ for constant $a \geq 0$ and $x = \omega(1)$, twice.
    Once for $x = \log_b(DT/k)$ and $a=2$, and once for $x=k/D$, $a=0$, combined with $(1+o(1))^2 = 1 + o(1)$.
    Identically to the $\rho$-zCDP case, we bound the norms on the left factor via \Cref{cor:adv-l-bounds} to arrive at
    \begin{align*}
        \| \hat{L}_b \|_{2\to \infty}
        \leq (1+o(1))\cdot\sqrt{\frac{(b-1)\log_b(T)}{2}}\,,
        \qquad
        \frac{\| \hat{L}_b\|_{F}}{\sqrt{T}}
        \leq (1+o(1))\cdot\frac{\sqrt{b(1-1/b^2)\log_b(T)}}{2}\,.
    \end{align*}
    For the final theorem statement, multiply the corresponding upper bounds together, multiply by $\sqrt{2}/\varepsilon$, and extract the leading constant to get the final expression.
    The statements on time and space usage are, again, proved in~\cite{AnderssonPST25}.
    The leading constants are also again convex functions in $b$, and are plotted in \Cref{fig:puredp-tree-constant-opt}.
\end{proof}
In both cases, the optimization of leading constant is independent of $k$ and  $D$.
For $\rho$-zCDP, tree-based factorizations achieve an error scaling as $O\left(\sqrt{Dk\log(DT/k)\log(T)}\right)$. We compare this to the asymptotics of other factorizations we consider in \Cref{sec:comparison}.

\section{A Naive Factorization}\label{sec:naive}
Given that tree-based factorization can yield improved bounds for large $k$, a natural question is if tree-based factorization is the correct construction in this regime.
We show that the answer is subtle: for large enough $k$, a naive factorization $A=I\times A$ outperforms it.
\begin{lemma}\label{lem:prefix-matrix-sens}
    Let $1 \leq D \leq k \leq T$, $p\geq 1$ and $A\in\mathbb{R}^{T\times T}$ be the lower-triangular matrix of all-ones.
    Then
    \begin{equation*}
        \sens_p(A, \boundedvecs{D,k})  = D \cdot T^{1/p}\,.
    \end{equation*}
\end{lemma}
\begin{proof}
    Note that for any $\dvec\in\boundedvecs{D,k}$, we have that $A\dvec\in[-D, D]^{T}$ by definition.
    As the vector $\dvec'$ with $\dvec'[0] = D$ and zero in every other entry always exists in $\boundedvecs{D,k}$, we immediately get $\sens_p(A, \boundedvecs{D,k}) = \|A\dvec' \|_p = D T^{1/p}$.
\end{proof}
\begin{theorem}\label{thm:naive}
    For $1 \leq D \leq k \leq T$, and $\rho > 0$, consider the $\rho$-zCDP factorization mechanism for counting on $\boundedvecs{D,k}$-streams with $L=I$ and $R=A$.
    It achieves error
    \begin{equation*}
        \maxse(L,R, \boundedvecs{D,k})
        = \meanse(L, R, \boundedvecs{D,k})
        = \frac{D\sqrt{T}}{\sqrt{2\rho}}\,.
    \end{equation*}
    Analogously, given $\epsilon > 0$, the $\epsilon$-DP factorization mechanism (releasing $R\xvec$ with the Laplace mechanism), attains the following error:
    \begin{equation*}
        \maxse(L,R,\cS_{k, D}) = \meanse(L, R, \cS_{k, D}) = \frac{\sqrt{2}DT}{\epsilon}\,,
    \end{equation*}
    Both versions run in time $O(T)$ and $O(1)$ space.
\end{theorem}
\begin{proof}
    We trivially have that $\| I \|_{2\to\infty} = \| I \|_F / \sqrt{T} = 1$, so the error statement reduces to invoking \Cref{lem:prefix-matrix-sens} for the $\ell_2$ and $\ell_1$ sensitivity.
    For the statement on time and space usage: the mechanism adds fresh, independent noise at each step, which can be done in constant time per step, and does not require storing noise across time steps.
\end{proof}
Focusing on the $\rho$-zCDP setting, the moral is clear.
An error of $\Theta(\min(D, k)\sqrt{T})$ is always possible by essentially using the Gaussian mechanism to release each prefix individually, and paying for composition. We compare this to the asymptotics of error bounds for other factorizations in \Cref{sec:comparison}.
\section{Analytical Comparison of Bounds}\label{sec:comparison}

Here, we analytically compare the bounds in our paper and discuss the results. We focus on the results with respect to the sensitivity vector set $\boundedvecs{1,k}$ (corresponding to the results for $\degreecount$ and $\countdistinct$).

 In this analytical comparison, we will use the dominating term in each of the upper bounds corresponding to the square root factorization and the tree-based mechanism we describe in this paper, and the exact bound for the naive mechanism that adds Gaussian noise at each time step and pays for composition. As in the introduction, we will consider $\rho = 1/2$ for our comparison (since the dependence on $\rho$ is identical in all our bounds, it does not affect our comparison). 

The corresponding bounds are $c_1 \sqrt{k} \log T$ (for the square root factorization mechanism), where $c_1$ is a constant that can be derived from the upper bound in \Cref{cor:toep-error}, $c_2 \sqrt{k \log(T/k) \log T}$ (for tree-based factorization mechanisms) where $c_2$ can be derived from the bounds in  \Cref{thm:leading-const-tree-approx} for the optimal choice of $b$, and $\sqrt{T}$ which is the exact bound for the naive mechanism (see \Cref{thm:naive}). Note that $c_2 > c_1$.

Comparing these bounds, we can make the following conclusions. Note that the square-root mechanism dominates over the naive bound when $k \leq \frac{T}{c_1^2 \log^2 T}$. Similarly, the upper bound for the square root factorization dominates over the upper bound for the tree-based mechanism when $k \leq T^{1-(c_1/c_2)^2}$ and otherwise the upper bound for the tree-based mechanism dominates. Finally, observe that the upper bound for the tree-based mechanism dominates over the bound for the naive mechanism as long as $k \log(T/k) \leq \frac{T}{c_2^2 \log T}$. Let the value of $k$ at which equality is achieved be $k^*$.

Note that these bounds imply that for $k \in [1,T^{1-(c_1/c_2)^2}]$, the upper bound based on the square root factorization dominates, for $k \in [T^{1-(c_1/c_2)^2}, k^*]$ (where $k^* > \frac{T}{\log^2 T}$),  we have that the upper bound based on the tree-based mechanism dominates, and for $k \in [k^*, T]$, the bound based on the naive composition mechanism dominates. We expect $k$ to be much smaller than $T$ in most practical applications, and hence our bounds give improvements for most practically significant values of $k$. We leave an exact characterization of the optimal bounds for different $k$ to future work.




\section{Conclusion}\label{sec:conclusion}
Our work gives improved error bounds for a range of cardinality estimation problems. We do so by recasting private cardinality estimation on fully dynamic streams as private linear queries over their difference streams and studying the associated sensitivity vector sets in detail. 
There are many possible directions for future work.

\textbf{Triangle Counting with Privacy for All Graphs.}
As already mentioned, the privacy afforded by reducing triangle counting to continual counting on the difference stream is conditioned on the graph stream having degree bounded by $D$ and maximum triangle contribution bounded by $k_{\mathrm{tri}}$. Coming up with an algorithm that guarantees privacy for all graphs, but continues to give accuracy matching our approach for streams satisfying the parameter constraints, is a natural open question.


\textbf{Other Linear Queries.}
While this paper exclusively deals with continual counting applied to difference streams, alternative linear queries on these streams might be interesting to explore. For example, instead of computing the exact number of distinct elements, one might be interested in just knowing whether the number of distinct elements \emph{recently has been increasing or decreasing}.
Such a statistic could naturally be derived from postprocessing the running count, but it might be more efficient to directly adapt the query matrix to the statistic in mind. Studying such query matrices via our framework is a natural direction for future work.


\textbf{Space Efficiency.}
There is evidence that the square-root factorization we employ does not admit any sublinear space algorithm~\cite{AnderssonY2025}.
Recent work~\cite{DvijothamMPST24} introduced a Toeplitz factorization approximating its error up to an additive constant in space $O(\log^2 T)$, but their construction does not enforce a monotonicity condition on the diagonals of the right factor.
Designing a space-efficient Toeplitz factorization with such a monotonicity constraint and near-optimal error for standard continual counting remains an open question, and would give error bounds for continual counting on $\boundedvecs{D, k}$-streams via the results in \Cref{sec:toep}.

\textbf{Lower Bounds.} Finally we raise the question of establishing general lower bounds for continual counting on $\boundedvecs{D,k}$-streams, and for factorization mechanisms in particular.
We already know that for $\boundedvecs{1,1}$, the square-root factorization is leading constant-optimal for our error metrics~\cite{matouvsek2020factorization,fichtenberger_constant_2023,HenzingerUU23}, but it is clear that it is not optimal for $\boundedvecs{1, T}$ (see~\Cref{thm:naive}).
As such, different factorizations are expected to be optimal for different parameter regimes.

\section*{Acknowledgments}

We thank Nikita Kalinin and Adam Smith for valuable comments on an early draft of this paper, and Sofya Raskhodnikova for valuable discussions on the results of \cite{RaskhodnikovaS24}.

Joel Daniel Andersson was funded by the European Union.
In particular, he received funding from the European Research Council (ERC) under the European Union's Horizon 2020 research and innovation programme (MoDynStruct, No. 101019564)  \includegraphics[width=0.9cm]{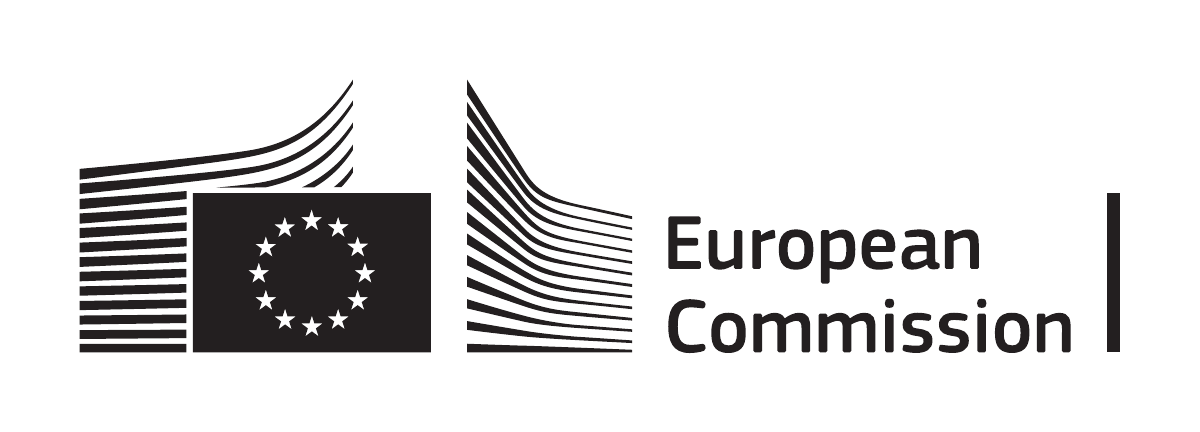}, with additional funding from Providentia, a Data Science Distinguished Investigator grant from Novo Nordisk Fonden.
Part of the work was carried out while visiting Boston University.



Views and opinions expressed are those of the authors only and do not necessarily reflect those of any sponsor.
This includes the European Union and the European Research Council Executive Agency.

\printbibliography






\end{document}